\documentclass[aps, prx, reprint, superscriptaddress,
notitlepage, nofootinbib, longbibliography,
floatfix]{revtex4-2}

\usepackage{natbib}
\usepackage{graphicx}
\usepackage{mathtools}
\usepackage{amsfonts} 
\usepackage{amsmath} 
\usepackage{amssymb}
\usepackage{amsthm}
\usepackage{bm}
\usepackage{braket}
\usepackage{color} 
\usepackage{bbm}
\usepackage{subcaption}
\usepackage{import}
\usepackage{microtype}
\usepackage{relsize}
\usepackage{framed}
\usepackage[dvipsnames]{xcolor}
\usepackage[percent]{overpic}

\usepackage{tikz}
\usetikzlibrary{calc}
\usetikzlibrary{arrows.meta} 
\usetikzlibrary{decorations.pathreplacing}

\usepackage{tikzit}

\tikzstyle{new edge style 0}=[-, line width=0.25 mm]

\usepackage{hyperref}

\newcommand{\CZ}{U_{\rm CZ}}
\newcommand{\sys}{\mathcal{S}}
\newcommand{\anc}{\mathcal{A}}
\newcommand{\Xe}{X_\text{even}}
\newcommand{\Xo}{X_\text{odd}}
\newcommand{\E}{\mathcal{E}}
\newcommand{\trivialstate}{\ket{\psi_\text{Trivial}}} 

\usepackage{widebar}
\newcommand{\comp}[1]{\widebar{#1}}

\newcommand{\rn}[2]{
    \tikz[remember picture,baseline=(#1.base)]\node [inner sep=0] (#1) {$#2$};%
}

\newcommand{\norm}[1]{\left\lVert#1\right\rVert}
\newcommand{\defeq}{\vcentcolon=}
\newcommand{\eqdef}{=\vcentcolon}

\newcommand{\ketbra}[2]{\ket{#1}\!\bra{#2}}

\newcommand{\one}{\mathbbm{1}}

\newcommand{\Z}{\mathbb{Z}}

\newcommand{\n}{\mathbf{n}}
\newcommand{\Hilb}{\mathcal{H}}

\DeclareMathOperator{\Tr}{Tr}

\newtheorem{theorem}{Theorem}
\newtheorem{conjecture}{Conjecture}
\newtheorem{corollary}{Corollary}
\newtheorem{lemma}{Lemma}

\theoremstyle{definition}
\newtheorem{definition}{Definition}

\begin{document}

\title{Mixed-state quantum anomaly and multipartite entanglement}

\author{Leonardo A. Lessa}
\email{llessa@pitp.ca}
\affiliation{Perimeter Institute for Theoretical Physics, Waterloo, Ontario N2L 2Y5, Canada}
\affiliation{Department of Physics and Astronomy, University of Waterloo, Waterloo, Ontario N2L 3G1, Canada}

\author{Meng Cheng}
\affiliation{Department of Physics, Yale University, New Haven, Connecticut 06511-8499, USA}

\author{Chong Wang}
\affiliation{Perimeter Institute for Theoretical Physics, Waterloo, Ontario N2L 2Y5, Canada}

\begin{abstract}
    Quantum entanglement measures of many-body states have been increasingly useful to characterize phases of matter. Here we explore a surprising connection between mixed state entanglement and 't Hooft anomaly. More specifically, we consider lattice systems in $d$ space dimensions with anomalous symmetry $G$ where the anomaly is characterized by an invariant in the group cohomology $H^{d+2}(G,U(1))$. We show that any mixed state $\rho$ that is strongly symmetric under $G$, in the sense that $G\rho\propto\rho$, is necessarily $(d+2)$-nonseparable, i.e. is not the mixture of tensor products of $d+2$ states in the Hilbert space. Furthermore, such states cannot be prepared from any $(d+2)$-separable states using finite-depth local quantum channels, so the nonseparability is long-ranged in nature. We provide proof of these results in $d\leq1$, and plausibility arguments in $d>1$.  The anomaly-nonseparability connection thus allows us to generate simple examples of mixed states with nontrivial long-ranged multipartite entanglement. In particular, in $d=1$ we found an example of \textit{intrinsically mixed} quantum phase, in the sense that states in this phase cannot be two-way connected to any pure state through finite-depth local quantum channels. We also analyze mixed anomaly involving both strong and weak symmetries, including systems constrained by the Lieb-Schultz-Mattis type of anomaly. We find that, while strong-weak mixed anomaly in general does not constrain quantum entanglement, it does constrain long-range correlations of mixed states in nontrivial ways. Namely, such states are not symmetrically invertible and not gapped Markovian, generalizing familiar properties of anomalous pure states.
\end{abstract}

\maketitle

{
    \hypersetup{linkcolor=black}
    \tableofcontents
}

\section{Introduction}
\label{sec:intro}

A fundamental characterization of global symmetry in a local quantum system is its associated 't Hooft anomaly~\cite{AdlerAnomaly,BJAnomaly,Hooft1980}. The defining feature of a nontrivial anomaly is that it forbids a symmetric and completely trivial state in the system, thus any state preserving the anomalous symmetry must be ``nontrivial". More precisely, it is long-range entangled~\cite{ChenLRE}. Constraining the behaviors of quantum systems using anomaly (also known as ``anomaly-matching'') has become an extremely powerful tool in the theoretical study of strongly interacting quantum systems in both condensed matter and high energy physics.

Heuristically, the anomaly describes intrinsic obstructions to ``localize" the symmetry action. In particular, an anomalous symmetry action cannot be ``on-site". A common physical setup for anomalous symmetry is the boundary states of nontrivial symmetry-protected topological (SPT) phases~\cite{chen_symmetryprotected_2012,chen2013}. They are short-range entangled (SRE) but cannot be adiabatically connected to a trivial symmetric product state without breaking the protecting symmetry. Due to their topological nature, the symmetry has to act anomalously on the spatial boundary of an SPT state, leading to a variety of interesting topological phenomena. It is, however, worth noting that anomalies are not limited to the boundary of SPT states. A large class of ``mixed anomalies" between internal and spatial symmetries have been identified in lattice systems, which generalizes the celebrated Lieb-Schultz-Mattis(-Oshikawa-Hastings) theorem in spin models~\cite{Lieb1961,Oshikawa1999,Hastings2003,Cheng2015,Po2017,Ye2021,Cheng2023}. Many examples of lattice models realizing anomalous internal symmetries have also been constructed and studied~\cite{CZX, LevinGu, yoshida_topological_2016}.

While the implications of quantum anomaly have been extensively studied for ground states and low-energy dynamics, real-world systems, however, are best described by mixed states due to the inevitable presence of noise, which correlates the system of interest with inaccessible external degrees-of-freedom. One natural question is whether we can generalize the notion of SPT phases and 't Hooft anomaly to mixed state setting. The main objective of this work is to formulate a set of (information-theoretic) constraints of 't Hooft anomaly on mixed states, in terms of spatial separability and symmetric invertibility. As a byproduct, we also identify examples of nontrivial mixed-state phases of matter, even in (1+1)-D.

To discuss this problem, it is crucial to first define appropriate generalizations of the notion of global symmetry to open quantum systems. Depending on the couplings to the environment, there can be two types of symmetries~\cite{buca_note_2012, deGroot2022, ma_average_2023, LeeYouXu2022, ZhangQiBi2022, Ma_intrisic_2023}: when the couplings preserve the symmetry of the system, the symmetry is said to be \emph{strong} (or \emph{exact}). In other words, the couplings do not transfer symmetry charges to the environment -- in the language of statistical mechanics the strong symmetry corresponds to a canonical (instead of a grand-canonical) ensemble. If instead only the total charge of the system and environment is conserved, the symmetry becomes \emph{weak} (or \emph{average}). Formally, a mixed state $\rho$ is strongly symmetric under a unitary $U$ if $U\rho=e^{i\theta}\rho$ for some phase factor $e^{i\theta}$, and is weakly symmetric if $U\rho U^{-1}=\rho$. In both cases, $\rho$ can be decomposed into eigenstates of $U$, but the strong symmetry condition further requires that all states share the same eigenvalue, i.e. the same charge.

From the definition, it is quite clear that weak symmetry has very little constraining power on the entanglement structure of states: the maximally mixed state $\one/\dim(\Hilb)$ is always weakly symmetric under any symmetry, regardless of whether the symmetry is anomalous or not. This is consistent with the classification of mixed state SPT phases recently established in Refs.~\cite{deGroot2022,ma_average_2023,Ma_intrisic_2023}, where the classification is trivial for weak symmetry alone. However, as we will see, the case involving strong symmetry is fundamentally different.

Here, we primarily study the universal features of mixed states $\rho$ that are strongly symmetric under an anomalous symmetry. We focus on systems made of bosons or spins (so no physical fermions), and consider 't Hooft anomalies of the group-cohomology type~\cite{chen_symmetryprotected_2012}. We prove that in $1d$, such states cannot be tripartite separable. In other words, for any spatial tripartition of the system into connected regions $A\cup B\cup C$ without boundary, a state strongly symmetric under an anomalous symmetry cannot be written as a convex sum of states of the form $\rho_{A}\otimes \rho_{B}\otimes\rho_{C}$. Furthermore, we prove that such tripartite entanglement is long-ranged: a symmetric state $\rho$ cannot be prepared from a tripartite separable state $\rho_0$ using a finite-depth local quantum channel. We further conjecture that such state in $d$ dimensions is not $(d+2)$-separable, and establish it for a special class of symmetry transformations. We call this result the \emph{anomaly-nonseparability connection}. It constrains the multipartite entanglement structure of anomalous many-body states in an inherent way that may not be captured by the usual bipartite entanglement measures, such as the von Neumann entropy or negativity. Furthermore, it is non-perturbative, and is valid even for highly entangled (e.g. volume-law) or highly mixed (e.g. infinite temperature) states.

Mixed anomalies involving both strong and weak symmetries appear to be more subtle, and will be discussed briefly in Secs.~\ref{sec:mixed} and \ref{sec:LSM}. As we will see, mixed strong-weak anomaly in general does not preclude separable states, unlike the strong-strong case considered in prior sections. Instead, such mixed anomaly constrains the long-range correlation in some nontrivial ways: a state with mixed strong-weak anomaly is necessarily (a) not \emph{symmetrically invertible} (Sec.~\ref{sec:sym-noninv}), and (b) not \emph{gapped Markovian} (Sec.~\ref{sec:sw-cmi}). Similar statements are in fact defining features of standard pure state SPT edge states. Curiously, non-Abelian continuous symmetries appear to allow for stronger results. In particular, for a spin-$1/2$ chain with strong $SO(3)$ and weak lattice translation symmetries (which features the celebrated Lieb-Schultz-Mattis anomaly), a symmetric state cannot be bipartite separable, or be prepared from a bipartite separable state using a finite-depth local quantum channel.

The rest of this paper is structured as follows. We start in section \ref{sec:preliminaries} by reviewing the three main concepts used here: partial separability (Sec. \ref{sec:intro-partial_separability}), strong and weak symmetries (Sec. \ref{sec:intro-strong_and_weak_symmetries}), and anomalies (Sec. \ref{sec:intro-anomaly}). Before proving the anomaly-nonseparability connection for generic strongly symmetric anomalous systems, we illustrate it with two examples: in Sec. \ref{sec:projective_rep} we review how the long-range entanglement between the edge modes of a (1+1)-D SPT system result from the projectiveness of the boundary symmetry; then in Sec. \ref{sec:4-qubit_CZX} we provide an elementary proof that, in a $1d$ qubit chain, states strongly symmetric under the anomalous symmetry of the $\Z_2$ CZX model are not tripartite separable. We then provide the general proof of the long-ranged anomaly-nonseparability connection in (1+1)-D in Secs.~\ref{sec:tripartite_entanglement} and \ref{sec:long-range_entanglement}. Using such proof, we argue in Sec.~\ref{sec:intrinsicallymixed} that the maximally mixed CZX-symmetric state is part of an \emph{intrinsically mixed} phase of matter. In Sec.~\ref{sec:higher_dims} we extend our statement to higher dimensions in the form of a conjecture. In Sec.~\ref{sec:mixed} we discuss a weaker result for mixed anomalies involving both strong and weak symmetries, in terms of non-invertibility (Sec.~\ref{sec:sym-noninv}) and Markovian gaplessness (Sec.~\ref{sec:sw-cmi}), with the main tool being the obstruction to symmetry localization due to the strong-weak anomaly (Sec.~\ref{sec:obstructlocalization}). In Sec.~\ref{sec:LSM} we discuss a special type of mixed anomaly, namely a spin-$1/2$ chain with strong $SO(3)$ (or $O(2)$) and weak translation symmetries, and show that a symmetric state in such systems is necessarily long-range bipartite entangled. We finish with a brief exposition to the relationship with other works in Sec.~\ref{sec:relation} and with other discussions in Sec.~\ref{sec:Discussion}. Several Appendices contain peripheral details.

\section{Preliminaries}
\label{sec:preliminaries}
\subsection{Partial separability}
\label{sec:intro-partial_separability}

Multipartite quantum systems have more complex entanglement structure than their bipartite counterparts~\cite{horodecki_multipartite_2024}. Famously, there are two inequivalent ways three qubits can be genuinely entangled under stochastic local operations and classical communication~\cite{dur_three_2000}. For more than three parties, there is a continuum of such classes, and a complete characterization of their entanglement is much harder~\cite{dur_three_2000, verstraete_four_2002}. Here, we review the characterization of the entanglement of multipartite systems by their partial separability: whether or not they can be split into $k$ disentangled parts. This gives a discrete gradation of separability between fully separable states and genuinely bipartite entangled states that will be the key to characterize anomalous mixed states.

First, some basic definitions. Given a Hilbert space $\Hilb$, we denote by $Q(\Hilb)$ the set of quantum states of $\Hilb$, i.e. of trace-one positive semi-definite linear operators $\rho : \Hilb \to \Hilb$.  We are mostly interested in lattice models, where the Hilbert space $\Hilb$ of the whole system is the tensor product of Hilbert spaces from each ``site":
\begin{equation}
    \Hilb = \Hilb_1 \otimes \cdots \otimes \Hilb_N,
    \label{tensor-product-hilb}
\end{equation}
where $N$ is the total number of sites. We will say that the system is $N$-partite. In this work, $\Hilb_n$ are finite-dimensional for all $n$, and for simplicity we assume they all have the same dimension.

Denote by $I=\{1,2,\ldots, N\}$ the set of site indices. A collection of pair-wise disjoint non-empty subsets $\{A_i\}_{i=1}^k \subset 2^I$ is a $k$-partition of $I$ if $\cup_{i=1}^k A_i=I$. Each subset $A_i$ is associated with a subsystem Hilbert space  $\Hilb_{A_i} = \bigotimes_{l \in A_i} \Hilb_l$. We are mostly interested in geometrically local partitions, namely dividing the system into several disjoint regions of space.

We now state the definition of partial separability with respect to a given partition as found in the literature~\cite{horodecki_quantum_2009, horodecki_multipartite_2024}:
\begin{definition}[Pure state partial separability]
    A pure state $\ket{\psi} \in \Hilb$ is \emph{separable with respect to a partition} $\{A_i\}_{i=1}^k$ if it is a tensor product state $\ket{\psi} = \ket{\psi_1} \otimes \ket{\psi_2} \otimes \cdots \otimes \ket{\psi_k}$, where $\ket{\psi_i} \in \Hilb_{A_i}$.
\end{definition}

\begin{table}[t]
    \centering
    \begin{tabular}{|c|c|c|}
        \hline
        State & 2-separable? (partition) & 3-separable? \\
        \hline
        $\ket{000}$          & Yes (all) & Yes \\
        $\ket{\Psi_{00}}\otimes\ket{0}$ & Yes ($12|3$) & No \\
        GHZ, W                   & No (all)  & No \\
        SHIFTS                   & Yes (all) & No \\
        $T=\infty$ CZX           & Yes (all) & No \\
        \hline
    \end{tabular}
    \caption{Examples of multipartite qubit states and whether they are $k$-separable. $\ket{\Psi_{00}} = \frac{1}{\sqrt{2}}(\ket{00} + \ket{11})$ is the Bell state, SHIFTS stands for the mixed state defined in \cite{bennett_unextendible_1999} (See main text), and the $T=\infty$ CZX states are defined in Eq.\eqref{eq:infinite-temperature_CZX}.}
    \label{tab:examples_separability}
\end{table}

We might not want to choose a specific partition, but only the number of sets $k$. In that case, we simply say $\ket{\psi}$ is \emph{$k$-separable}. The negation of $k$-separable is \emph{$k$-nonseparable} or \emph{$k$-entangled}. We call an $N$-separable state \emph{fully separable}, which is sometimes simplified to just ``separable'', with the negation being just \emph{entangled}. On the opposite direction, $2$-nonseparable states are sometimes called \emph{genuinely entangled}. See Table \ref{tab:examples_separability} for examples.

The definition of $k$-separable pure states can be naturally generalized to mixed states by taking its convex hull:
\begin{definition}[Mixed state partial separability]
\label{def:msep}
    A mixed state $\rho \in Q(\Hilb)$ is \emph{separable with respect to a partition} $\mathcal{P} = \{A_i\}_{i=1}^k$ if it is the mixture of pure states separable under the same partition $\mathcal{P}$:
    \begin{equation}
        \rho = \sum_{\alpha} p_\alpha \ketbra{\psi^{(\alpha)}_1}{\psi^{(\alpha)}_1} \otimes \cdots \otimes \ketbra{\psi^{(\alpha)}_k}{\psi^{(\alpha)}_k},
    \end{equation}
    with $\ket{\psi^{(\alpha)}_i} \in \Hilb_{A_i} = \bigotimes_{l \in A_i} \Hilb_l$ and $\sum_\alpha p_\alpha = 1$.
\end{definition}

{Notably, deciding whether a generic mixed state $\rho$ is separable is a computationally hard problem~\cite{gurvits_classical_2003,ioannou_computational_2007,gharibian_strong_2010,gutoski_quantum_2015} as we need to find the optimal decomposition $\rho=\sum_{\alpha}p_{\alpha}|\Psi^{(\alpha)}\rangle\langle\Psi^{(\alpha)}|$ satisfying the separability condition. This layer of complexity makes studying entanglement much more challenging for mixed states compared to pure states.}

One can also define separability in a broader sense~\cite{horodecki_quantum_2009, horodecki_multipartite_2024}, called just $k$-separable, by allowing different pure states $|\Psi^{(\alpha)}\rangle$ in a decomposition $\rho=\sum_{\alpha}p_{\alpha}|\Psi^{(\alpha)}\rangle\langle\Psi^{(\alpha)}|$ to be separable with respect to different $k$-partitions $\{A_i^{(\alpha)}\}_{i=1}^k$. We will focus on the narrower definition in Def.~\ref{def:msep} in this work for simplicity, but we note that some of our key results such as the anomaly-nonseparability connection (Theorem~\ref{thm:3-sep} and Conjecture~\ref{conj:higherdim}) apply equally well for the broader definition.

\begin{figure}[t]
    \centering
    \includegraphics[width=0.8\linewidth]{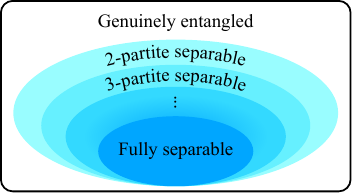}
    \caption{Inclusion relations of multipartite states with respect to separability: any $k$-separable with $k > l$ is also $l$-separable, and genuinely entangled states are not 2-partite separable.}
    \label{fig:nested_separability}
\end{figure}

If a state is $k$-separable, then it is $l$-separable for every $l < k$ (See Fig.~\ref{fig:nested_separability}). The converse does not hold: the state $\ket{\Psi_{00}} \otimes \ket{0}$, where $\ket{\Psi_{00}} = \frac{1}{\sqrt{2}} (\ket{00} + \ket{11})$ is the Bell state, is bipartite separable under the $12|3$ partition, but it is not tripartite separable. If, however, an $n$-partite pure state $\ket{\psi}$ is bipartite separable under \emph{all} partitions, then it is easy to check that it is also fully separable.\footnote{By taking a bipartition $A|B$ with $A$ being one site $i$ and $B$ its complement, then we find $\ket{\psi} = \ket{\psi_i} \otimes \ket{\tilde{\psi}_B}$. By repeating the same procedure for all other sites $i$, we conclude $\ket{\psi} = \bigotimes_i \ket{\psi_i}$.} Surprisingly, the same does not hold for mixed states. Indeed, Bennett \emph{et al.}~\cite{bennett_unextendible_1999} constructed a three-qubit state that is biseparable under any bipartition, but nevertheless entangled, refuting a conjecture made in \cite{brassard_multiparticle_1999}. Such phenomenon is possible because the basis used to decompose $\rho$ into separable states for one bipartition may be very different from the basis used for another bipartition, so we cannot simply use a unified basis to separate the state for a tripartition. The existence of such bipartite separable, but multipartite entangled states also means that it is not sufficient to just focus on the familiar bipartite entanglement measures (such as negativity or entanglement of formation). Indeed, many multipartite entanglement measures and tests have been proposed to (partially) separate $k$-separable states with different $k$~\cite{guhne_entanglement_2009, brandao_separable_2004, horodecki_separability_2006, huber_detection_2010, gabriel_criterion_2010, gao_separability_2011, huber_entropy_2013, liu_separability_2015, hong_detecting_2016, sarbicki_family_2020, hong_detection_2021a, li_parameterized_2023b}.

Even though our motivation for the current work comes from the many-body physics of anomalous symmetries, some of our results are of interest to multipartite entanglement theory \textit{per se}. In particular, in Sec.\ref{sec:czx-bipartite-separability} we present a family of many-body states that are bipartite separable but tripartite entangled, and in Appendix \ref{appendix:measure_multip_entanglement} we study a ``good'' multipartite entanglement measure of mixed states presented in \cite{szalayMultipartiteEntanglementMeasures2015} that faithfully distinguishes $k$-partite separable from entangled states.

\subsection{Strong and weak symmetries}
\label{sec:intro-strong_and_weak_symmetries}

To generalize symmetry anomaly to mixed systems, we first have to define what it means for a mixed state to be symmetric. Following previous works~\cite{buca_note_2012, deGroot2022, ma_average_2023, LeeYouXu2022, ZhangQiBi2022, Ma_intrisic_2023}, we contrast two definitions. The first definition is of \emph{average}, or \emph{weak} symmetry:
\begin{definition}[Weakly symmetric states]
    A state $\rho$ is \emph{weakly symmetric} under an (anti)unitary operator $U$ if $U \rho U^\dagger = \rho$.
\end{definition}
The symmetry $U$ can be antiunitary, but we will only treat the unitary case in this paper.

The definition of weak symmetry comes naturally from the symmetry action of $U$ on pure states, $\ket{\psi} \mapsto U \ket{\psi}$. Furthermore, there is an equivalent condition on purifications $\ket{\Psi} \in \Hilb \otimes \Hilb_\anc$: $\rho$ is weakly symmetric under $U$ if, and only if, its purification $\ket{\Psi}$ is symmetric under $U \otimes U_\anc$, for some unitary $U_\anc$ acting only on the ancilla space $\Hilb_\anc$. Moreover, we can choose $U_\anc$ to be (unitarily equivalent to) the complex conjugate of $U$.  Indeed, one purification of $\rho$ is $\ket{\Psi} = \sqrt{\rho} \otimes \one \ket{\Psi_{00}}$, where $\ket{\Psi_{00}} = \sum_i \ket{i i} \in \Hilb \otimes \Hilb$ is maximally entangled, satisfying $U \otimes U^* \ket{\Psi} = \ket{\Psi}$. One physical interpretation is that $\rho$ may come from an interaction with an environment that is symmetric under the extended symmetry $U \otimes U_\anc$.

All previous statements generalize accordingly to group representations $U : G \to U(\Hilb)$. In such case, another equivalent description comes from Schur's lemma, which implies $\rho$ is block-diagonal $\rho = \oplus_\alpha \rho^{(\alpha)}$, where each block matrix $\rho^{(\alpha)}$ acts on the multiplicity space of the irreducible representation (irrep) $\alpha$ of $U$ as a multiple of identity between each irrep copy~\cite{georgi_lie_2018}.

Of importance to us is the subset of weakly symmetric states for which the ancilla symmetry action $U_\anc$ can be made trivial, i.e. $U_\anc \propto \one_\anc$. If $\rho$ comes from the system interacting with its environment, this condition prohibits charge exchange between the symmetry and its environment. In such case, we say $\rho$ is \emph{strongly symmetric}:
\begin{definition}[Strongly symmetric state]
    A state $\rho$ is \emph{strongly symmetric} under a unitary operator $U$ if $U \rho = \lambda \rho$, for some phase $\lambda \in U(1)$ that we call the \emph{charge} of $\rho$.
\end{definition}

In Sec. \ref{sec:tripartite_entanglement}, we prove that any tripartite separable mixed state $\sum_i p_i \ket{A_i}\ket{B_i}\ket{C_i}$ is not \emph{strongly} symmetric under an anomalous (1+1)-D symmetry. The reason why we restrict this result to strongly symmetric states, instead of weakly symmetric or to another subset entirely, is twofold. Firstly, the weak symmetry condition is too weak, since even the maximally mixed state is weakly symmetric under any symmetry: $U\one U^{\dagger} \equiv \one$. Secondly, and most importantly, the following lemma about strongly symmetric states enables us to leverage the anomaly-nonseparability proof from pure states to mixed states:
\begin{lemma}\label{lemma:strongly_symmetric_decomposition}
    All (unnormalized) vectors $\ket{\tilde\psi_i}$ in any decomposition of $\rho = \sum_i \ketbra{\tilde\psi_i}{\tilde\psi_i}$ are symmetric with equal charge if, and only if, $\rho$ is strongly symmetric with the same charge.
\end{lemma}
\begin{proof}
    The ``only if'' statement is trivial. For the ``if'' part, note that the strong symmetry condition $U \rho = \lambda \rho$ implies any state in the image of $\rho$ is also symmetric with the same charge. Since $\operatorname{Ker}(\rho)^\perp = \operatorname{Im}(\rho)$, we prove that each $\ket{\tilde\psi_i}$ is orthogonal to all vectors in the kernel of $\rho$. Indeed, for any $\ket{\tilde\phi} \in \operatorname{Ker}(\rho)$, we have $0 = \braket{\tilde\phi | \rho | \tilde\phi} = \sum_i |\langle\tilde\phi | \tilde\psi_i\rangle|^2$, implying $\braket{\tilde\phi | \tilde\psi_i} = 0$. Hence, $\ket{\tilde\psi_i} \in \operatorname{Ker}(\rho)^\perp = \operatorname{Im}(\rho)$.
\end{proof}
Therefore, if no \emph{pure} symmetric state is tripartite separable, then by the lemma above, no \emph{strongly} symmetric \emph{mixed} state is either. Note that the lemma above does not assume the states $\ket{\tilde\psi_i}$ in a given decomposition of $\rho$ are orthogonal.

A direct corollary of Lemma \ref{lemma:strongly_symmetric_decomposition} ties together strong symmetry to trivial symmetry extensions, as alluded earlier:
\begin{corollary}\label{corollary:strong_symmetry_extension}
    Any purification $\ket{\Psi} \in \Hilb \otimes \Hilb_\anc$ of a strongly symmetric state $\rho$, i.e. $U \rho = \lambda \rho$, is symmetric under the trivially extended unitary $U \otimes \one_\anc$ with same charge $\lambda$.
\end{corollary}
\begin{proof}
    Let $\ket{\Psi} = \sum_i \sqrt{p_i} \ket{\psi_i} \ket{i}_\anc$ be the Schmidt decomposition of $\ket{\Psi}$. Since the vectors $\ket{\tilde\psi_i} \defeq \sqrt{p_i} \ket{\psi_i}$ forms a decomposition of $\rho$, Lemma \ref{lemma:strongly_symmetric_decomposition} implies they are symmetric under $U$, which in turn implies $U \otimes \one_\anc \ket{\Psi} = \lambda \ket{\Psi}$.
\end{proof}

When $U(g)$ is part of a group representation, $\lambda(g)$ forms a one-dimensional irrep. In such case, a strongly symmetric $\rho$ is {still} a mixture of pure states with the same charge given by $\lambda$, whereas a weakly symmetric one can be an ensemble of pure states with different charges.

These definitions are compatible with weakly and strongly symmetric channels~\cite{buca_note_2012, albert_symmetries_2014, albert_lindbladians_2018, lieu_symmetry_2020, deGroot2022}, which we briefly review here. A weakly symmetric channel commutes with the adjoint action of the symmetry, $\rho \mapsto U\rho U^\dagger$, and a strongly symmetric channel has a collection of Kraus operators $K_i$, all commuting with $U$. Naturally, a strongly symmetric channel is also weakly symmetric. A physical motivation for the former is that the strong symmetry condition for Lindbladian evolutions implies each jump operator commutes with the symmetry, which is the algebraic form of requiring symmetric interactions between the system and its environment. As expected, a weakly symmetric state stays weakly symmetric after the application of a weakly symmetric channel, and the same is true for strong symmetry.

\subsection{Anomaly}
\label{sec:intro-anomaly}

 For a large class of anomalies -- specifically those classified by group cohomology (bosons)~\cite{chen_symmetryprotected_2012,CZX,else_classifying_2014} and supercohomology (fermions)~\cite{Supercohomology,Metlitski2019} -- the anomaly can be realized in an on-site Hilbert space (i.e. given by \eqref{tensor-product-hilb}), with the symmetry acting in a non-on-site manner. In this work we focus on bosonic systems with anomaly classified by group cohomology -- for bosonic open systems in space dimension $d<4$ all nontrivial SPT phases fall into this class~\cite{Ma_intrisic_2023}. We begin with (0+1)-D systems with a unitary $G$ symmetry, which can be thought of as the boundary of a (1+1)-D SPT.

The anomaly of the boundary of a (1+1)-D SPT system is manifested by the projective action of the symmetry near each edge~\cite{AKLT,Pollman2010,Pollman2012}. More precisely, the symmetry representation has the same action on ground states as another representation $U(g) = U_L(g) \otimes U_R(g)$ with support only near the left and right endpoints of the chain. This restriction to the edge is possible because the bulk SPT system is short-range entangled~\cite{else_classifying_2014}. Importantly, each restriction $U_{L,R}$ is a \emph{projective} representation,
\begin{equation}
    U_{L,R}(g) U_{L,R}(h) = \omega_{L,R}(g, h) U_{L,R}(g h),
\end{equation}
with $\omega_L(g, h) = \omega_R(g,h)^{-1} \in U(1)$ complex phases that cancel each other so $U$ remains non-projective. Thus, without loss of generality, we can focus only on the left endpoint. Associativity of $U_L$ implies
\begin{equation}
    \omega_L(g, h)\omega_L(gh, k) = \omega_L(g, hk) \omega_L(h, k),
\end{equation}
which means $\omega_L : G \times G \to U(1)$ is a 2-cocycle. Moreover, the freedom to redefine $U_L(g) \mapsto \lambda(g) U_L(g)$ implies we should identify $\omega_L(g, h) \sim \omega_L(g, h) \lambda(gh) \lambda(g)^{-1} \lambda(h)^{-1}$. The equivalence classes formed from this identification are elements of the second cohomology group $H^2(G, U(1))$. When an edge symmetry $U$ corresponds to a nontrivial cohomology class, we say it is \emph{anomalous}. The anomaly is a property of the SPT phase itself since it is invariant under symmetric local unitaries acting on the bulk and how the symmetry is restricted to the edge.

\begin{figure}[t]
    \centering
    \def\svgwidth{0.95\linewidth}
    \begingroup%
  \makeatletter%
  \providecommand\color[2][]{%
    \errmessage{(Inkscape) Color is used for the text in Inkscape, but the package 'color.sty' is not loaded}%
    \renewcommand\color[2][]{}%
  }%
  \providecommand\transparent[1]{%
    \errmessage{(Inkscape) Transparency is used (non-zero) for the text in Inkscape, but the package 'transparent.sty' is not loaded}%
    \renewcommand\transparent[1]{}%
  }%
  \providecommand\rotatebox[2]{#2}%
  \newcommand*\fsize{\dimexpr\f@size pt\relax}%
  \newcommand*\lineheight[1]{\fontsize{\fsize}{#1\fsize}\selectfont}%
  \ifx\svgwidth\undefined%
    \setlength{\unitlength}{152.84838218bp}%
    \ifx\svgscale\undefined%
      \relax%
    \else%
      \setlength{\unitlength}{\unitlength * \real{\svgscale}}%
    \fi%
  \else%
    \setlength{\unitlength}{\svgwidth}%
  \fi%
  \global\let\svgwidth\undefined%
  \global\let\svgscale\undefined%
  \makeatother%
  \begin{picture}(1,0.34255795)%
    \lineheight{1}%
    \setlength\tabcolsep{0pt}%
    \put(0,0){\includegraphics[width=\unitlength,page=1]{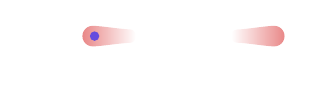}}%
    \put(0.29868085,0.16597666){\makebox(0,0)[t]{\lineheight{1.25}\smash{\begin{tabular}[t]{c}1\end{tabular}}}}%
    \put(0.85543232,0.16597666){\makebox(0,0)[t]{\lineheight{1.25}\smash{\begin{tabular}[t]{c}L\end{tabular}}}}%
    \put(0.7640501,0.1646255){\makebox(0,0)[t]{\lineheight{1.25}\smash{\begin{tabular}[t]{c}L-1\end{tabular}}}}%
    \put(0.66887509,0.16597666){\makebox(0,0)[t]{\lineheight{1.25}\smash{\begin{tabular}[t]{c}L-2\end{tabular}}}}%
    \put(0.38981347,0.16597666){\makebox(0,0)[t]{\lineheight{1.25}\smash{\begin{tabular}[t]{c}2\end{tabular}}}}%
    \put(0.48388727,0.16659441){\makebox(0,0)[t]{\lineheight{1.25}\smash{\begin{tabular}[t]{c}3\end{tabular}}}}%
    \put(0,0){\includegraphics[width=\unitlength,page=2]{cluster_chain.pdf}}%
    \put(0.76117487,0.03115297){\color[rgb]{0.87058824,0.54901961,0}\makebox(0,0)[t]{\lineheight{1.25}\smash{\begin{tabular}[t]{c}X\end{tabular}}}}%
    \put(0.38935349,0.03115297){\color[rgb]{0.87058824,0.54901961,0}\makebox(0,0)[t]{\lineheight{1.25}\smash{\begin{tabular}[t]{c}X\end{tabular}}}}%
    \put(0.29867802,0.09984856){\color[rgb]{0.39607843,0.29803922,0.87058824}\makebox(0,0)[t]{\lineheight{1.25}\smash{\begin{tabular}[t]{c}X\end{tabular}}}}%
    \put(0.85551877,0.09984856){\color[rgb]{0.39607843,0.29803922,0.87058824}\makebox(0,0)[t]{\lineheight{1.25}\smash{\begin{tabular}[t]{c}X\end{tabular}}}}%
    \put(0.66837194,0.09984856){\color[rgb]{0.39607843,0.29803922,0.87058824}\makebox(0,0)[t]{\lineheight{1.25}\smash{\begin{tabular}[t]{c}X\end{tabular}}}}%
    \put(0.48352491,0.09984856){\color[rgb]{0.39607843,0.29803922,0.87058824}\makebox(0,0)[t]{\lineheight{1.25}\smash{\begin{tabular}[t]{c}X\end{tabular}}}}%
    \put(0.26618948,0.09984856){\color[rgb]{0.39607843,0.29803922,0.87058824}\makebox(0,0)[rt]{\lineheight{1.25}\smash{\begin{tabular}[t]{r}$X_\text{odd}=$\end{tabular}}}}%
    \put(0.26618947,0.03586815){\color[rgb]{0.87058824,0.54901961,0}\makebox(0,0)[rt]{\lineheight{1.25}\smash{\begin{tabular}[t]{r}$X_\text{even}=$\end{tabular}}}}%
    \put(0.57054582,0.21587472){\makebox(0,0)[t]{\lineheight{1.25}\smash{\begin{tabular}[t]{c}$\cdots$\end{tabular}}}}%
    \put(0.57054582,0.1004523){\makebox(0,0)[t]{\lineheight{1.25}\smash{\begin{tabular}[t]{c}$\cdots$\end{tabular}}}}%
    \put(0.57054582,0.03175671){\makebox(0,0)[t]{\lineheight{1.25}\smash{\begin{tabular}[t]{c}$\cdots$\end{tabular}}}}%
    \put(0,0){\includegraphics[width=\unitlength,page=3]{cluster_chain.pdf}}%
  \end{picture}%
\endgroup%

    \caption{Cluster chain for odd number of sites $L$. The $\Z_2 \times \Z_2$ symmetry is realized by $\Xo$ and $\Xe$. When the low-energy system is anomalous and symmetric, the degenerate edge states (in red) are entangled.}
    \label{fig:cluster_chain}
\end{figure}

Let us look at a concrete SPT system with exactly solvable edge theory: the 1d cluster chain, with Hamiltonian
\begin{equation}\label{eq:cluster_ham}
    H_{ZXZ} = -\sum_{i=2}^{L-1} Z_{i-1} X_i Z_{i+1}.
\end{equation}
This Hamiltonian has a $\Z_2 \times \Z_2$ symmetry generated by $\Xe \defeq \prod_{k=1}^{\lfloor L/2 \rfloor} X_{2k}$ and $\Xo \defeq \prod_{k=1}^{\lceil L/2 \rceil} X_{2k-1}$ (See Fig.~\ref{fig:cluster_chain}). $H_{ZXZ}$ can be exactly solved by transforming it under the local unitary transformation $\CZ \defeq \prod_{i=1}^{L-1} CZ_{i, i+1}$, as $\CZ^\dagger H_{ZXZ} \CZ = -\sum_{i=2}^{L-2} X_i$. Importantly, the ground state subspace of $H_{ZXZ}$ is four-fold degenerate, generated by states of the form $\CZ \ket{\phi_L} \otimes \ket{+ + \cdots +} \otimes \ket{\phi_R}$, where $\ket{\phi_L}$ and $\ket{\phi_R}$ are arbitrary one-qubit states at the left and right endpoints of the chain, respectively, and $\ket{+}$ is the $+1$ eigenstate of $X$. Since this mapping from the $ZXZ$ Hamiltonian to a trivial quantum paramagnet is a local unitary, then, by definition, the ground states of $H_{ZXZ}$ are short-range entangled in the bulk. However, note that each gate of $\CZ$ is not symmetric under $\Xe$ and $\Xo$. In fact, no local unitary with symmetric gates could trivialize $H_{ZXZ}$, as it is topologically protected by the $\Z_2 \times \Z_2$ symmetry above.

Restricting to the ground state subspace, $\Xo \stackrel{\text{GS}}{=} X_1 \otimes X_L$ and $\Xe \stackrel{\text{GS}}{=} Z_1 X_2 \otimes X_{L-1} Z_L$ for odd system size $L$, and for even $L$, $\Xo \stackrel{\text{GS}}{=} X_1 \otimes X_{L-1} Z_L$ and $\Xe \stackrel{\text{GS}}{=} Z_1 X_2 \otimes X_L$, where we have used $\otimes$ to indicate the separation between the two endpoints of the chain. In both cases, we have $U(g) = U_L(g) \otimes U_R(g)$, $g \in \Z_2 \times \Z_2$, with $U_L$ and $U_R$ \emph{projective} representations of $\Z_2 \times \Z_2$, as their generators anticommute. These representations are nothing but the logical $X$ and $Z$ operators of the effective qubits at each end of the chain.

For (2+1)-D SPTs, the boundary anomaly is manifested by the non-on-site character of the symmetry action~\cite{chen2013,else_classifying_2014}. Generically, the boundary symmetry operator $U(g)$ is a finite-depth local unitary (FDLU), up to exponential tails. Moreover, there is no equivalent on-site representation. One way to see this is by the reduction procedure of Else and Nayak~\cite{else_classifying_2014}, which we summarize below. Since we will revisit the same arguments in Sec.\ref{sec:tripartite_entanglement}, we will leave the details to that section.

Given a (2+1)-D SPT, we can still restrict the symmetry action to the boundary when the bulk is in the (SRE) ground state. Let $U(g)$ be this boundary symmetry representation, which acts as a finite-depth local unitary. Following Else and Nayak~\cite{else_classifying_2014}, $U(g)$ can be restricted to a unitary $U_M(g)$ with support on a region $M$ of the boundary so that $U_M$ acts as $U$ in the bulk of $M$. Thus, $U_M$ should still be a linear representation, up to boundary operators $\Omega = \Omega_L \otimes \Omega_R$:
\begin{equation}
U_M(g) U_M(h) = \Omega(g, h) U_M(g h).
\end{equation}
Similarly to the 0d anomaly, associativity of $U_M(g) U_M(h) U_M(k)$ imposes a consistency relation on $\Omega$, which is also satisfied by $\Omega_L$ and $\Omega_R$ individually \emph{up to a phase} $\omega(g, h, k) \in U(1)$.
\begin{equation}
\begin{split}
    &\omega(g,h,k) \Omega_L(g,h)\Omega_L(gh,k) \\ & = U_{M}(g)\Omega_L(h,k) U_{M}(g)^{-1} \Omega_L(g,hk).
\end{split}
\end{equation}
One can prove~\cite{else_classifying_2014} that such $\omega$ is a closed 3-cocyle:
\begin{equation}
    \omega(g,h,k)\omega^{-1}(gh,k,l)\omega(g,hk,l)\omega^{-1}(g,h,kl)\omega(h,k,l)=1,
\end{equation}
or more compactly $\delta\omega=1$. Furthermore, the phase indetermination of $\Omega_{L,R}$ implies we have to identify $\omega$'s that differ by exact 3-cocyles:
\begin{equation}
    \omega(g,h,k)\sim\omega(g,h,k)\mu(g,h)\mu^{-1}(gh,k)\mu(h,k)\mu^{-1}(g,hk),
\end{equation}
or more compactly $\omega\sim\omega\cdot\delta\mu$. Thus, to each (2+1)-D SPT phase, there is an element $[\omega]$ in the third cohomology group $H^3(G, U(1))$ that classifies it.

As an example, for $G=\Z_2=\{1,g\}$, we can canonically choose $\Omega(1,g)=\Omega(g,1)=\one$, and then the only nontrivial 3-cocycle $\omega(g,g,g)$ is defined by
\begin{equation}
    \omega(g,g,g) \Omega_L(g,g)= U_M(g) \Omega_L(g,g) U_M(g)^{-1}.
\end{equation}
In words, $\omega(g,g,g)$ measures whether $\Omega_L(g,g)$ is charged under $g$ or not.

For higher dimensional SPT systems, similar reduction procedures can be performed to reduce the problem of classifying $d$-dimension anomaly to $d=0$, where it goes back to a projective phase, but now in the $H^{d+2}(G, U(1))$ cohomology group. We leave the details of the higher-dimensional construction to Sec.~\ref{sec:higher_dims}, where it is applied to the anomaly-nonseparability connection.

\section{Prototypical examples}
\label{sec:prototypical_examples}

Before discussing the anomaly-nonseparability connection for generic anomalous systems, we illustrate it with some concrete examples.

\subsection{(0+1)-D: Projective representation}
\label{sec:projective_rep}

As we saw in Sec. \ref{sec:intro-anomaly}, one-dimensional SPT systems have anomalous edge states characterized by the projective nature of the symmetry action. Furthermore, we will argue now that the symmetry operator $U_L$ at the left endpoint cannot have an invariant state $\ket{\psi} \in \Hilb_L$, and equally for $U_R$. Indeed, if $U_{L}(g) \ket{\psi} = \lambda(g) \ket{\psi}$ for some $\lambda : G \to U(1)$, then the projective factor $\omega(g, h) \in U(1)$ in
\begin{equation}\label{eq:projective_rep}
    U_L(g) U_L(h) = \omega(g, h) U_L(gh)
\end{equation}
would equal to $\omega(g, h) = \lambda(g) \lambda(h) / \lambda(gh)$, by acting on \eqref{eq:projective_rep} with $\ket{\psi}$. Such $\omega$ can be trivialized by redefining $U_L(g) \mapsto \lambda(g)^{-1} U_L(g)$ and $U_R(g) \mapsto \lambda(g) U_R(g)$, contradicting the assumption that they are (nontrivial) projective representations. This implies that any symmetric pure state $|\Psi\rangle\in\Hilb_L\otimes\Hilb_R$ must be entangled, namely $|\Psi\rangle\neq|\psi_L\rangle\otimes|\psi_R\rangle$. By Lemma~\ref{lemma:strongly_symmetric_decomposition}, we immediately conclude that any strongly symmetric mixed state $\rho \in Q(\Hilb_L \otimes \Hilb_R)$ must be bipartite nonseparable. This is the simplest case of the anomaly-nonseparability connection.

Coming back to the cluster chain example, it is easy to see that there cannot be a symmetric pure state at one of the edges, since it would be a simultaneous eigenstate of both logical $X$ and $Z$. Therefore, the only symmetric ground states are the ones where the endpoints are entangled (see Fig.~\ref{fig:cluster_chain}). For the cluster chain, these correspond to the states $\CZ \ket{\Psi_{ab}}_{1L} \bigotimes_{i=2}^{L-1} \ket{+}_i$, $a,b \in \{0, 1\}$, where $\ket{\Psi_{ab}} = (1/\sqrt{2})(\one \otimes Z^a X^b) (\ket{00} + \ket{11})$ is a Bell pair state. Equivalently, the only \emph{mixed} state of \emph{one qubit} that is \emph{weakly} symmetric under the anomalous symmetry is the maximally mixed state $\one / 2$, which purifies to a Bell state. Since there is only one Bell state in each $\Z_2 \times \Z_2$ symmetry sector, then this trivially says that any \emph{strongly} symmetric mixed state at the (0+1)-D edge of the (1+1)-D cluster chain SPT is also bipartite entangled.

Similar statements can be established for more general $G$. For simplicity, let us assume that the edge Hilbert space transforms as an irreducible projective representation. In this case, by Schur's lemma, the only state $\rho$ weakly symmetric under $G$ (i.e. $\rho U(g)=U(g)\rho$ for all $g\in G$) is the maximally mixed state $\rho\propto \one$, which can be purified into a maximally entangled state between the two edges.

\subsection{(1+1)-D: 4-qubit CZX}

Moving on to one spatial dimension, we now consider a system with anomalous $\Z_2$ symmetry as our prototypical example. It can appear on the boundary of a $\Z_2$-protected bosonic topological phase~\cite{CZX,LevinGu}. Concretely, the system is a simple 1d chain of $L$ qubits with $\Z_2$ symmetry realized by the CZX unitary~\cite{CZX}:
\begin{equation}
    U_{\rm CZX}=\prod_{i=1}^L X_i \prod_{i=1}^L CZ_{i, i+1},
\end{equation}
where $CZ_{i j} = 1 - 2 \ket{1_i 1_j}\bra{1_i 1_j}$ is the controlled-Z gate between qubits $i$ and $j$ with $\{\ket{0}, \ket{1}\}$ the Pauli-Z basis. Since $U_{\text{CZX}}^2 = (-1)^L$, we require even number of sites. We will show that any density matrix $\rho$ that is strongly symmetric under this $\Z_2$ symmetry ($U_{\rm CZX}\rho=\lambda\rho$) cannot be tripartite separable. The essence of our arguments can already be seen through $4$-qubit system -- $L=4$ qubits is the smallest system size for which $U_{\text{CZX}}$ is nontrivially distinct from ordinary on-site $\Z_2$. For higher even $L$, the same arguments and conclusions follow, with suitable modifications\footnote{For odd system size $L \geq 3$, we may redefine $U_{\rm CZX} \to i U_{\rm CZX}$ so it squares to $+1$. In that case, the results for even $N$ follow accordingly, up to global phases.}. We will then show, with explicit examples, that a similar statement on bipartite nonseparability does not hold.

\subsubsection{Tripartite entanglement}\label{sec:czx_tripartite_entanglement}

\label{sec:4-qubit_CZX}
\begin{figure}[t]
    \centering
    \scalebox{1.3}{\begin{tikzpicture}
    \def\dotradius{0.05cm}
    \def\squaresize{1cm}
    \def\linewidth{1}
    \def\circleradius{0.75*\squaresize}
    \def\circlewidth{2}

    \definecolor{colorA}{HTML}{55AFE0}
    \definecolor{colorB}{HTML}{FFAE00}
    \definecolor{colorC}{HTML}{009E73}

    \begin{scope}
        \fill (0, 0) circle (\dotradius);
        \fill (0, \squaresize) circle (\dotradius);
        \fill (\squaresize, 0) circle (\dotradius);
        \fill (\squaresize, \squaresize) circle (\dotradius);
    
        \draw[color=colorC,rounded corners, line width=\linewidth] (\squaresize*3/4, -\squaresize/4) rectangle (1.25*\squaresize, \squaresize/4);
        \draw[color=colorB,rounded corners, line width=\linewidth] (\squaresize*3/4, \squaresize*3/4) rectangle (1.25*\squaresize, 1.25*\squaresize);
    
        \draw[color=colorA,rounded corners, line width=\linewidth] (-\squaresize/4, -\squaresize/4) rectangle (\squaresize/4, \squaresize*5/4);

        \node[color=colorA] at (0, 1.5*\squaresize) {A};
        \node[color=colorB] at (\squaresize, 1.5*\squaresize) {B};
        \node[color=colorC] at (\squaresize, 0.5*\squaresize) {C};
    
        \node at (-\squaresize, \squaresize/2) 
        {$\begin{pmatrix} \alpha_{00} \\ \alpha_{01} \\ \alpha_{10} \\ \alpha_{11} \end{pmatrix}$=};
    
        \draw[->] (1.5*\squaresize, \squaresize/2) -- (1.8*\squaresize, \squaresize/2);
    \end{scope}


    \begin{scope}[xshift=2*\squaresize, yshift=-0.25*\squaresize]
        \draw[color=colorA, line width=\circlewidth] (\circleradius, 2*\circleradius) arc (90:270:\circleradius);
        \draw[color=colorB, line width=\circlewidth] (2*\circleradius, \circleradius) arc (0:90:\circleradius);
        \draw[color=colorC, line width=\circlewidth] (\circleradius, 0) arc (270:360:\circleradius);


        \def\nsites{5}
        \pgfmathtruncatemacro{\totnsites}{4*\nsites}
        \foreach \i in {1,...,\totnsites}
        {
            \pgfmathtruncatemacro{\angle}{90*(\i-0.5)/\nsites}
            \pgfmathtruncatemacro{\checki}{\i == 1 || \i == 2*\nsites || \i == 2*\nsites+1 || \i == 4*\nsites ? 1 : 0}
        
            \ifnum\checki=1
                \filldraw[black,fill=black] ({\circleradius*(1+sin(\angle))}, {\circleradius*(1+cos(\angle))}) circle (\dotradius);
            \else
                \filldraw[black,fill=none] ({\circleradius*(1+sin(\angle))}, {\circleradius*(1+cos(\angle))}) circle (\dotradius);
            \fi
        }

    \end{scope}
    
\end{tikzpicture}}
    \caption{4-qubit tripartite separable state of the CZX model. Their behavior is representative of spins near the boundaries between regions AB and AC of a much larger spin chain.}
    \label{fig:4_dots}
\end{figure}
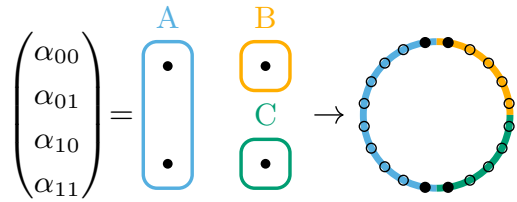

We now prove that there is no mixed state $\rho$ that is strongly symmetric under $U_{\text{CZX}}$ and is tripartite separable, that is, there is no $\rho$ satisfying $U_{\text{CZX}} \rho = \lambda \rho$, $\lambda \in U(1)$, and also
\begin{equation}\label{eq:tripartite_separable_mixed_def}
    \rho = \sum_{i} p_i \ketbra{A_i}{A_i} \otimes \ketbra{B_i}{B_i} \otimes \ketbra{C_i}{C_i},
\end{equation}
where we have partitioned the 4-qubit chain into three regions, $A$, $B$ and $C$, with $\ket{A_i} \in \Hilb_A$ etc. Without loss of generality, we may assume $A$ contains two adjacent spins, whereas $B$ and $C$ contain just one. See Fig.~\ref{fig:4_dots}.

First, we use Lemma \ref{lemma:strongly_symmetric_decomposition} to reduce the problem from mixed to pure states by noting that $U_{\text{CZX}} \rho = \lambda \rho$ implies $U_{\text{CZX}} \ket{\psi} = \lambda \ket{\psi}$ for all $\ket{\psi}$ in the range of $\rho$. For the tripartite separable state of Eq.\eqref{eq:tripartite_separable_mixed_def}, this means $U_{\text{CZX}} \ket{A_i} \ket{B_i} \ket{C_i} = \lambda \ket{A_i} \ket{B_i} \ket{C_i}$, even if such states are not orthogonal to each other. Let us now focus on one such pure state $\ket{A} \ket{B} \ket{C}$ in the decomposition of $\rho$ and assume by contradiction that it is symmetric under $U_{\text{CZX}}$.

Let $\alpha_{ij}$, $i,j \in \{0,1\}$, be a nonzero coefficient in the $Z$-basis expansion of $\ket{A} = \sum_{i,j\in \{0,1\}} \alpha_{ij} \ket{ij}$. By the CZX symmetry $|\alpha_{\bar{i}\bar{j}}|=|\alpha_{ij}|\neq0$ (we denote $\bar{k}=1-k$). We then have
\begin{equation}
    \prescript{}{A}{\langle \bar{i}\bar{j}}|U_{\rm CZX}|A\rangle|B\rangle|C\rangle=\alpha_{\bar{i}\bar{j}}\lambda|B\rangle|C\rangle.
\end{equation}
This means that we can restrict the symmetry $U_{\rm CZX}$ to the $BC$ region via
\begin{align}
\label{eq:4qubitrestrict}
    U_{BC} & \defeq \prescript{}{A}{\braket{\bar{i}\bar{j} | U_{\text{CZX}} | A}} / \alpha_{\bar{i}\bar{j}} \nonumber \\
           & = \frac{\alpha_{ij}}{\alpha_{\bar{i}\bar{j}}}(-1)^{ij} Z_B^i Z_C^j \cdot U_{BC}^{\text{strict}}  ,
\end{align}
where $U_{BC}^\text{strict} = X_B X_C CZ_{BC}$ is the restriction of symmetry by taking only the operators that act strictly inside region $BC$. Crucially, the restricted operator $U_{BC}$ is an unitary symmetry for the remaining state in $BC$:
\begin{equation}
    U_{BC} \ket{B} \ket{C} =  \lambda \ket{B} \ket{C}.
\end{equation}
Nevertheless, instead of being a $\Z_2$-valued operator, $U_{BC}$ satisfies the algebra
\begin{equation}
 U_{BC}^2 \propto (U_{BC}^{\rm strict})^2\propto Z_B Z_C,
\end{equation}
 where $\propto$ means equal up to some phase factor. This is a simple manifestation of the Else-Nayak mechanism~\cite{else_classifying_2014}. Since $|B\rangle|C\rangle$ must also be an eigenstate of $U_{BC}^2$, it follows that $Z_B \ket{B} = \mu_B \ket{B}$ and $Z_C \ket{C} = \mu_C \ket{C}$ for some $\mu_B, \mu_C = \pm 1$. This immediately contradicts with the fact that the operators $Z_B$ and $Z_C$ are charged under the $\Z_2$ symmetry, i.e. $U_{\rm CZX} Z_{B(C)} = - Z_{B(C)} U_{\rm CZX}$, as this would imply $\lambda \mu_{B(C)} = - \mu_{B(C)} \lambda = 0$. Thus, our initial assumption that there is a strongly symmetric tripartite separable state cannot hold.

The above argument can be extended, with minor modifications, to arbitrary even system length $L$. For a larger segment $A$, we can write
\begin{equation}
    |A\rangle=\sum_{ij}\alpha_{ij}|ij\rangle\otimes|\psi_{ij}\rangle,
\end{equation}
where $|ij\rangle$ labels the two spins, in the $Z$-basis, at the two boundaries of $A$ (one adjacent to $B$ and the other to $C$), and $|\psi_{ij}\rangle$ labels all the other spins in $A$. We can then construct the restricted symmetry action on $BC$ similar to Eq.~\eqref{eq:4qubitrestrict}:
\begin{align}
    U_{BC} & \defeq \prescript{}{A}{\langle\psi_{\bar{i}\bar{j}}}|\braket{\bar{i}\bar{j} | U_{\text{CZX}} | A} / \alpha_{\bar{i}\bar{j}} \nonumber \\
           & \propto  Z_{b}^i Z_{c}^j \cdot U_{BC}^{\text{strict}}  ,
    \label{eq:CZXUBC}
\end{align}
where $U_{BC}^\text{strict}$ is again the restriction of symmetry by taking only the operators that act strictly inside region $BC$, and $Z_{b}$ ($Z_{c}$) is the Pauli-$Z$ operator for the spin at the boundary of $B$ ($C$) adjacent to $A$. Crucially, we still have
\begin{equation}
\label{eq:CZXZZ}
    U^2_{BC} \propto Z_{b}Z_{c},
\end{equation}
which leads to the conclusion that a symmetric tripartite separable state does not exist.

\subsubsection{Bipartite separability}\label{sec:czx-bipartite-separability}

Although there are no tripartite separable states strongly symmetric under the CZX anomalous symmetry, there is a complete basis of bipartite separable symmetric pure states on any two connected regions $A$ and $B = \comp{A}$, the complement of $A$. These are formed from the tensor product of cat states in each of the regions: $\ket{\alpha^{(s)}} \otimes \ket{\beta^{(t)}} \in \Hilb_A \otimes \Hilb_B$, where $\ket{\xi^{(r)}} \equiv \frac{1}{\sqrt{2}} (\ket{\xi} + r \ket{\comp{\xi}})$, with $\xi$ a bitstring and $r \in U(1)$. Without doing additional calculations, we can already guess that such cat states might form symmetric bipartite-separable states from the proof of tripartite entanglement presented in the last section. This is because one can still follow a similar procedure outline above to find a restricted symmetry action $U_A$ of the state in $A$ satisfying $U_A^2 \propto Z_1 Z_{|A|}$. This implies a long-range connected correlation function $|\langle Z_1 Z_{|A|} \rangle - \langle Z_1 \rangle \langle Z_{|A|} \rangle| > 0$, which is satisfied by the prototypical cat states $\ket{\alpha^{(s)}}$. More explicitly, we will find the constraints over $\alpha, \beta, t, s$ and $\lambda$ such that
\begin{equation}
    U \ket{\alpha^{(s)}} \ket{\beta^{(t)}}
     = \lambda \ket{\alpha^{(s)}} \ket{\beta^{(t)}}.
\end{equation}
After some algebra, we find that the symmetry condition can be satisfied if and only if
\begin{equation}\label{eq:bipartite_symmetry_condition}
    s^2=t^2=(-1)^{N_A+1+\alpha_1+\alpha_{|A|}+\beta_1+\beta_{|B|}},
\end{equation}
where $N_A$ is the number of sites of region $A$, $\alpha_1, \alpha_{|A|}$ ($\beta_1,\beta_{|B|}$) are the values of $\alpha$ ($\beta$) at the two boundaries of $A$ ($B$), where $\alpha_1$ is adjacent to $\beta_1$ and $\alpha_{|A|}$, to $\beta_{|B|}$. Equation \eqref{eq:bipartite_symmetry_condition} dictates whether $s$ and $t$ take values in $\pm 1$ or $\pm i$, depending on the bitstrings $\alpha$ and $\beta$. The corresponding symmetry eigenvalue is
\begin{equation}
\label{eq:separableeigenvalue}
\lambda=st(-1)^{\sum_{i=1}^{|A|-1}\alpha_i\alpha_{i+1}+\sum_{j=1}^{|B|-1}\beta_i\beta_{i+1}+\alpha_1\beta_1+\alpha_{|A|}\beta_{|B|}}.
\end{equation}

Since any pair of bitstrings $\alpha, \beta$ gives a total of 4 orthogonal symmetry eigenstates by varying the signs of $t$ and $s$, the collection of all such bipartite separable states forms a complete basis.

We expect this procedure of finding a symmetric bipartite separable basis to apply similarly for $\Z_n$ and any finite Abelian symmetry group.

\subsubsection{Strongly symmetric infinite temperature state}\label{sec:infinite_T_CZX}

A natural way to generate strongly symmetric states is to consider a thermal state of a symmetric Hamiltonian $H$ at temperature $T = \beta^{-1}$ in the ``canonical ensemble", i.e. that only includes states transforming as a particular 1d irrep of the symmetry group:
\begin{equation}
 \rho_{T,\lambda} \propto \sum_{\ket{E} \in V_\lambda} e^{-\beta E} \ketbra{E}{E} = P_\lambda e^{-\beta H}.
\end{equation}
Here $\lambda: G\rightarrow U(1)$ denotes the 1d irrep, and $V_\lambda \subset \Hilb$, its subspace:
\begin{equation}
    V_\lambda=\{\ket{\psi} \mid U(g)\ket{\psi}=\lambda(g) \ket{\psi},\forall g\in G\}.
\end{equation}
 and $P_\lambda$ is the orthogonal projector to $V_{\lambda}$. In particular, the infinite temperature limit $\lim_{T \to \infty} \rho_{T, \lambda} = \frac{1}{\dim V_\lambda}P_\lambda$ depends only on the symmetry, and not on the particular Hamiltonian $H$.

In the case of the $\Z_2$ CZX symmetry, we have two infinite temperature states
\begin{equation}\label{eq:infinite-temperature_CZX}
    \rho_{\infty, \pm} = \frac{1}{2^L} (\one \pm U_{\text{CZX}}),
\end{equation}
corresponding to the two values the symmetry operator $U_{\text{CZX}}$ can take. $\rho_{\infty,\pm}$ is the maximally mixed state within the charge sector $V_{\pm} \subset \mathcal{H}$. From the previous discussion in Sec.~\ref{sec:czx-bipartite-separability}, we know that for each charge sector $V_{\pm}$ and any bipartition $\{A, B=\comp{A}\}$, we have a complete basis of bipartite separable states $\{|\alpha^{(s)}\rangle\otimes|\beta^{(t)}\rangle\}$. This immediately implies that
\begin{equation}
    \rho_{\infty,\pm}\propto \sum_{\substack{\alpha^{(s)},\beta^{(t)} \\ \lambda=\pm}} \ketbra{\alpha^{(s)}}{\alpha^{(s)}}_A \otimes \ketbra{\beta^{(t)}}{\beta^{(t)}}_{B},
\end{equation}
where the summation is only constrained by the eigenvalue condition Eq.~\eqref{eq:separableeigenvalue}. Since $\{A, B\}$ is arbitrary, we conclude that $\rho_{\infty,\pm}$ is bipartite separable for all bipartitions. Yet, surprisingly, from our discussion in Sec.~\ref{sec:czx_tripartite_entanglement} we know that $\rho_{\infty,\pm}$ is not separable for any tripartition.

To numerically verify that $\rho_{\infty, \pm}$ is tripartite entangled, we employ the multipartite permutation criterion of \cite{horodecki_separability_2006}. In summary, the criterion says that permuting the indices $(i_1,\ldots,i_n,j_1,\ldots,j_n)$ -- including input $i_\alpha$ and output $j_\alpha$ indices -- of an $n$-partite \emph{separable} density matrix $\rho = [\rho_{i_1,\ldots,i_n}^{j_1,\ldots,j_n}]$ can only decrease its trace norm from $\norm{\rho}_{1} \defeq \Tr |\rho| = 1$. Conversely, if some index permutation increases the trace norm to a value higher than $\norm{\rho}_{1} = 1$, then it is nonseparable. For the 4-qubit $\rho_{\infty, \pm}$, this happens if one permutes $i_B$ and $j_C$ in $[\rho_{i_A, i_B, i_C}^{j_A, j_B, j_C}]$, where $i_R$ and $j_R$ are the input and output (multi-)indices of region $R \in \{A, B, C\}$, respectively (See Fig. \ref{fig:4_dots}). Namely, we prove in Appendix \ref{appendix:permutation_creterion} that the resulting matrix after the permutation  $\tilde{\rho} = [\rho_{i_A, j_C, i_C}^{j_A, j_B, i_B}]$ has trace norm $\norm{\tilde\rho}_1 = \frac{1}{2} + \frac{1}{\sqrt{2}} > 1$.

From the results to be shown in Sec.~\ref{sec:long-range_entanglement}, we can conclude that not only $\rho_{\infty, \pm}$ is tripartite-entangled, but that it cannot be connected to a tripartite-separable state via a finite-depth local channel. This means that the tripartite entanglement is long-range in nature. Despite this, we show in Appendix \ref{appendix:adaptive_CZX} that it can be prepared via a finite-depth local \emph{adaptive} circuit, which includes not only local unitaries and local measurements but also (nonlocal) classical feedback~\cite{briegel_persistent_2001, lu_measurement_2022, bravyi_adaptive_2022}. This adaptive procedure adds insight into the state preparation complexity of $\rho_{\infty, \pm}$. Combined with the fact that many interesting properties of $\rho_{\infty, \pm}$ are already present for a system with just four qubits, the local adaptive procedure shows a practical way to realize such state in a quantum computer via mid-circuit measurements and feedback, as it has recently been done with long-range entangled states~\cite{foss-feig_experimental_2023a, iqbal_creation_2023}.

Another interesting property of the states $\rho_{\infty,\pm}$ is that they have the maximally disordered subsystems property~\cite{horodecki_informationtheoretic_1996}. This means that taking the partial trace of even one qubit at any site $i$ completely trivializes the system: $\Tr_i \rho_{\infty, \pm} = \one / 2^{L-1}$. In other words, $\rho_{\infty, \pm}$ has only \emph{global} correlations, and as such the state cannot be distinguished from the trivial maximally mixed state $\one / 2^{L}$ by measuring any $k$-point ($k$ finite) correlation function ${\rm Tr}[O(x_1)O(x_2)...O(x_k)\rho]$. The maximally disordered subsystems property, however, is not a generic feature of anomalous states, and it is not stable against local operations.

We have seen that strongly symmetric $T=\infty$ anomalous states are great examples of almost featureless states that, at the same time, cannot be trivial due to the anomaly-nonseparability constraint. Next, we will generalize the anomaly-nonseparability connection to other symmetries and dimensions.

\section{Multipartite nonseparability: General results}

In this section we generalize the discussions in Sec.~\ref{sec:prototypical_examples}. We consider bosonic lattice systems in $d$ spatial dimensions, and a global symmetry $G$ implemented by finite-depth local unitaries $U(g\in G)$ with maximum circuit depth $D\sim O(1)$. We further assume that the symmetry implementation has a quantum anomaly characterized by an element in the group-cohomology $[\omega]\in H^{d+2}(G,U(1))$. In one dimension our central result is a general connection between anomaly and tripartite nonseparability:

\begin{theorem}
\label{thm:3-sep}
    For $d=1$, if the quantum anomaly $[\omega]\in H^3(G,U(1))$ is nontrivial, and $\rho$ is strongly symmetric under $G$ in the sense that $\forall g\in G$, $U(g)\rho\propto \rho$,  then  \\
    \begin{enumerate}
        \item $\rho$ is tripartite nonseparable, namely
        \begin{equation}
            \rho\neq \sum_i p_i|A_i\rangle\langle A_i|\otimes |B_i\rangle\langle B_i|\otimes |C_i\rangle\langle C_i|,
        \end{equation}
        as long as $|A|, |B|, |C| \gg D$, the maximum depth of $U(g)$, and $A, B, C$ are all connected;
        \item  The tripartite entanglement of $\rho$ is long-ranged, in the sense that $\rho$ cannot be prepared from a tripartite-separable state (call it $\rho_{3s}$) via a finite-depth local channel, namely
        \begin{equation}
            \rho\neq \E_{\rm FD}\left[\rho_{3s}\right].
        \end{equation}
    \end{enumerate}
\end{theorem}

We will prove the first statement in Sec.~\ref{sec:tripartite_entanglement} and the second statement in Sec.~\ref{sec:long-range_entanglement}.

A natural generalization of Theorem~\ref{thm:3-sep} (and the trivial $d=0$ example from Sec.~\ref{sec:projective_rep}) is the following conjecture:
\begin{conjecture}
    In $d$ space dimension, if the quantum anomaly $[\omega]\in H^{d+2}(G,U(1))$ is nontrivial, and $\rho$ is strongly symmetric under $G$, then
    \begin{enumerate}
        \item $\rho$ is $(d+2)$-partite nonseparable for some $(d+2)$-partitions;
        \item The $(d+2)$-partite entanglement of $\rho$ is long-ranged:
        \begin{equation}
            \rho\neq\E_{\rm FD}[\rho_{(d+2)s}],
        \end{equation}
        where $\rho_{(d+2)s}$ is a $(d+2)$-partite separable state.
    \end{enumerate}
    \label{conj:higherdim}
\end{conjecture}
Although a general proof is unknown, in Sec.~\ref{sec:higher_dims} we motivate the above conjecture by proving a weaker version, when the symmetry action $U_g$ is assumed to take some specific form.

\subsection{Tripartite entanglement for (1+1)-D systems}\label{sec:tripartite_entanglement}

Let $U(g)$ be a finite-depth local unitary action of a symmetry group $G$ on a (1+1)-D system with the anomaly represented by a cohomology class $[\omega] \in H^3(G, U(1))$. We now show that if there is a strongly symmetric state $\rho$ that is tripartite separable, then there is no anomaly, i.e. $[\omega]=[1]$. As discussed in Secs.~\ref{sec:intro-strong_and_weak_symmetries} and \ref{sec:prototypical_examples}, since any pure state in the ensemble $\rho$ is also symmetric under $G$, it suffices to consider just a pure state.

\begin{figure}[t]
    \centering
    \begin{subfigure}{0.9\linewidth}
        \centering
        \def\svgwidth{\linewidth}
        \larger[1.5]{\begingroup%
  \makeatletter%
  \providecommand\color[2][]{%
    \errmessage{(Inkscape) Color is used for the text in Inkscape, but the package 'color.sty' is not loaded}%
    \renewcommand\color[2][]{}%
  }%
  \providecommand\transparent[1]{%
    \errmessage{(Inkscape) Transparency is used (non-zero) for the text in Inkscape, but the package 'transparent.sty' is not loaded}%
    \renewcommand\transparent[1]{}%
  }%
  \providecommand\rotatebox[2]{#2}%
  \newcommand*\fsize{\dimexpr\f@size pt\relax}%
  \newcommand*\lineheight[1]{\fontsize{\fsize}{#1\fsize}\selectfont}%
  \ifx\svgwidth\undefined%
    \setlength{\unitlength}{195.07769727bp}%
    \ifx\svgscale\undefined%
      \relax%
    \else%
      \setlength{\unitlength}{\unitlength * \real{\svgscale}}%
    \fi%
  \else%
    \setlength{\unitlength}{\svgwidth}%
  \fi%
  \global\let\svgwidth\undefined%
  \global\let\svgscale\undefined%
  \makeatother%
  \begin{picture}(1,0.37717084)%
    \lineheight{1}%
    \setlength\tabcolsep{0pt}%
    \put(0,0){\includegraphics[width=\unitlength,page=1]{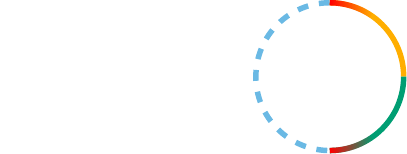}}%
    \put(0.80898172,0.31422176){\color[rgb]{1,0,0}\makebox(0,0)[lt]{\lineheight{1.25}\smash{\begin{tabular}[t]{l}$\Omega_b$\end{tabular}}}}%
    \put(0.81027703,0.03210221){\color[rgb]{1,0,0}\makebox(0,0)[lt]{\lineheight{1.25}\smash{\begin{tabular}[t]{l}$\Omega_c$\end{tabular}}}}%
    \put(0,0){\includegraphics[width=\unitlength,page=2]{tripartite_entanglement.pdf}}%
    \put(0.04923948,0.17378514){\color[rgb]{0.33333333,0.68627451,0.87843137}\makebox(0,0)[t]{\lineheight{1.25}\smash{\begin{tabular}[t]{c}$A$\end{tabular}}}}%
    \put(0.29617886,0.26039712){\color[rgb]{0.90196078,0.62352941,0}\makebox(0,0)[t]{\lineheight{1.25}\smash{\begin{tabular}[t]{c}$B$\end{tabular}}}}%
    \put(0.29617886,0.09011787){\color[rgb]{0,0.61960784,0.45098039}\makebox(0,0)[t]{\lineheight{1.25}\smash{\begin{tabular}[t]{c}$C$\end{tabular}}}}%
    \put(0,0){\includegraphics[width=\unitlength,page=3]{tripartite_entanglement.pdf}}%
    \put(0.49957203,0.2014786){\makebox(0,0)[t]{\lineheight{1.25}\smash{\begin{tabular}[t]{c}Restrict\end{tabular}}}}%
  \end{picture}%
\endgroup%
}
        \caption{If $A$, $B$ and $C$ are connected, we can restrict the symmetry to the complement of $A$ and conclude $B$ and $C$ are entangled if the symmetry is anomalous.}
        \label{fig:tripartite_circle_connected}
    \end{subfigure}
    \hfill
    \begin{subfigure}{0.9\linewidth}
        \centering
        \def\svgwidth{\linewidth}
        \larger[1.5]{\begingroup%
  \makeatletter%
  \providecommand\color[2][]{%
    \errmessage{(Inkscape) Color is used for the text in Inkscape, but the package 'color.sty' is not loaded}%
    \renewcommand\color[2][]{}%
  }%
  \providecommand\transparent[1]{%
    \errmessage{(Inkscape) Transparency is used (non-zero) for the text in Inkscape, but the package 'transparent.sty' is not loaded}%
    \renewcommand\transparent[1]{}%
  }%
  \providecommand\rotatebox[2]{#2}%
  \newcommand*\fsize{\dimexpr\f@size pt\relax}%
  \newcommand*\lineheight[1]{\fontsize{\fsize}{#1\fsize}\selectfont}%
  \ifx\svgwidth\undefined%
    \setlength{\unitlength}{195.07769727bp}%
    \ifx\svgscale\undefined%
      \relax%
    \else%
      \setlength{\unitlength}{\unitlength * \real{\svgscale}}%
    \fi%
  \else%
    \setlength{\unitlength}{\svgwidth}%
  \fi%
  \global\let\svgwidth\undefined%
  \global\let\svgscale\undefined%
  \makeatother%
  \begin{picture}(1,0.37717084)%
    \lineheight{1}%
    \setlength\tabcolsep{0pt}%
    \put(0,0){\includegraphics[width=\unitlength,page=1]{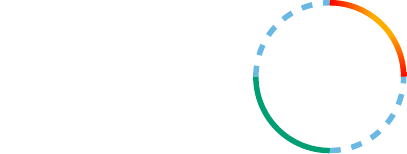}}%
    \put(0.07960658,0.26039712){\color[rgb]{0.33333333,0.68627451,0.87843137}\makebox(0,0)[t]{\lineheight{1.25}\smash{\begin{tabular}[t]{c}$A$\end{tabular}}}}%
    \put(0.29617886,0.26039712){\color[rgb]{0.90196078,0.62352941,0}\makebox(0,0)[t]{\lineheight{1.25}\smash{\begin{tabular}[t]{c}$B$\end{tabular}}}}%
    \put(0,0){\includegraphics[width=\unitlength,page=2]{tripartite_entanglement_disconnected.pdf}}%
    \put(0.49957203,0.2014786){\makebox(0,0)[t]{\lineheight{1.25}\smash{\begin{tabular}[t]{c}Restrict\end{tabular}}}}%
    \put(0,0){\includegraphics[width=\unitlength,page=3]{tripartite_entanglement_disconnected.pdf}}%
    \put(0.08087938,0.09011787){\color[rgb]{0,0.61960784,0.45098039}\makebox(0,0)[t]{\lineheight{1.25}\smash{\begin{tabular}[t]{c}$C$\end{tabular}}}}%
    \put(0,0){\includegraphics[width=\unitlength,page=4]{tripartite_entanglement_disconnected.pdf}}%
    \put(0.29490606,0.09011787){\color[rgb]{0.33333333,0.68627451,0.87843137}\makebox(0,0)[t]{\lineheight{1.25}\smash{\begin{tabular}[t]{c}$A$\end{tabular}}}}%
    \put(0,0){\includegraphics[width=\unitlength,page=5]{tripartite_entanglement_disconnected.pdf}}%
  \end{picture}%
\endgroup%
}
        \caption{If $A$ is disconnected, $B$ and $C$ are not necessarily entangled. In this case, two boundary operators appear for each interval in the complement of $A$. Since the anomaly comes from the phase freedom of only one of them, $B$ and $C$ might not be entangled.}
        \label{fig:tripartite_circle_disconnected}
    \end{subfigure}
    \caption{The anomaly-nonseparability connection on a tripartite ring depends crucially if there is an intersection point between an unique set of regions.}
    \label{fig:tripartite_circles}
\end{figure}

\begin{figure*}[t!]
    \centering
    \def\svgwidth{\textwidth}
    \begingroup%
  \makeatletter%
  \providecommand\color[2][]{%
    \errmessage{(Inkscape) Color is used for the text in Inkscape, but the package 'color.sty' is not loaded}%
    \renewcommand\color[2][]{}%
  }%
  \providecommand\transparent[1]{%
    \errmessage{(Inkscape) Transparency is used (non-zero) for the text in Inkscape, but the package 'transparent.sty' is not loaded}%
    \renewcommand\transparent[1]{}%
  }%
  \providecommand\rotatebox[2]{#2}%
  \newcommand*\fsize{\dimexpr\f@size pt\relax}%
  \newcommand*\lineheight[1]{\fontsize{\fsize}{#1\fsize}\selectfont}%
  \ifx\svgwidth\undefined%
    \setlength{\unitlength}{430.36813354bp}%
    \ifx\svgscale\undefined%
      \relax%
    \else%
      \setlength{\unitlength}{\unitlength * \real{\svgscale}}%
    \fi%
  \else%
    \setlength{\unitlength}{\svgwidth}%
  \fi%
  \global\let\svgwidth\undefined%
  \global\let\svgscale\undefined%
  \makeatother%
  \begin{picture}(1,0.15112757)%
    \lineheight{1}%
    \setlength\tabcolsep{0pt}%
    \put(0,0){\includegraphics[width=\unitlength,page=1]{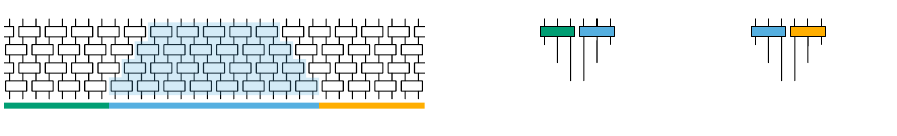}}%
    \put(0.23129494,0.00657818){\color[rgb]{0.33333333,0.68627451,0.87843137}\makebox(0,0)[t]{\lineheight{1.25}\smash{\begin{tabular}[t]{c}$A$\end{tabular}}}}%
    \put(0.41469862,0.00686753){\color[rgb]{1,0.68235294,0}\makebox(0,0)[t]{\lineheight{1.25}\smash{\begin{tabular}[t]{c}$B$\end{tabular}}}}%
    \put(0.12066298,0.13697148){\color[rgb]{0.39607843,0.39607843,0.39607843}\makebox(0,0)[t]{\lineheight{1.25}\smash{\begin{tabular}[t]{c}$u_{AC}$\end{tabular}}}}%
    \put(0.06259926,0.00686753){\color[rgb]{0,0.61960784,0.45098039}\makebox(0,0)[t]{\lineheight{1.25}\smash{\begin{tabular}[t]{c}$C$\end{tabular}}}}%
    \put(0.35692277,0.13697148){\color[rgb]{0.39607843,0.39607843,0.39607843}\makebox(0,0)[t]{\lineheight{1.25}\smash{\begin{tabular}[t]{c}$u_{AB}$\end{tabular}}}}%
    \put(0,0){\includegraphics[width=\unitlength,page=2]{local_unitary_surgery_2.pdf}}%
    \put(0.23129494,0.13697148){\color[rgb]{0.33333333,0.68627451,0.87843137}\makebox(0,0)[t]{\lineheight{1.25}\smash{\begin{tabular}[t]{c}$\tilde{U}_A$\end{tabular}}}}%
    \put(0.04125682,0.13697148){\color[rgb]{0,0,0}\makebox(0,0)[t]{\lineheight{1.25}\smash{\begin{tabular}[t]{c}$\tilde{U}_{BC}$\end{tabular}}}}%
    \put(0.43806242,0.13551129){\color[rgb]{0,0,0}\makebox(0,0)[t]{\lineheight{1.25}\smash{\begin{tabular}[t]{c}$\tilde{U}_{BC}$\end{tabular}}}}%
    \put(0,0){\includegraphics[width=\unitlength,page=3]{local_unitary_surgery_2.pdf}}%
    \put(0.621858,0.13559573){\color[rgb]{0,0.61960784,0.45098039}\makebox(0,0)[t]{\lineheight{1.25}\smash{\begin{tabular}[t]{c}$w_c$\end{tabular}}}}%
    \put(0.66612864,0.13559573){\color[rgb]{0.33333333,0.68627451,0.87843137}\makebox(0,0)[t]{\lineheight{1.25}\smash{\begin{tabular}[t]{c}$w_a$\end{tabular}}}}%
    \put(0.8569357,0.13559573){\color[rgb]{0.33333333,0.68627451,0.87843137}\makebox(0,0)[t]{\lineheight{1.25}\smash{\begin{tabular}[t]{c}$v_a$\end{tabular}}}}%
    \put(0.90120634,0.13559573){\color[rgb]{1,0.68627451,0}\makebox(0,0)[t]{\lineheight{1.25}\smash{\begin{tabular}[t]{c}$v_b$\end{tabular}}}}%
    \put(0.50077289,0.0833845){\color[rgb]{0,0,0}\makebox(0,0)[t]{\lineheight{1.25}\smash{\begin{tabular}[t]{c}$=$\end{tabular}}}}%
    \put(0,0){\includegraphics[width=\unitlength,page=4]{local_unitary_surgery_2.pdf}}%
  \end{picture}%
\endgroup%

    \caption{Decomposition of $U = (u_{AB} \otimes u_{AC}) (\tilde{U}_A \otimes \tilde{U}_{BC})$ into unitaries $u_{AB}$ and $u_{AC}$ that straddle the boundaries, and unitaries $\tilde{U}_A$ and $\tilde{U}_{BC}$, entirely supported in $A$ and $BC$, respectively. Lemma \ref{lemma:separable_unitary} enables us to ``cut'' $U$ into separate pieces acting independently on each region.}
    \label{fig:local_unitary_surgery}
\end{figure*}

We partition the system (a ring) into three disjoint intervals (See Fig.~\ref{fig:tripartite_circle_connected}). In order to reduce the symmetry and extract the anomalous phase $\omega$ out of one boundary operator, we impose two requirements on the tripartition. Firstly, each region must be larger than the maximum depth of the quantum circuits defining the symmetry operators $U(g)$. This guarantees we can reduce the symmetry into the subregions. Secondly, there must exist two regions whose shared boundary is exactly one point. Otherwise, we are not able to extract the anomalous phase, as will be clear from the proof below (See Fig.~\ref{fig:tripartite_circle_disconnected} for an example)\footnote{For example, it is easy to check that some tensor products of GHZ-like states in the disconnected configuration of Fig.~\ref{fig:tripartite_circle_disconnected} are symmetric under the CZX symmetry.}. In (1+1)-D, one natural choice for which the second condition is valid is when $A$, $B$ and $C$ are connected, so we will be using this one in what follows.

Under the preceding conditions, suppose there is a tripartite-separable pure state $\ket{\psi} = \ket{A}\ket{B}\ket{C}$ that is symmetric under $U$:
\begin{equation}\label{eq:symmetry_condition}
    U(g) \ket{A}\ket{B}\ket{C} = \lambda(g) \ket{A}\ket{B}\ket{C}.
\end{equation}
We will prove that this leads to a contradiction if $[\omega]$ is nontrivial.

The anomaly of the symmetry is manifest in its reductions to subregions with boundaries. Thus, we will now construct one such reduction $U_{BC}(g)$ to region $BC$ that acts in the same way as $U(g)$ away from the boundaries of $BC$, and, importantly, has $\ket{B}\ket{C}$ as an eigenvector.

For now, we omit the dependence of the symmetry on group elements $g$ for simplicity of notation. First, we group the gates of $U$ near the intersection of regions $A$ and $B$ into a separate local unitary $u_{AB}$, so that $u_{AB}^\dagger U$ has no gate acting in $A$ and $B$ at the same time. We do the same for $A$ and $C$ with $u_{AC}$ (See Fig.~\ref{fig:local_unitary_surgery} for an explicit decomposition.). In this way, we can separate $U$ into $U = (u_{AB} \otimes u_{AC}) (\tilde{U}_A \otimes \tilde{U}_{BC})$, with $\tilde{U}_A$ ($\tilde{U}_{BC}$) supported on $A$ ($BC$) and acting as $U$ in the bulk of its support. 

For the moment, let us focus on the region $AB$. {Although $u_{AB}$ can entangle states in $AB$, equation \eqref{eq:symmetry_condition} suggests that it does not entangle the state $\tilde{U}_A\ket{A} \otimes \tilde{U}_B\ket{B}$, since the final state is still separable with respect to $AB$ after applying all three unitaries. We show in Appendix \ref{appendix:uABproof} this is indeed the case\footnote{We thank Chaoming Jian for pointing out the necessity of having such a proof.}, i.e. $u_{AB}$ does not entangle $\tilde{U}_A\ket{A} \otimes \tilde{U}_B\ket{B}$, if the third region $C$ is also included in the argument. Then we can replace $u_{AB}$ by the tensor product of two unitaries, $v_A \otimes v_B$, each one supported either in $A$ or in $B$.} This is guaranteed by the following simple lemma:

\begin{lemma}\label{lemma:separable_unitary}
    Given pure states $\ket{\psi_{12}}, \ket{\phi_{12}} \in \Hilb_1 \otimes \Hilb_2$, $\ket{\psi_{34}}, \ket{\phi_{34}} \in \Hilb_3 \otimes \Hilb_4$ and an unitary $U_{23} \in U(\Hilb_2 \otimes \Hilb_3)$ satisfying $U_{23} \ket{\psi_{12}} \ket{\psi_{34}} = \ket{\phi_{12}} \ket{\phi_{34}}$, then there exist unitaries $V_2 \in U(\Hilb_2)$ and $V_3 \in U(\Hilb_3)$ satisfying $\ket{\phi_{12}} = V_2 \ket{\psi_{12}}$ and $\ket{\phi_{34}} = V_3 \ket{\psi_{34}}$. In diagrammatic form,
    \begin{equation}
        \begin{tikzpicture}
	\begin{pgfonlayer}{nodelayer}
		\node [style=none] (0) at (-1, 0.5) {};
		\node [style=none] (1) at (1, 0.5) {};
		\node [style=none] (2) at (0, -1) {};
		\node [style=none] (3) at (0, 0) {$\psi_{12}$};
		\node [style=none] (4) at (1.5, 0.5) {};
		\node [style=none] (5) at (3.5, 0.5) {};
		\node [style=none] (6) at (2.5, -1) {};
		\node [style=none] (7) at (2.5, 0) {$\psi_{34}$};
		\node [style=none] (8) at (-0.5, 0.5) {};
		\node [style=none] (9) at (0.5, 0.5) {};
		\node [style=none] (10) at (0, 1) {};
		\node [style=none] (11) at (0, 2) {};
		\node [style=none] (14) at (0.5, 1) {};
		\node [style=none] (16) at (0.5, 2) {};
		\node [style=none] (19) at (0.5, 2.5) {};
		\node [style=none] (20) at (-0.5, 2.5) {};
		\node [style=none] (21) at (2.5, 1) {};
		\node [style=none] (22) at (2.5, 2) {};
		\node [style=none] (23) at (2, 0.5) {};
		\node [style=none] (24) at (2, 1) {};
		\node [style=none] (25) at (2, 2) {};
		\node [style=none] (26) at (2, 2.5) {};
		\node [style=none] (27) at (3, 2.5) {};
		\node [style=none] (28) at (3, 0.5) {};
		\node [style=none] (29) at (1.25, 1.5) {};
		\node [style=none] (30) at (1.25, 1.5) {};
		\node [style=none] (31) at (1.25, 1.5) {$U_{23}$};
		\node [style=none] (32) at (4.5, 0.5) {};
		\node [style=none] (33) at (6.5, 0.5) {};
		\node [style=none] (34) at (5.5, -1) {};
		\node [style=none] (35) at (5.5, 0) {$\phi_{12}$};
		\node [style=none] (36) at (7, 0.5) {};
		\node [style=none] (37) at (9, 0.5) {};
		\node [style=none] (38) at (8, -1) {};
		\node [style=none] (39) at (8, 0) {$\phi_{34}$};
		\node [style=none] (40) at (5, 0.5) {};
		\node [style=none] (41) at (6, 0.5) {};
		\node [style=none] (46) at (6, 2.5) {};
		\node [style=none] (47) at (5, 2.5) {};
		\node [style=none] (50) at (7.5, 0.5) {};
		\node [style=none] (53) at (7.5, 2.5) {};
		\node [style=none] (54) at (8.5, 2.5) {};
		\node [style=none] (55) at (8.5, 0.5) {};
		\node [style=none] (56) at (10, 0.5) {};
		\node [style=none] (57) at (12, 0.5) {};
		\node [style=none] (58) at (11, -1) {};
		\node [style=none] (59) at (11, 0) {$\psi_{12}$};
		\node [style=none] (60) at (12.5, 0.5) {};
		\node [style=none] (61) at (14.5, 0.5) {};
		\node [style=none] (62) at (13.5, -1) {};
		\node [style=none] (63) at (13.5, 0) {$\psi_{34}$};
		\node [style=none] (64) at (10.5, 0.5) {};
		\node [style=none] (65) at (11.5, 0.5) {};
		\node [style=none] (66) at (11.5, 2.5) {};
		\node [style=none] (67) at (10.5, 2.5) {};
		\node [style=none] (68) at (13, 0.5) {};
		\node [style=none] (69) at (13, 2.5) {};
		\node [style=none] (70) at (14, 2.5) {};
		\node [style=none] (71) at (14, 0.5) {};
		\node [style=none] (72) at (11, 2) {};
		\node [style=none] (73) at (11, 1) {};
		\node [style=none] (74) at (12, 1) {};
		\node [style=none] (75) at (12, 2) {};
		\node [style=none] (76) at (13.5, 2) {};
		\node [style=none] (77) at (13.5, 1) {};
		\node [style=none] (78) at (12.5, 1) {};
		\node [style=none] (79) at (12.5, 2) {};
		\node [style=none] (80) at (11.5, 2) {};
		\node [style=none] (81) at (11.5, 1) {};
		\node [style=none] (82) at (13, 1) {};
		\node [style=none] (83) at (13, 2) {};
		\node [style=none] (84) at (11.5, 1.5) {};
		\node [style=none] (85) at (13, 1.5) {};
		\node [style=none] (86) at (11.5, 1.5) {$V_2$};
		\node [style=none] (87) at (13, 1.5) {$V_3$};
		\node [style=none] (88) at (4, 1) {};
		\node [style=none] (89) at (9.5, 1) {};
		\node [style=none] (90) at (4, 1) {$=$};
		\node [style=none] (91) at (9.5, 1) {};
		\node [style=none] (92) at (9.5, 1) {$\stackrel{!}{=}$};
	\end{pgfonlayer}
	\begin{pgfonlayer}{edgelayer}
		\draw [style=new edge style 0] (2.center)
			 to (1.center)
			 to (0.center)
			 to cycle;
		\draw [style=new edge style 0] (6.center)
			 to (5.center)
			 to (4.center)
			 to cycle;
		\draw [style=new edge style 0] (21.center)
			 to (22.center)
			 to (11.center)
			 to (10.center)
			 to cycle;
		\draw [style=new edge style 0] (9.center) to (14.center);
		\draw [style=new edge style 0] (23.center) to (24.center);
		\draw [style=new edge style 0] (25.center) to (26.center);
		\draw [style=new edge style 0] (16.center) to (19.center);
		\draw [style=new edge style 0] (8.center) to (20.center);
		\draw [style=new edge style 0] (28.center) to (27.center);
		\draw [style=new edge style 0] (34.center)
			 to (33.center)
			 to (32.center)
			 to cycle;
		\draw [style=new edge style 0] (38.center) to (37.center);
		\draw [style=new edge style 0] (37.center) to (36.center);
		\draw [style=new edge style 0] (36.center) to (38.center);
		\draw [style=new edge style 0] (40.center) to (47.center);
		\draw [style=new edge style 0] (55.center) to (54.center);
		\draw [style=new edge style 0] (41.center) to (46.center);
		\draw [style=new edge style 0] (50.center) to (53.center);
		\draw [style=new edge style 0] (58.center) to (57.center);
		\draw [style=new edge style 0] (57.center) to (56.center);
		\draw [style=new edge style 0] (56.center) to (58.center);
		\draw [style=new edge style 0] (62.center) to (61.center);
		\draw [style=new edge style 0] (61.center) to (60.center);
		\draw [style=new edge style 0] (60.center) to (62.center);
		\draw [style=new edge style 0] (64.center) to (67.center);
		\draw [style=new edge style 0] (71.center) to (70.center);
		\draw [style=new edge style 0] (74.center)
			 to (73.center)
			 to (72.center)
			 to (75.center)
			 to cycle;
		\draw [style=new edge style 0] (79.center)
			 to (78.center)
			 to (77.center)
			 to (76.center)
			 to cycle;
		\draw [style=new edge style 0] (80.center) to (66.center);
		\draw [style=new edge style 0] (83.center) to (69.center);
		\draw [style=new edge style 0] (81.center) to (65.center);
		\draw [style=new edge style 0] (82.center) to (68.center);
	\end{pgfonlayer}
\end{tikzpicture}

    \end{equation}
\end{lemma}
\begin{proof}
    Tracing out subsystems 2, 3 and 4 gives two pure states reducing to the same density matrix in $\Hilb_1$: $\Tr_2[\ketbra{\psi_{12}}{\psi_{12}}] = \Tr_2[\ketbra{\phi_{12}}{\phi_{12}}]$. Hence, by the unitary equivalence of purifications~\cite{watrous_theory_2018a}, there exists a unitary $V_2$ such that $\ket{\phi_{12}} = V_2 \ket{\psi_{12}}$. $V_3$ can be analogously defined by tracing out subsystems 1, 2 and 3.
\end{proof}

Applying {the same reasoning} to region $AC$, we arrive at the similar conclusion that $u_{AC}$ can be replaced by $w_a \otimes w_c$. Finally, we can define $U_{BC} \defeq (v_b \otimes w_c) \tilde{U}_{BC}$ satisfying
\begin{equation}\label{eq:symmetry_condition_restricted}
    U_{BC}(g) \ket{B}\ket{C} = \lambda(g) \ket{B}\ket{C},
\end{equation}
after tracing out region $A$ from \eqref{eq:symmetry_condition} (See Fig.~\ref{fig:local_unitary_surgery}). For the simple CZX example in Sec.~\ref{sec:4-qubit_CZX}, this restricted unitary $U_{BC}$ can be explicitly constructed via Eq.~\eqref{eq:CZXUBC}.

The reduced operator $U_{BC}(g)$ is still a finite-depth local unitary, since the supports of $v_b(g)$ and $w_c(g)$ have size of the same order as the depth $D = O(1)$ of the original $U(g)$. For the same reason, we know $U_{BC}(g)$ acts in the same way as $U(g)$ in the interior of the $BC$ region, $O(D)$ sites away from its boundary points.  In particular, this means that $U_{BC}$ is an unitary representation of $G$ up to boundary operators $\Omega(g_1, g_2) = \Omega_b(g_1, g_2) \otimes \Omega_c(g_1, g_2)$, since $U(g)$ is a representation: $U(g_1) U(g_2) = U(g_1 g_2)$. More precisely, we have
\begin{equation}\label{eq:restricted_boundary_op}
        U_{BC}(g_1) U_{BC}(g_2) = \Omega(g_1, g_2) U_{BC}(g_1 g_2).
\end{equation}
Applying this equation to $\ket{B}\ket{C}$ and using \eqref{eq:symmetry_condition_restricted}, we have
\begin{equation}\label{eq:boundary_op_eigenstates}
\begin{split}
    \Omega_b(g_1, g_2) \ket{B} & = \mu_b(g_1, g_2) \ket{B} \\
    \Omega_c(g_1, g_2) \ket{C} & = \mu_c(g_1, g_2) \ket{C}
\end{split}
\end{equation}
with $\mu_b \cdot \mu_c(g_1, g_2) = \delta\lambda(g_1, g_2) \equiv \lambda(g_1)\lambda(g_2)/\lambda(g_1 g_2)$, where $\delta$ is the coboundary operator of the second group cohomology $H^2(G, U(1))$.

Equation \eqref{eq:boundary_op_eigenstates} means that we can assign definite phases to $\Omega_b$ and $\Omega_c$, at least when acting on states $\ket{B}$ and $\ket{C}$, respectively. Intuitively, the 3-cocycle $\omega$ measures the obstruction to remove phase factors from $\Omega_b$ and $\Omega_c$ for all states they can act on. Hence, we expect the SPT state is trivial, with $[\omega] = 1$. Indeed, associativity for the boundary operators $\Omega_c(g_1, g_2)$ imply
\begin{equation}\label{eq:associativity_reduced_boundary_ops}
\begin{split}
    &\omega(g_{1},g_{2},g_{3}) \Omega_{c}(g_{1},g_{2})\Omega_{c}(g_{1}g_{2},g_{3}) \\ & = U_{BC}(g_{1})\Omega_{c}(g_{2},g_{3}) U_{BC}(g_{1})^{-1} \Omega_{c}(g_{1},g_{2}g_{3}).
\end{split}
\end{equation}
Applying this equation to {$\ket{B}\ket{C}$} gives $\omega = \delta \mu_c$, a coboundary. Thus, the group cohomology class $[\omega] \in H^3(G, U(1))$ is trivial.

\subsection{Long-range entanglement}
\label{sec:long-range_entanglement}

In the previous section, we proved that no tripartite separable state $\ket{A}\ket{B}\ket{C}$ can be symmetric under an anomalous symmetry in (1+1)-D, no matter how entangled each of the three states are in their internal degrees of freedom. Although this space of tripartite-separable states is very large, it is fragile to any local perturbation on the boundary between any two of the three regions. This stands in contrast with the well known result that pure symmetric states under an anomalous symmetry are not short-range entangled~\cite{CZX,else_classifying_2014,KapustinSopenko2024}. In this section, we unify these two results by proving that not only is any tripartite separable mixed state not symmetric under an anomalous symmetry, but also any other state that can be prepared from it by a finite-depth local channel (FDLC).

More specifically, by finite-depth local channel we mean a completely positive trace-preserving map $\E_{\text{FDLC}}$ that is the restriction of a finite-depth local unitary acting over an enlarged Hilbert space $\Hilb_\sys \otimes \Hilb_\anc$, where $\Hilb_\sys$ is the original, ``physical'' Hilbert space, and $\Hilb_\anc$ the ancilla space. We enlarge it by adding to each on-site Hilbert space a finite-dimensional auxiliary Hilbert space, geometrically positioned as close to the original physical site as the other neighboring physical sites. Furthermore, the auxiliary degrees of freedom start in some trivial product state, such as $\ket{00\cdots 0}_\anc$. Following previous works~\cite{hastings_topological_2011, coser_classification_2019a, ma_average_2023, Ma_intrisic_2023, sang_mixedstate_2023}, we use finite-depth local channels as the extension of finite-depth local unitaries to mixed states.

Suppose, by contradiction, that $\rho = \E_{\text{FDLC}}(\rho_3)$ is a strongly symmetric state under $U(g)$, where $\rho_3 = \sum_{i} p_i [A_i B_i C_i]$ is tripartite separable, where $[\psi] = \ketbra{\psi}{\psi}$. By definition, there exists a local unitary $V$ acting on an enlarged Hilbert space $\Hilb_\sys \otimes \Hilb_\anc$ such that $\rho = \sum_i p_i \Tr_E[V \ket{\tilde A_i} \ket{\tilde B_i} \ket{\tilde C_i}]$, where $\ket{\tilde{A}_i} = \ket{A_i}_\sys \ket{0}_\anc$ and similarly for $\ket{\tilde{B}_i}$ and $\ket{\tilde{C}_i}$. By virtue of the strong symmetry condition, Lemma \ref{lemma:strongly_symmetric_decomposition} and its corollary imply that not only each mixed state in the previous mixture is strongly symmetric, but also that each enlarged pure state $V \ket{\tilde A_i} \ket{\tilde B_i} \ket{\tilde C_i}$ is symmetric under the trivially enlarged symmetry $\tilde{U}(g) = U(g)_\sys \otimes \one_\anc$. Equivalently, the state $\ket{\tilde A_i} \ket{\tilde B_i} \ket{\tilde C_i}$ is symmetric under $U'(g) \defeq V^\dagger \tilde{U}(g) V$.

With this conjugated symmetry representation, we can follow the same steps that reached equation \ref{eq:associativity_reduced_boundary_ops} and identify a trivial group cohomology phase $[\omega'] \in H^3(G, U(1))$. Even though the symmetry representation $U'(g)$ is different, we expect that this group cohomology class still corresponds to the same bulk SPT phase, since we the anomaly should be stable under finite-depth local unitaries. To make this precise, we prove in Appendix~\ref{appendix:LRE_proof}, using techniques developed in Ref.~\cite{else_classifying_2014}, that $\omega'$ is in the same cohomology class as $\omega$ of the original symmetry, appearing in \eqref{eq:associativity_reduced_boundary_ops}.

For mixed quantum states, a natural notion of ``phases of matter'' is based on the two-way connectivity via finite-depth quantum channels~\cite{hastings_topological_2011, coser_classification_2019a, ma_average_2023, Ma_intrisic_2023, sang_mixedstate_2023}: two mixed states $\rho_{1}$ and $\rho_2$ belong to the same phase if, and only if, there are two finite-depth quantum channels $\E$ and $\E'$ such that $\E[\rho_1]=\rho_2$ and $\E'[\rho_2]=\rho_1$. Here we do not impose any symmetry restriction on $\E$ and $\E'$ (if we did, the phases would be symmetry-protected). From this point of view, our result on long-range tripartite entanglement can be rephrased as: any state $\rho$ that is strongly symmetric under an anomalous symmetry $G$ must belong to a nontrivial phase, in which any state must be tripartite entangled.

Let us now revisit the infinite-temperature example in Sec.~\ref{sec:4-qubit_CZX}: $\rho_{\infty,+}\propto\one+U_{\rm CZX}$. From our discussion here we see that $\rho_{\infty,+}$ is long-range (tripartite) entangled~\footnote{We note that Ref.~\cite{Alberto2022} conjectured that states like $\rho_{\infty,+}$ are long-ranged entangled when the strong symmetry is anomalous.}. The long-range entanglement is also ``non-invertible'', in the sense that there is no ``inverse'' state $\tilde{\rho}$ such that $\rho_{\infty,+}\otimes\tilde{\rho}$ can be prepared from a trivial state in finite depth -- otherwise we immediately obtain a finite-depth preparation of $\rho_{\infty,+}$ by tracing out $\tilde{\rho}$. Furthermore, any $k$-point ($k$ finite) connected correlation function $\langle O(x_1)...O(x_k)\rangle_c$ vanishes exponentially. In fact, any region $A$ only shares an ``area law'' correlation with its complement $\bar{A}$, in the sense that the mutual information $I_{A\bar{A}}=\log 2$. Furthermore, for any two regions $A$ and $B$, as long as $B\neq \bar{A}$ we have $I_{AB}=0$. These features, in the pure state context, are the defining characteristics of intrinsic topological orders, which only exist in two space dimensions or above. From this point of view, states like $\rho_{\infty,+}$ can be viewed as certain type of intrinsic (or non-invertible) topological phases in (1+1)-D, which are possible only for mixed states.

\subsection{\texorpdfstring{$\rho_{\infty,\pm}$}{ρ∞±} as an intrinsically mixed quantum phase}
\label{sec:intrinsicallymixed}
We now recall two key properties of $\rho_{\infty,\pm}=(\one \pm U_{\rm CZX})/2^L$:
\begin{enumerate}
    \item $\rho_{\infty,\pm}$ is long-range tripartite entangled, meaning $\rho_{\infty,\pm} \neq \E_{\rm FD}(\rho_{\rm tri-sep})$.
    \item $\rho_{\infty,\pm}$ is bipartite separable for any (connected) bipartition of the system.
\end{enumerate}
It turns out these two properties guarantee that $\rho_{\infty,\pm}$ belongs to an \textit{intrinsically mixed quantum phase}, in the sense that it is not two-way connected to any pure state through FD local channels. To see this, suppose there are two pure states $|\Psi_{1,2}\rangle$ and FD local channels $\E_{1,2}$ such that
\begin{equation}
\label{eq:IM}
    |\Psi_1\rangle \xrightarrow{\E_1} \rho_{\infty,\pm} \xrightarrow{\E_2} |\Psi_2\rangle.
\end{equation}
Then the long-range tripartite entanglement of $\rho_{\infty,\pm}$ guarantees that $|\Psi_1\rangle$ must be long-range entangled. The bipartite separability of $\rho_{\infty,\pm}$, on the other hand, requires $|\Psi_2\rangle$ to be short-range entangled. To see this, notice that for any decomposition $\rho_{\infty,\pm}=\sum_{i}p_i|\psi_i\rangle\langle\psi_i|$, we have $\E_2: |\psi_i\rangle \mapsto |\Psi_2\rangle$. Since we can choose $|\psi_i\rangle$ to be bipartite separable for a interval bipartition $A|\bar{A}$, $|\Psi_2\rangle$ must also be bipartite separable for $A|\bar{A}$, up to a finite-depth local unitary (denoted as $U_{\partial A}$) acting near the boundary $\partial A$, with width $W$ proportional to the depth of the channel $\E_2$. Since this is true for any bipartition $A|\bar{A}$, we can have $O(L)$ number of such disentangling unitaries $U_{\partial A_i}$, each separated by a finite distance $l\gg W$ from the others. Upon acting with $\prod_iU_{\partial A_i}$, $|\Psi_2\rangle$ becomes a product state with enlarged unit cells of size $l$ -- in other words, $|\Psi_2\rangle$ is short-range entangled.  Therefore $|\Psi_1\rangle$ and $|\Psi_2\rangle$ cannot be the same state, meaning $\rho_{\infty,\pm}$ is not two-way connected to any pure state.

We can further consider a more general type of finite-depth local channel: instead of starting with an ancilla $\Hilb_{\mathcal{A}}$ in a pure product state $|00...0\rangle_{\mathcal{A}}$, we allow initiating the ancilla in a convex sum of short-range entangled pure states: $\rho^{(0)}_{\mathcal{A}}=\sum_{j}p_j|\psi^{SRE}_j\rangle\langle\psi^{SRE}_j|_{\mathcal{A}}$. We can then define different states to belong to the same quantum phase if they are two-way connected through such generalized finite-depth local channels. This definition of phases is stronger, in the sense that long-range classical correlations, reflected in the ordinary correlation function such as $\Tr\rho O_1O_2$, become trivializable \cite{hastings_topological_2011}.

It is easy to check that, for the $\rho_{\infty,\pm}$ state, our previous argument can be modified even for such generalized channels~\cite{Lessatoappear}. We again conclude that if
Eq.~\eqref{eq:IM} holds, then $|\Psi_1\rangle$ must be long-range (tripartite) entangled while $|\Psi_2\rangle$ must be short-range entangled, and there is no pure state that is two-way connected to $\rho_{\infty,\pm}$. The state $\rho_{\infty,\pm}$ is thus \textit{intrinsically mixed} in a rather strong sense -- to our knowledge such intrinsically mixed many-body state has not been discussed before.\footnote{We note that Ref.~\cite{WangWuWang2023} discussed a type of ``intrinsic mixed-state quantum topological order (TO)''. The authors proved that such states cannot be two-way connected to any fully separable states, and argued that because ``It is widely believed that non-modular anyon theories cannot be realized by local gapped Hamiltonians in 2D bosonic system'', then ``this implies the lack of a pure-state counterpart of the mixed-state TO; thus it is indeed intrinsically mixed.''. However, currently it is unknown whether such states can be two-way connected to some gapless pure states, especially when the channels are generalized and are able to create/eliminate long-range classical correlations.}

\subsection{Higher dimensions}
\label{sec:higher_dims}

We can now generalize the anomaly-nonseparability connection to higher dimensions. The central statement (Conjecture~\ref{conj:higherdim}) is that a $d$ dimensional bosonic system with an anomalous strong symmetry $G$ is nonseparable for certain $(d+2)$-partition of the system.

To motivate this conjecture, we follow the Else-Nayak argument~\cite{else_classifying_2014} to prove a weaker version, in which we assume a specific form of the symmetry actions. Specifically, we assume each symmetry element $g\in G$ to act as
\begin{equation}
    U(g)=\sum_{\alpha}e^{i\mathcal{N}(g)[\alpha]}|g\alpha\rangle\langle\alpha|,
    \label{eq:fpform}
\end{equation}
where $\{\ket{\alpha} = \otimes_{x \in \Lambda} \ket{\alpha_x} \}_\alpha$ is an on-site basis, $\alpha_x \mapsto g \alpha_x$  is an on-site $g$-action on classical labels $\alpha$, and $\mathcal{N}(g)$ a local functional of $\alpha$. For each anomaly described by the group-cohomology $H^{d+2}(G,U(1))$, there are certain ``fixed-point'' models realizing the boundary symmetry action in this specific form~\cite{else_classifying_2014, chen_symmetryprotected_2012}. The anomalous symmetries described in Sec.~\ref{sec:prototypical_examples} all take the form of Eq.~\eqref{eq:fpform}.

Now consider a symmetric $k$-separable state $\ket{A_1}\ket{A_2}\cdots\ket{A_{k}} \eqdef \ket{A_1} \ket{\comp{A_1}}$, satisfying
\begin{equation}\label{eq:symmetry_condition_higher_dims}
        U(g) \ket{A_1}\ket{\comp{A_1}} = \lambda(g) \ket{A_1}\ket{\comp{A_1}}.
    \end{equation}
We will show there is a restricted FDLU symmetry action $U_{\comp{A_1}}$, where $\comp{A_1} = A_2 \cup \cdots \cup A_{k}$, that acts like $U$ in the bulk of $\comp{A_1}$ and has $\ket{A_2}\cdots\ket{A_{k}}$ as an eigenvector. The construction of $U_{\comp{A_1}}$ is similar to that in Sec.~\ref{sec:4-qubit_CZX}: choose a product state $|\alpha_{A_1}\rangle$ in the classical on-site basis such that $\langle\alpha_{A_1}|A_1\rangle\neq0$, then a restricted local unitary is given by
\begin{equation}
\label{eq:constructrestriction}
    U(g)_{\comp{A_1}}=:\frac{\langle g\alpha_{A_1}|U(g)|A_1\rangle}{\langle\alpha_{A_1}|A_1\rangle}=\sum_{\alpha}e^{i\mathcal{N}_{\comp{A_1}}(g)[\alpha]}|g\alpha\rangle\langle\alpha|,
\end{equation}
where $|\alpha\rangle$ lives in $\comp{A_1}$ and $\mathcal{N}_{\comp{A_1}}(g)[\alpha]=:\mathcal{N}(g)[(g\alpha_A)\alpha]$ is a local functional supported on $\comp{A_1}$ that acts in the same way as $\mathcal{N}$ in the bulk of $\comp{A_1}$. We have thus obtained a local unitary restriction to the subsystem $\comp{A_1}=A_2\cup A_3...\cup A_k$, with the same form as Eq.~\eqref{eq:fpform} but the local functional $\mathcal{N}_{\comp{A_1}}$ is modified near the boundary of $\comp{A_1}$.

{Importantly, the reduced unitary will be a representation up to boundary operators $\Omega(g, h)$. These, in turn, will also be of the form \eqref{eq:fpform}, now with the local functional depending on two group elements, $\mathcal{N}(g,h)$.} We can then iterate the above procedure, restricting the symmetry action to fewer regions and lower dimensions -- each step producing a local functional $\mathcal{N}$ that is modified from the previous one near the boundaries. This collection of local functionals $\{\mathcal{N}\}$ satisfy the symmetry algebra only up to boundary terms, similar to the 0d and 1d cases discussed earlier (See \cite{else_classifying_2014} for details). The violation of the symmetry algebra on the subsystem boundaries is characterized precisely by the group-cohomology $H^{d+2}(G,U(1))$~\cite{else_classifying_2014}.

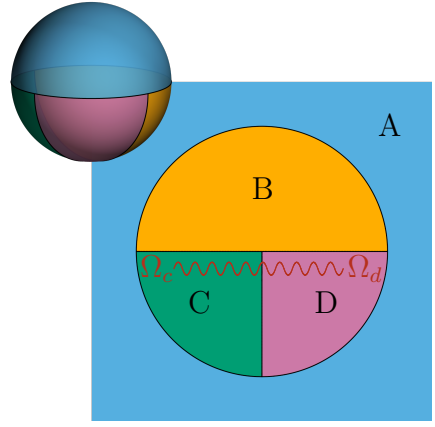
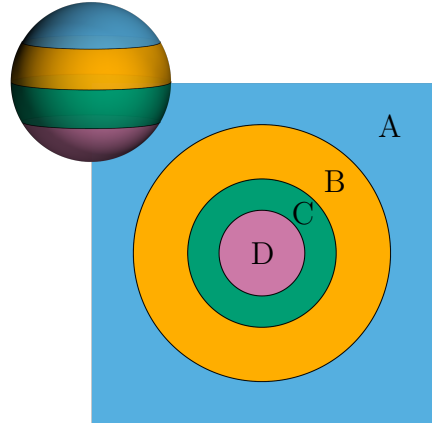
\begin{figure}
    \centering
    \begin{subfigure}{\linewidth}
        \centering
        \begin{overpic}[width=0.65\linewidth]{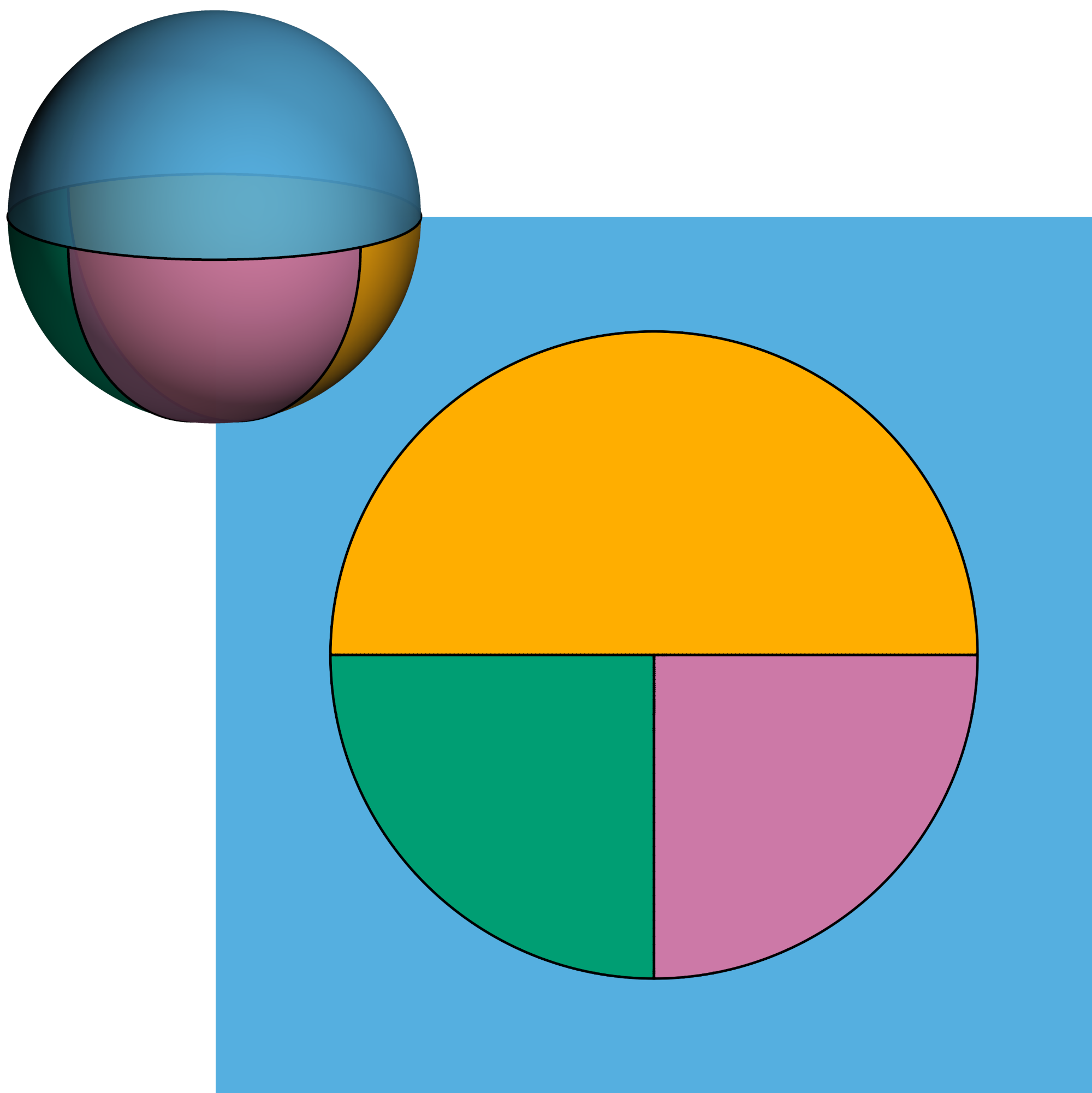}
            \put(90,70){\makebox(0,0){\large A}}
            \put(60,55){\makebox(0,0){\large B}}
            \put(45,28){\makebox(0,0){\large C}}
            \put(75,28){\makebox(0,0){\large D}}
            \put(31,34){\color{BrickRed}\large $\Omega_c$}
            \put(80,34){\color{BrickRed}\large $\Omega_d$}
            \foreach \i in {0,...,9}{%
                \pgfmathtruncatemacro{\n}{39 + 4*\i}
                \pgfmathtruncatemacro{\nn}{41 + 4*\i}
                \put(\n,36){\color{BrickRed}\qbezier(0,0)(1,3)(2,0)}
                \put(\nn,36){\color{BrickRed}\qbezier(0,0)(1,-3)(2,0)}
            }
        \end{overpic}
        \caption{This partition obstructs $4$-partite separability. By reducing the symmetry to the edge of BCD, and then to the points $c$ and $d$, we reach a contradiction if regions C and D are separable.}
        \label{fig:fourpartite-nontrivial}
    \end{subfigure}
    \begin{subfigure}{\linewidth}
        \centering
        \begin{overpic}[width=0.65\linewidth]{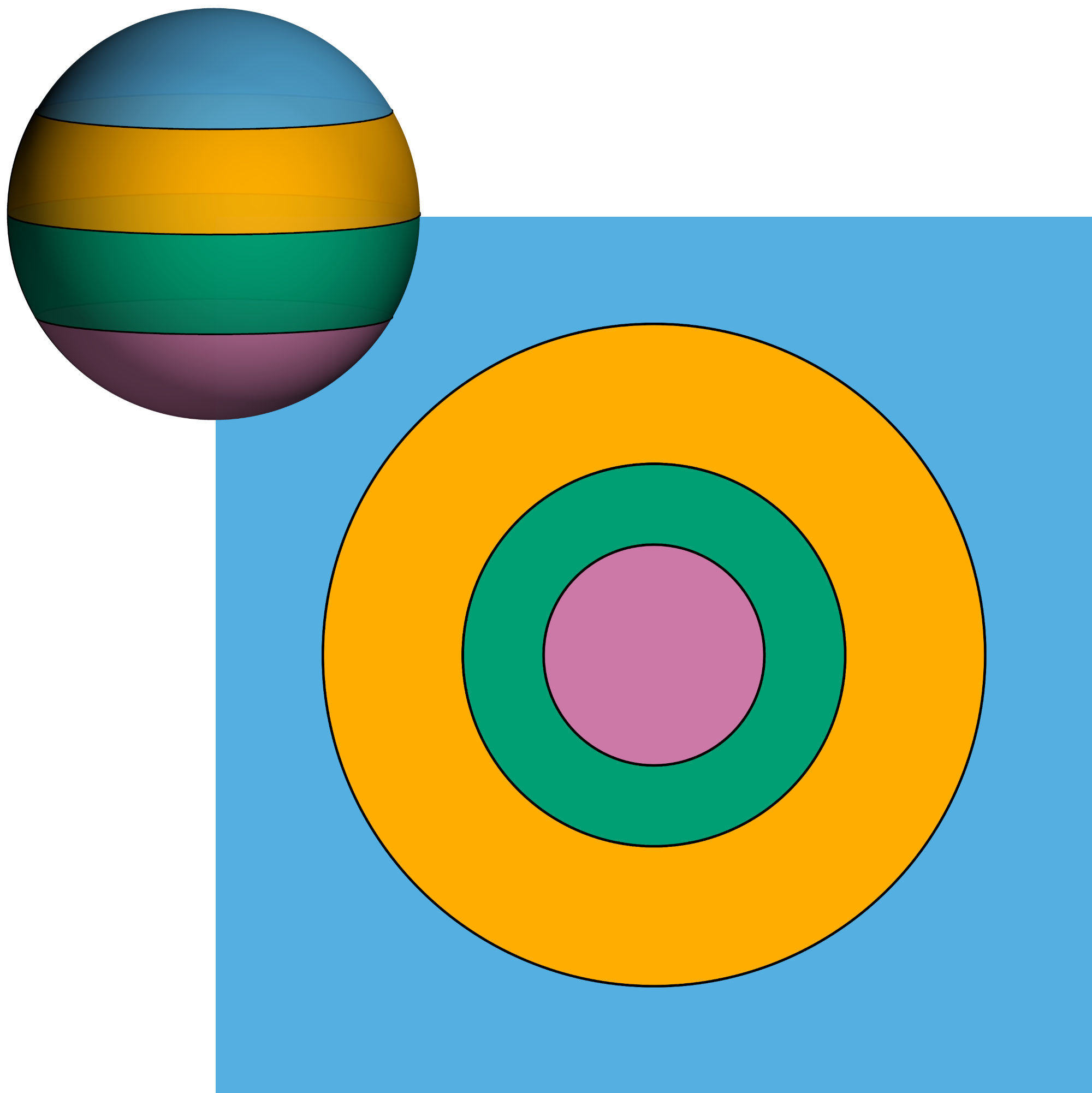}
            \put(90,70){\makebox(0,0){\large A}}
            \put(77,57){\makebox(0,0){\large B}}
            \put(69.5,49.5){\makebox(0,0){\large C}}
            \put(60,40){\makebox(0,0){\large D}}
        \end{overpic}
        \caption{The anomaly-nonseparability connection presented here cannot be applied to this partition, since it does not have any intersection point to which we can reduce the symmetry.}
        \label{fig:fourpartite-trivial}
    \end{subfigure}
    \caption{Distinct four-partitions in (2+1)-D, shown on the sphere and the plane after stereographic projection.}
    \label{fig:fourpartite}
\end{figure}

To illustrate the idea more concretely, we consider a (2+1)-D anomalous system and a symmetric pure state $\ket{\Psi} = \ket{A}\ket{B}\ket{C}\ket{D}$ that is separable for the four-partition $\{A,B,C,D\}$, as in Fig.~\ref{fig:fourpartite-nontrivial}. Following Eq.~\eqref{eq:constructrestriction}, we can reduce the symmetry to region $BCD$. Then, we extract its boundary operator $\Omega_{BCD}(g,h)$ and reduce it to $CD$. Extracting the boundary operator again, we arrive at $\Omega_c(g,h,k)$ (and $\Omega_d(g,h,k)$), supported near the intersection points of $A$, $B$ and $C$ (and $D$), inside $C$ ($D$). The separability of $C$ and $D$ regions require $\langle\Omega_c\rangle\langle\Omega_d\rangle=\langle\Omega_c\Omega_d\rangle\neq0$, which in turns requires $\Omega_{c,d}$ to carry trivial symmetry charge. Finally, since the four-cocycle $[\omega] \in H^4(G, U(1))$ is extracted from the symmetry charges of $\Omega_{c}(g,h,k)$, we conclude that the anomaly $[\omega]$ must be trivial for the symmetric state to be separable for such four-partitions.

Notice that for $d\geq2$ different partitions may have different topology, and not all $(d+2)$-partitions will result in a boundary intersection consisting of only one point. This is essential for our anomaly-nonseparability proof, and we have already seen partitions in (1+1)-D failing to satisfy this condition (See Fig.~\ref{fig:tripartite_circle_disconnected}). In higher dimensions, there are cases with no intersection point at all, such as the partition of the space into parallel ``strips'' in Fig.~\ref{fig:fourpartite-trivial}. For such topologies our argument does not hold, and it is not clear if anomaly still implies nonseparability for them.

Even though our argument is only valid for the specific form of symmetry action Eq.~\eqref{eq:fpform}, the universal nature of the result in $d\leq1$ motivates us to conjecture that the anomaly-nonseparability connection holds as long as the symmetry action $U(g)$ is anomalous. The argument in Sec.~\ref{sec:long-range_entanglement} then implies that the nonseparability is long-ranged, in the sense that a symmetric state cannot be prepared from a $(d+2)$-separable state using a finite-depth local channel.

\section{Mixed strong-weak anomaly}
\label{sec:mixed}
We now consider effects of weak symmetries. If there is only weak symmetry, then the maximally mixed state $\rho \propto \one$ will always be symmetric, since the identity matrix $\one$ commutes with all symmetry elements. This means that weak symmetry alone does not guarantee nontrivial universal structure, which agrees with previous results that weak symmetry alone does not support nontrivial bulk SPT phases~\cite{deGroot2022,ma_average_2023,Ma_intrisic_2023}. The next step, then, is to consider systems with both strong and weak symmetries, with mixed anomalies involving both. In the simplest case, a system has a mixed anomaly under $G \times H$ if it is anomalous under a symmetry representation of the product group $G \times H$, but not under one of its factors, say $G$, alone. If a subgroup, say $H$, is only a weak symmetry, then according to previous results the anomaly remains nontrivial~\cite{ma_average_2023,Ma_intrisic_2023}. However, our preceding result for nonseparability will no longer apply.

Let us first consider a simple example. Consider a periodic qubit chain (1d) with $\Z_2^3$ symmetry generated by
\begin{equation}\label{eq:Z2cubed_symmetries}
    \Xo = \prod_{i\text{ odd}} X_i, \ \Xe = \prod_{i\text{ even}} X_i, \ \CZ = \prod_{i=1}^L CZ_{i,i+1}.
\end{equation}
This symmetry action is anomalous~\cite{yoshida_topological_2016, bultinck_uv_2019}, as can be seen by reducing it to a region $A = \{1, 2, \ldots, 2n\}$ and calculating the boundary operators: $\CZ^{(A)} X_{\text{odd}}^{(A)} \CZ^{(A)} = Z_{2n} X_{\text{odd}}^{(A)}$. In this case, the boundary operator $Z_{2n}$ anticommutes not with the symmetry $\Xo$ that generated it, but with the other one, $\Xe$. This is the signature of a mixed anomaly. In contrast, the boundary operator of the anomalous $\Z_2$ of Sec.\ref{sec:4-qubit_CZX} was charged under the same symmetry $U_{BC}$ that generated it\footnote{A more precise statement of the mixed anomaly would come from calculating the 3-cocycle $\omega$ from Eq.\ref{eq:associativity_reduced_boundary_ops}. Apart from permutation of labels, it is easy to check that $\omega(a, b, c) = \exp(i \pi a_1 b_2 c_3)$, for $a, b, c \in \Z_2^3$. This 3-cocycle is not a coboundary, and it matches with the expected low-energy term in the gauged SPT bulk's Lagrangian, $\frac{1}{2} A_1 A_2 A_3 \in H^3(\Z_2^3, U(1))$~\cite{tantivasadakarn_building_2023}. Importantly, this anomaly is mixed because restricting $\omega$ to any $\Z_2^2 \simeq (\Z_2, \Z_2, 0) \subset \Z_2^3$ subgroup gives the trivial 3-cocycle $\omega|_{\Z_2^2} = 1$.}. In this Section, we will mainly focus on this particular $\Z_2^3$ symmetry for simplicity, even though we expect our discussions to be generalizable to other anomalous symmetries in $1d$.

If a mixed state is strongly symmetric under, say, just the $\Z_2^2$ subgroup, and weakly on the other $\Z_2$, our anomaly-nonseparability result no longer applies. As an example, if we choose the strong $\Z_2^2$ to be generated by $\Xe$ and $\Xo$, the infinite temperature states $\rho_{\infty} = \frac{1}{2^L}(\one + s_1 \Xe)(\one + s_2 \Xo)$ are equal mixtures of fully separable states $\ket{\pm \pm \cdots \pm}$, with appropriate choice of signs $\pm$ to respect $\Xe = s_1$ and $\Xo = s_2$. Hence, as mixed states, they are also fully separable.

Even though mixed anomaly can evade the nonseparability results, recent works on average SPT phases~\cite{ma_average_2023,Ma_intrisic_2023} suggest that states with strong-weak mixed anomaly must still be nontrivial in some sense. So what is nontrivial about the above $\rho_{\infty}$, or more generally states that respect strong and weak symmetries with mixed anomaly? We now show several ways in which strong-weak anomaly constrains its symmetric states. First we show that, similar to the pure state case, mixed anomaly serves as an obstruction of ``symmetry localization'' (Sec.~\ref{sec:obstructlocalization}). This has two consequences: the obstruction of localization for strong symmetry indicates that the mixed state cannot be symmetrically invertible~\cite{Ma_intrisic_2023} (Sec. \ref{sec:sym-noninv}); the obstruction of localization for weak symmetry, on the other hand, indicates that the mutual information and conditional mutual information cannot both be exponentially decaying (Sect. \ref{sec:sw-cmi}) -- this condition is analogous to being ``gapless'' for a mixed quantum state~\cite{sang2024stability}. The general lesson is that while the strong-weak mixed anomaly in general does not constrain quantum entanglement, it does constrain correlations of mixed states in more subtle ways.\footnote{We note that after the first version of our manuscript appeared online, several works\cite{WangLi2024,XuJian2024} also independently studied the consequences of mixed strong-weak anomaly. The constraints obtained here are different from and complementary to those in Refs.~\cite{WangLi2024,XuJian2024}.} Finally, we remind the reader that strong symmetries are also weak, by definition. Thus, all of the following results also apply to states with purely strong anomalous symmetries.

\subsection{Obstruction of symmetry localization}
\label{sec:obstructlocalization}

For pure states, the proposition of ``symmetry localization'' states that if a gapped ground state has a global symmetry $U|\psi\rangle=|\psi\rangle$, then the action of the symmetry restricted to a region $A$ (denoted as $U_A$) on $|\psi\rangle$ is equivalent to a unitary $V_{\partial A}$ acting only on the boundary of $A$:
\begin{equation}
    U_A|\psi\rangle=V_{\partial A}|\psi\rangle.
\end{equation}
Here the boundary region $\partial A$ has a width larger than the correlation length of the state. Symmetry localization is crucial for our understanding of symmetric gapped states, including symmetry-protected and symmetry-enriched topological orders\cite{else_classifying_2014, gauging3}. On the other hand, states with spontaneous symmetry breaking do not satisfy symmetry localization.

We now generalize the definition of symmetry localization to mixed states, for both strong and weak symmetries. We will show that in the presence of mixed anomaly, localization of one symmetry inevitably leads to the spontaneous breaking of the other symmetry. First we start with strong symmetry:
\begin{definition}[Strong symmetry localization]
    \label{def:strong_sym_localization}
    A strong symmetry $U$ is localized on state $\rho$ if, when restricted on a large region $A$, $U_A$ acts as a boundary operator $V_{\partial A}$:
    \begin{equation}
    \label{eq:stronglocalization}
        U_A\cdot\rho=V_{\partial A}\cdot\rho,
    \end{equation}
    where 
    \begin{enumerate}
        \item $V_{\partial A}$ has operator norm (the largest singular value) $\|V_{\partial A}\|=1$;
        \item In $1d$, $V_{\partial A}$ is locality-preserving.  In particular, when $\partial A$ is composed of two disconnected components, $\partial A = L \sqcup R$, we should have $V_{\partial A}=V_L \otimes V_R$. 
        \end{enumerate}
\end{definition}
The norm-$1$ condition on $V_{\partial A}$ (instead of unitarity) is needed for later purposes, and does not seem to bear much physical meaning beyond improving mathematical precision. Furthermore, the locality-preserving condition is not required for dimensions $d \geq 2$ to allow for more exotic states and symmetries that may elude the group cohomology classification outlined in section~\ref{sec:higher_dims} and localize to ``nonlocal'' boundary operators $V_{\partial A}$.

Using an argument similar to those in Sec.~\ref{sec:czx_tripartite_entanglement} and \ref{sec:tripartite_entanglement}, we have the following result:
\begin{theorem}
\label{thm:stronglocalization}
    For a $1d$ state $\rho$ with the anomalous $\Z_2^3$ symmetry Eq.~\eqref{eq:Z2cubed_symmetries}, where $K=\Z_2$ is strong and $G=\Z_2^2$ is weak, if the strong symmetry $K$ is localized on $\rho$, then $\rho$ must break $G$ spontaneously.
\end{theorem}
One example of the situation described above when $K$ is generated by $X_{\text{odd}}$, and $G$, by $X_{\text{even}}$ and $U_{CZ}$, is the state $\rho = \ketbra{+}{+}_{\text{odd}} \otimes (\frac{1}{2}[00\cdots0]+\frac{1}{2}[11\cdots1])_{\text{even}}$, where $[\psi] \defeq \ketbra{\psi}{\psi}$, for which $V_{\partial A} = \one$ and $X_{\text{even}}$ is spontaneously broken from a weak symmetry to nothing.

We expect this result to hold for general $1d$ symmetries with mixed strong-weak anomaly. But for simplicity we will only prove the statement for the special $\Z_2^3$ case.
\begin{proof}
    Without loss of generality, we choose the strong $\Z_2$ symmetry to be generated by $\Xo$, and the weak $\Z_2^2$ symmetry by $\Xe$ and $\CZ$ (see Eq.~\eqref{eq:Z2cubed_symmetries}). Let us denote the restricted action of $\Xo$ on a segment $A$ as $X_{{\rm odd},A}$, and the corresponding boundary operator $V_{\partial A}=V_LV_R$. For convenience we can always define the boundary so that $X_{{\rm odd},A}$ and $V_{L/R}$ commute.
    
    Now consider the effect of the weak symmetry $\Xe$. First, we show that as long as $\Xe$ is not spontaneously broken, $V_L$ (and $V_R$) can at most transform as a single irrep of $\Xe$ under conjugation. Suppose on the contrary $V_{L/R}=V_{L/R,+}+V_{L/R,-}$ where $V_{L/R,+}$ is invariant under conjugation by $\Xe$ and $V_{L/R,-}$ picks up a minus sign, and none of these operators annihilate $\rho$. Since $X_{{\rm odd},A}\cdot\rho$ commutes with $\Xe$, we must have
    \begin{equation}
        (V_{L,+}V_{R,-}+V_{L,-}V_{R,+})\rho=0.
    \end{equation}
    By assumption the operators $V_{L/R,\pm}$ do not annihilate $\rho$ individually. Since $V_{L,\pm}$ and $V_{R,\pm}$ act on different regions, their bilinear product should not annihilate $\rho$ either. We can then left-multiply the above equation with $V_{L,-}^{\dagger}V_{R,+}^{\dagger}$ to obtain
    \begin{eqnarray}
        \Tr [(V_{L,-}^{\dagger}V_{L,+})(V_{R,+}^{\dagger}V_{R,-})\rho]&=&-\Tr[(V_{L,-}^{\dagger}V_{L,-})(V_{R,+}^{\dagger}V_{R,+})\rho] \nonumber \\
        &=&-c<0,
    \end{eqnarray}
    where $c$ is an $O(1)$ number. The above equation indicates spontaneous breaking of the $\Xe$ symmetry, with $V_{L,-}^{\dagger}V_{L,+}$ being the local order parameter. Therefore in the absence of spontaneous symmetry breaking, $V_{L/R}$ should transform as a single irrep under $\Xe$.
    
    We can now consider the effect of the weak $\Z_2$ generated by $\CZ$:
    \begin{eqnarray}
    \label{eq:ODLRanomaly}
        \rho&=&V_L X_{\text{odd},A} V_R\rho=V_L X_{\text{odd},A} V_R\CZ\rho\CZ^{\dagger} \nonumber \\
        &=&V'_LZ_L X_{\text{odd},A} Z_RV'_R\rho \nonumber \\
        &=&(V'_LZ_LV^{\dagger}_L)(V'_RZ_RV^{\dagger}_R)\rho,
    \end{eqnarray}
    where $V'_{L/R}:=\CZ V_{L/R}\CZ^{\dagger}$. $Z_{L,R}$ are the Pauli-$Z$ operators at the boundary of $A$, transforming nontrivially under $\Xe$ -- the manifestation of the mutual anomaly. Since $\CZ$ commutes with $\Xe$, $V'_{L/R}$ carries the same $\Xe$ irrep as $V_{L/R}$. Therefore Eq.~\eqref{eq:ODLRanomaly} indicates spontaneous breaking of the $\Xe$ symmetry, with $V'Z V^{\dagger}$ being the local order parameter. 

    We therefore conclude that in the absence of spontaneous breaking of the weak symmetries, the strong symmetry (which has a mixed anomaly with the weak ones) cannot be localized.
\end{proof}

In Sec.~\ref{sec:sym-noninv} we will show that the non-localization of strong symmetry indicates that the anomalous state cannot be ``symmetrically invertible''.

Now we turn to weak symmetries:
\begin{definition}[Weak symmetry localization]
    \label{def:weak_sym_localization}
    A weak symmetry $W$ is localized on state $\rho$ if, when restricted on a large region $A$, $W_A$ acts as a boundary quantum channel $\E_{\partial A}$:
    \begin{equation}
    \label{eq:weaklocalization}
        W_A\rho W_A^{\dagger}=\E_{\partial A}[\rho],
    \end{equation}
    where the channel $\E_{\partial A}$ acts nontrivially only near the boundary, and is locality-preserving in $1d$: $\E_{\partial A}=\E_L \otimes \E_R$ for $A$ an interval with $\partial A = L \sqcup R$.
\end{definition}

The strong-weak mixed anomaly also serves as an obstruction of the weak symmetry localization:
\begin{theorem}
\label{thm:weaknonlocalization}
    For a $1d$ state $\rho$ with the anomalous $\Z_2^3$ symmetry Eq.~\eqref{eq:Z2cubed_symmetries}, where $K=\Z_2$ is strong and $G=\Z_2^2$ is weak, if the weak symmetry $G$ is localized on $\rho$, then $\rho$ must break the strong symmetry $K$ spontaneously.
\end{theorem}
An example the above when $K = \langle X_{\text{odd}}\rangle$ and $G = \langle X_{\text{even}}, U_{CZ}\rangle$, is the state $\rho = \frac{1}{2^L}(\one + X_{\text{odd}})$. When acting on $\rho$, the boundary channel of $X_{\text{even}}$ is $\E^{\text{even}}_{\partial A} = \one$, and the one of $U_{CZ}$ is conjugation by $\prod_{\langle i,j\rangle \in \partial A} CZ_{ij}$. Indeed, $\rho$ is the prototypical state with a strong symmetry $X_{\text{odd}}$ that is spontaneously broken to a weak symmetry \cite{SSBtoappear, sala_spontaneous_2024}.

Again we expect this result to hold for general $1d$ symmetries with mixed strong-weak anomaly, but will prove the statement only for the special $\Z_2^3$ case.
\begin{proof}
    Without loss of generality, we choose the strong $\Z_2$ symmetry to be generated by $\Xo$, and the weak $\Z_2^2$ symmetry by $\Xe$ and $\CZ$ (see Eq.~\eqref{eq:Z2cubed_symmetries}). Consider the restriction of $W:=\Xe$ on a large segment $A$ (denoted as $X_{\text{even},A}$), which by assumption satisfies Eq.~\eqref{eq:weaklocalization} (again for convenience we can choose the boundary region so that $X_{\text{even},A}$ and $\E_{L/R}$ commute). By considering the dilated (or purified) form of Eq.~\eqref{eq:weaklocalization} $W_A|\Psi_{\rho}\rangle=V_{\partial A}|\Psi_{\rho}\rangle$, we can use the same argument in the proof of Thm.~\ref{thm:stronglocalization} to show that $\E_{L/R}$ can at most create a definite $\Z_2$ charge ($+1$ or $-1$) of the strong symmetry $\Xo$. We then consider the effect of the weak $\CZ$ symmetry:
    \begin{eqnarray}
    \label{eq:swssb}
        \rho&=&X_{\text{even},A}\E_L\circ\E_R[\rho]X_{\text{even},A} \nonumber \\
        &=&X_{\text{even},A}\E_L\circ\E_R[\CZ\rho\CZ^\dagger]X_{\text{even},A} \nonumber \\
        &=&X_{\text{even},A}Z_LZ_R\E'_L\circ\E'_R[\rho]X_{\text{even},A}Z_LZ_R \nonumber \\
        &=&Z_LZ_R\left(\E'_L\circ \E_L \circ \E'_R\circ \E_R[\rho]\right) Z_LZ_R,
    \end{eqnarray}
    where the modified channel 
    \begin{equation}
        \E'_{L/R}[\rho]:=\CZ^{\dagger}\E_{L/R}[\CZ\rho\CZ^{\dagger}]\CZ
    \end{equation} 
    creates the same strong $\Xo$ charge as $\E_{L/R}$. This means that the local channel $\E'_{L/R}\circ\E_{L/R}$ preserves the strong $\Xo$ symmetry, while $Z_{L/R}$ on $\partial A$ is charged under $\Xo$. Eq.~\ref{eq:swssb} is thus the statement that the strong $\Xo$ symmetry is spontaneously broken to a weak symmetry, or nothing\cite{SSBtoappear}.

    We therefore conclude that in the absence of spontaneous breaking of the strong symmetry, the weak symmetry which has a mixed anomaly with the strong one cannot be localized.
\end{proof}

In Sec.~\ref{sec:sw-cmi} we will show that the non-localization of weak symmetry indicates that, in $1d$, anomalous state cannot be ``gapped'', which is analogous to the statement in pure states.

\subsection{Strong symmetry localization and symmetric invertibility}
\label{sec:sym-noninv}

First we review the notion of \emph{symmetric invertibility} from Ref.~\cite{Ma_intrisic_2023}:

\begin{definition}[Symmetric invertibility]\label{def:sym_inv}
   A mixed state on the physical Hilbert space $\rho\in\Hilb_\sys$ is \emph{symmetrically invertible} if there exists an ancillary mixed states (the ``inverse'' state) $\rho_I, \in \Hilb_\anc$, and a pair of symmetric finite-depth local channels $\E_T, \E_P : Q(\Hilb_\sys \otimes \Hilb_\anc) \to Q(\Hilb_\sys \otimes \Hilb_\anc)$ such that
   \begin{equation}
    \label{eq:invertible}
       \rho \otimes \rho_I \xrightarrow{\E_T} \trivialstate \xrightarrow{\E_P} \rho \otimes \rho_I,
    \end{equation}
    where $\trivialstate \in \Hilb_\sys \otimes \Hilb_\anc$ is a symmetric pure product state.
\end{definition}
We now explain the details of the definition above. First, the ancillary Hilbert space $\Hilb_\anc$ and the enlarged Hilbert space to which the local channels are purified have to be compatible with the geometry of the physical Hilbert space $\Hilb_\sys$. Here, this means that each on-site Hilbert space of $\Hilb_\sys$ is individually enlarged. Second, we extend the symmetry action from $\Hilb_\sys$ to $\Hilb_\sys \otimes \Hilb_\anc$ in a faithful manner: for a symmetry element $g$ we extend the symmetry action to $U_g\to U_g\otimes U_g^{\mathcal{A}}$, where $U_g^{\mathcal{A}}$ is a finite-depth unitary acting on the ancilla. Third, each gate composing the channels $\E_T$ and $\E_P$ is symmetric under this extended symmetry in accordance to the strong and weak parts of the symmetry as explained in Sec. \ref{sec:intro-strong_and_weak_symmetries}.

For pure states, an invertible state can be viewed as a unique gapped ground state of some local Hamiltonian, which forms the basis of classifying unique gapped ground states using invertible topological quantum field theory~\cite{invertible2,invertible3}. For mixed state SPT, invertibility is also needed to make the bulk topological invariants well-defined~\cite{Ma_intrisic_2023}.

In Ref.~\cite{Ma_intrisic_2023} it was argued that symmetric invertibility guarantees strong symmetry localization. We review this result below:
\begin{theorem}
    For a symmetric invertible state $\rho$, a strong symmetry $g$ can always be localized (in the sense of Def.~\ref{def:strong_sym_localization}).
\end{theorem}
\begin{proof}
    Consider a dilated form of $\E_T$ in Eq.~\eqref{eq:invertible} (with the ancilla $a$):
    \begin{equation}
       U_T \rho\otimes\rho_I\otimes|0\rangle\langle0|_a U_T^{\dagger}=|\psi_{\rm Trivial}\rangle\langle \psi_{\rm Trivial}|\otimes \rho_a.
    \end{equation}
    For the trivial product state $|\psi_{\rm Trivial}\rangle$, the restricted symmetry action $g_A\otimes g^{\mathcal{A}}_A$ acts trivially, so we have
    \begin{equation}
        g_A\otimes g^{\mathcal{A}}_A\cdot\rho\otimes\rho_I\otimes|0\rangle\langle0|_a=\tilde{V}_{\partial A}\cdot \rho\otimes\rho_I\otimes|0\rangle\langle0|_a,
    \end{equation}
    where $\tilde{V}_{\partial A}=(g_Ag^{\mathcal{A}}_A)U^{\dagger}_T(g_Ag^{\mathcal{A}}_A)^{\dagger}U_T$ acts nontrivially on the boundary region $\partial A$ because $U_T$ commutes with the strong symmetry in the bulk of $A$.  Tracing out $\mathcal{H}_{\mathcal{A}}\otimes \mathcal{H}_a$, we obtain the desired form of symmetry localization Eq.~\eqref{eq:stronglocalization}.
\end{proof}

If we combine the above theorem with Thm~\ref{thm:stronglocalization}, we conclude that a $1d$ mixed state with mixed strong-weak anomaly cannot be symmetrically invertible as long as the weak symmetry is not spontaneously broken. Then we notice that spontaneous breaking of weak symmetry also leads to non-invertibility, since the existence of $\E_P$ in Eq.~\eqref{eq:invertible} requires all connected correlation functions to vanish. We therefore conclude that
\begin{corollary}
    A $1d$ mixed state $\rho$ with mixed strong-weak anomaly cannot be symmetrically invertible.
\end{corollary}

The above {result} covers both mixed strong-weak anomaly and purely strong anomaly. Purely weak symmetries do not have nontrivial anomaly, as discussed in Refs.~\cite{deGroot2022,ma_average_2023,LeeYouXu2022,ZhangQiBi2022,Ma_intrisic_2023}. For example, the maximally mixed state is weakly symmetric under any symmetry\footnote{See, however, Refs.~\cite{Hsin2023,Zang2023} for recent discussions on ``anomaly'' defined in different ways, which can be nontrivial even for purely weak symmetries.}.

It is natural to expect that our result can be extended to higher dimensions. However in higher dimensions we do not have a controlled argument, which has to do with the fact that the Else-Nayak mechanism for t'Hooft anomaly has not been rigorously developed. In this work we will leave the higher dimensional version as a conjecture.

Recall that, for pure states, a non-invertible state always come with some sharp feature, such as spontaneous symmetry breaking (SSB), gapless low-energy spectrum or topological order. From this perspective, the non-invertibility of the seemingly innocuous mixed state $\rho_{\infty}$ may come as a surprise. Even though we do not fully understand the landscape of non-invertible mixed states at this point, the simple $\rho_{\infty}$ state does come with a sharp feature: it spontaneously breaks the strong $\Z_2\times\Z_2$ symmetry  down to only a weak symmetry, {akin to an SSB cat state}. Such unusual pattern of spontaneous symmetry-breaking, possible only in mixed quantum states, has been proposed recently in Refs.~\cite{LeeJianXu2023, Ma_intrisic_2023} and was studied in more detail in Ref.~\cite{SSBtoappear}.

\subsection{Weak symmetry localization and gapped Markovian states}
\label{sec:sw-cmi}

First we define the analogue of ``gapped states'' in mixed states:
\begin{definition}[Gapped Markovian states] 
    \label{def:gapped}
    Consider a disk region $A$, a buffer region $B$ surrounding $A$ of width $w$ and the rest of the system $C$, so that $\mathrm{dist}(A,C) \geq w$ (See Fig.~\ref{fig:gapped_markovian_def}). We call $\rho$ a \textit{gapped Markovian state} if for any such partitioning of space, we have
    \begin{enumerate}
        \item exponentially decaying mutual information (MI) $I(A:C) = O(e^{-w/\xi_{\text{MI}}})$ for some finite correlation length $\xi_{\text{MI}}$; and
        \item exponentially decaying conditional mutual information (CMI) $I(A:C|B) = O(e^{-w/\xi_{\text{CMI}}})$ for some finite ``Markov length'' $\xi_{\text{CMI}}$.
    \end{enumerate}
\end{definition}

\begin{figure}[th]
    \centering
    \def\svgwidth{\linewidth}
\begingroup%
  \makeatletter%
  \providecommand\color[2][]{%
    \errmessage{(Inkscape) Color is used for the text in Inkscape, but the package 'color.sty' is not loaded}%
    \renewcommand\color[2][]{}%
  }%
  \providecommand\transparent[1]{%
    \errmessage{(Inkscape) Transparency is used (non-zero) for the text in Inkscape, but the package 'transparent.sty' is not loaded}%
    \renewcommand\transparent[1]{}%
  }%
  \providecommand\rotatebox[2]{#2}%
  \newcommand*\fsize{\dimexpr\f@size pt\relax}%
  \newcommand*\lineheight[1]{\fontsize{\fsize}{#1\fsize}\selectfont}%
  \ifx\svgwidth\undefined%
    \setlength{\unitlength}{231.69452337bp}%
    \ifx\svgscale\undefined%
      \relax%
    \else%
      \setlength{\unitlength}{\unitlength * \real{\svgscale}}%
    \fi%
  \else%
    \setlength{\unitlength}{\svgwidth}%
  \fi%
  \global\let\svgwidth\undefined%
  \global\let\svgscale\undefined%
  \makeatother%
  \begin{picture}(1,0.40291719)%
    \lineheight{1}%
    \setlength\tabcolsep{0pt}%
    \put(0,0){\includegraphics[width=\unitlength,page=1]{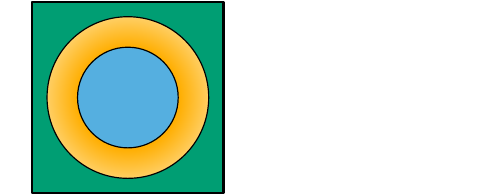}}%
    \put(0.2649139,0.18886319){\color[rgb]{0,0,0}\makebox(0,0)[t]{\lineheight{0.80000001}\smash{\begin{tabular}[t]{c}$A$\end{tabular}}}}%
    \put(0.35339794,0.29102078){\color[rgb]{0,0,0}\makebox(0,0)[t]{\lineheight{0.80000001}\smash{\begin{tabular}[t]{c}$B$\end{tabular}}}}%
    \put(0.41248512,0.35010777){\color[rgb]{0,0,0}\makebox(0,0)[t]{\lineheight{0.80000001}\smash{\begin{tabular}[t]{c}$C$\end{tabular}}}}%
    \put(0,0){\includegraphics[width=\unitlength,page=2]{gapped_markovian_def.pdf}}%
    \put(0.12907221,0.20734866){\color[rgb]{0,0,0}\makebox(0,0)[t]{\lineheight{0.80000001}\smash{\begin{tabular}[t]{c}$w$\end{tabular}}}}%
    \put(0.96688686,0.00163133){\color[rgb]{0,0,0}\makebox(0,0)[t]{\lineheight{0.80000001}\smash{\begin{tabular}[t]{c}$w$\end{tabular}}}}%
    \put(0,0){\includegraphics[width=\unitlength,page=3]{gapped_markovian_def.pdf}}%
    \put(0.67338928,0.10045887){\color[rgb]{0.56470588,0.15686275,0.61568627}\makebox(0,0)[rt]{\lineheight{0.80000001}\smash{\begin{tabular}[t]{r}MI\end{tabular}}}}%
    \put(0.73316027,0.18005851){\color[rgb]{0.74509804,0.28235294,0.23529412}\makebox(0,0)[lt]{\lineheight{0.80000001}\smash{\begin{tabular}[t]{l}CMI\end{tabular}}}}%
    \put(0,0){\includegraphics[width=\unitlength,page=4]{gapped_markovian_def.pdf}}%
    \put(0.61269196,0.36355832){\color[rgb]{0,0,0}\makebox(0,0)[lt]{\lineheight{0.80000001}\smash{\begin{tabular}[t]{l}Correlation\end{tabular}}}}%
    \put(0.68253129,0.00564201){\color[rgb]{0,0,0}\makebox(0,0)[t]{\lineheight{0.80000001}\smash{\begin{tabular}[t]{c}$\xi_{\text{MI}}$\end{tabular}}}}%
    \put(0.78595869,0.00564201){\color[rgb]{0,0,0}\makebox(0,0)[t]{\lineheight{0.80000001}\smash{\begin{tabular}[t]{c}$\xi_{\text{CMI}}$\end{tabular}}}}%
    \put(0.04974802,0.377021){\color[rgb]{0,0,0}\makebox(0,0)[rt]{\lineheight{0.80000001}\smash{\begin{tabular}[t]{r}(a)\end{tabular}}}}%
    \put(0.56451067,0.377021){\color[rgb]{0,0,0}\makebox(0,0)[rt]{\lineheight{0.80000001}\smash{\begin{tabular}[t]{r}(b)\end{tabular}}}}%
  \end{picture}%
\endgroup%
    \caption{(a) Partitioning of the space used for definition \ref{def:gapped} of gapped Markovian states, and (b) plot of the exponential decay of the MI and CMI with the width $w$ of the buffer region $B$ for such states. Note that we do not have $I(A:C|B) \geq I(A:C)$ nor $\xi_{\text{CMI}} \geq \xi_{\text{MI}}$ in general, and figure (b) illustrates just one possible scenario.}
    \label{fig:gapped_markovian_def}
\end{figure}

Recall that the MI is defined as $I(A:C)\equiv S(A)+S(C)-S(AC)$ ($S$ being the von Neumann entropy) and the CMI is $I(A:C|B)\equiv S(AB)+S(BC)-S(C)-S(ABC)$. For pure states, the MI coincides with CMI, and the gap condition reduces to the usual one that demands correlation functions to decay exponentially, which is indeed satisfied by ground states of gapped Hamiltonians~\cite{hastings_liebschultzmattis_2004, nachtergaele_liebrobinson_2006}. For general mixed states, the MI condition is rather intuitive as it demands correlation functions to decay exponentially. The CMI condition was emphasized recently in Ref.~\cite{sang2024stability} as a desirable property of mixed states with features similar to ground states of gapped Hamiltonians. For our purpose, the most important consequence of the CMI condition is that the state on $A$ can be recovered, using the Petz map, by a channel acting on $AB$ only as long as $w \gg \xi_{\text{CMI}}$~\cite{fawzi_quantum_2015}, up to an error that decays exponentially with $r/\xi_{\text{CMI}}$ in trace distance. 

We now present the main result of this subsection:

\begin{figure*}[t]
    \centering
    \def\svgwidth{0.9\linewidth}
\begingroup%
  \makeatletter%
  \providecommand\color[2][]{%
    \errmessage{(Inkscape) Color is used for the text in Inkscape, but the package 'color.sty' is not loaded}%
    \renewcommand\color[2][]{}%
  }%
  \providecommand\transparent[1]{%
    \errmessage{(Inkscape) Transparency is used (non-zero) for the text in Inkscape, but the package 'transparent.sty' is not loaded}%
    \renewcommand\transparent[1]{}%
  }%
  \providecommand\rotatebox[2]{#2}%
  \newcommand*\fsize{\dimexpr\f@size pt\relax}%
  \newcommand*\lineheight[1]{\fontsize{\fsize}{#1\fsize}\selectfont}%
  \ifx\svgwidth\undefined%
    \setlength{\unitlength}{340.57868056bp}%
    \ifx\svgscale\undefined%
      \relax%
    \else%
      \setlength{\unitlength}{\unitlength * \real{\svgscale}}%
    \fi%
  \else%
    \setlength{\unitlength}{\svgwidth}%
  \fi%
  \global\let\svgwidth\undefined%
  \global\let\svgscale\undefined%
  \makeatother%
  \begin{picture}(1,0.29304658)%
    \lineheight{1}%
    \setlength\tabcolsep{0pt}%
    \put(0,0){\includegraphics[width=\unitlength,page=1]{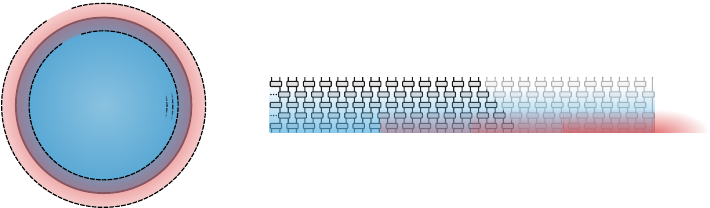}}%
    \put(0.06835012,0.26648355){\makebox(0,0)[lt]{\lineheight{1.25}\smash{\begin{tabular}[t]{l}$R^+_w$\end{tabular}}}}%
    \put(0.08906051,0.2315012){\makebox(0,0)[lt]{\lineheight{1.25}\smash{\begin{tabular}[t]{l}$R^-_w$\end{tabular}}}}%
    \put(0,0){\includegraphics[width=\unitlength,page=2]{gapped_localization.pdf}}%
    \put(0.26079629,0.2654585){\color[rgb]{0.79607843,0.29803922,0.29803922}\makebox(0,0)[lt]{\lineheight{1.25}\smash{\begin{tabular}[t]{l}$\partial R_w$\end{tabular}}}}%
    \put(0,0){\includegraphics[width=\unitlength,page=3]{gapped_localization.pdf}}%
    \put(0.44940037,0.19774406){\makebox(0,0)[lt]{\lineheight{1.25}\smash{\begin{tabular}[t]{l}$\tilde{W}_R$\end{tabular}}}}%
    \put(0.6519117,0.1948118){\makebox(0,0)[t]{\lineheight{1.25}\smash{\begin{tabular}[t]{c}\large $W_R$\end{tabular}}}}%
    \put(0.79393841,0.07325078){\color[rgb]{0,0,0}\makebox(0,0)[t]{\lineheight{0.80000001}\smash{\begin{tabular}[t]{c}$\frac{1}{2}w$\end{tabular}}}}%
    \put(0.66557741,0.07325078){\color[rgb]{0,0,0}\makebox(0,0)[t]{\lineheight{0.80000001}\smash{\begin{tabular}[t]{c}$\frac{1}{2}\tilde{w}$\end{tabular}}}}%
    \put(0.53721627,0.07325078){\color[rgb]{0,0,0}\makebox(0,0)[t]{\lineheight{0.80000001}\smash{\begin{tabular}[t]{c}$\frac{1}{2}\tilde{w}'$\end{tabular}}}}%
    \put(0,0){\includegraphics[width=\unitlength,page=4]{gapped_localization.pdf}}%
    \put(0.02122652,0.14894803){\color[rgb]{0,0,0}\makebox(0,0)[t]{\lineheight{0.80000001}\smash{\begin{tabular}[t]{c}$w$\end{tabular}}}}%
    \put(0.14605266,0.13847027){\color[rgb]{0,0,0}\makebox(0,0)[t]{\lineheight{0.80000001}\smash{\begin{tabular}[t]{c}\large $R$\end{tabular}}}}%
    \put(0.37320351,0.09476555){\color[rgb]{0.3372549,0.68627451,0.87843137}\makebox(0,0)[rt]{\lineheight{0.80000001}\smash{\begin{tabular}[t]{r}\large $R$\end{tabular}}}}%
    \put(0,0){\includegraphics[width=\unitlength,page=5]{gapped_localization.pdf}}%
  \end{picture}%
\endgroup%

    \caption{Partioning of the space used in the proof of theorem \ref{thm:gappedMarkovian}. On the left, a disk-like region $R$ is pictured. $\partial R$ is its boundary region with thickness $w$, from which we define $R^-_w \defeq R \setminus \partial R_w$ and $R^+_w \defeq R \cup \partial R_w$. On the right, a snapshot of the radial direction around $\partial R$ is depicted. The symmetry FDLU $W$ is restricted to region $R$ by removing all gates acting outside it, resulting in $W_R$. Importantly, we can remove further gates (the ones more transparent to the right) to arrive at a FDLU $\tilde{W}_R$ that acts like $W_R$ on $R^-_{\tilde{w}}$ but trivially on $\partial R_{w}$.}
    \label{fig:gapped_Markovian_proof}
\end{figure*}

\begin{theorem}
    For a gapped Markovian state $\rho$, a weak symmetry $g$ can always be localized (in the sense of Def.~\ref{def:weak_sym_localization}).
    \label{thm:gappedMarkovian}
\end{theorem}
\begin{proof}
    Consider a disk region $R$ and the state
    \begin{equation}
        \rho' \defeq W_R\rho W_R^{\dagger},
    \end{equation}
    where $W_R$ is the weak symmetry restricted to $R$. We will first use the exponentially decaying MI property to argue that $\rho'$ is equal to $\rho$ far from the boundary of $R$, and then use the exponentially decaying CMI to locally recover $\rho'$ from $\rho$, thus arriving at a localized action of the weak symmetry (See Eq.~\eqref{eq:weaklocalization}).

    Let $\partial R_w$ be the boundary of $R$ thickened by a width $w$, and define $R^-_w \defeq R \setminus \partial R_w$ and $R^+_w \defeq R \cup \partial R_w$ (See Fig. \ref{fig:gapped_Markovian_proof}). If $w \gg \xi_{\text{MI}}$, then $\rho$ has almost no correlations between regions $R^-_w$ and $(R^+_w)^c$, separated by a distance $w$\footnote{Recall that the mutual information bounds the distance between a correlated state $\rho_{AB}$ and its marginals $\rho_A$ and $\rho_B$, as $I_{\rho}(A:B) = S(\rho_{AB} \parallel \rho_A \otimes \rho_B) \geq \frac{1}{2 \ln(2)} \norm{\rho_{AB} - \rho_A \otimes \rho_B}^2_1$, where $S(\cdot \parallel \cdot)$ is the relative entropy, and $\norm{\cdot}_1$ the trace norm~\cite{hiai_sufficiency_1981, watrous_theory_2018a}.}. That is,
    \begin{equation}\label{eq:uncorrelated_rho}
        \Tr_{\partial R_w}[\rho] \approx \rho_{R^-_w} \otimes \rho_{(R^+_w)^c},
    \end{equation}
    up to an exponentially small deviation in $w/\xi_{\text{MI}}$. 

    Now let $r_W$ be the range of the $W$, defined as the smallest length $r$ such that the support of $W^\dagger X_A W$ is contained in $A^+_r$ for $X_A$ an arbitrary observable supported in a region $A$. Then, for an extended width $\tilde{w} > 2 r_W + w$, the state $\rho' = W_R \rho W_R^\dagger$ is similarly uncorrelated at large distances:
    \begin{align}
        \Tr_{\partial R_{\tilde{w}}}[\rho'] & = \Tr_{\partial R_{\tilde{w}} \setminus \partial R_w}[\tilde{W}_R \Tr_{\partial R_w}[\rho] \tilde{W}_R^\dagger] \\
        & \stackrel{\eqref{eq:uncorrelated_rho}}{\approx} \rho'_{R_{\tilde{w}}^-} \otimes \rho_{(R^+_{\tilde{w}})^c},
    \end{align}
    where $\tilde{W}_R$ acts like $W_R$ in $R^-_{\tilde{w}}$ but trivially in $\partial R_w$. Such operator can be constructed by removing gates outside the ``light cone'' coming from $R^-_ {\tilde{w}}$ (See Fig. \ref{fig:gapped_Markovian_proof}). Moreover, we have
    \begin{equation}
        \rho'_{R_{\tilde{w}}^-} = \Tr_{(R_{\tilde{w}}^-)^c}[W_R \rho W_R^\dagger] 
        = \Tr_{(R_{\tilde{w}}^-)^c}[W \rho W^\dagger]
        = \rho_{R_{\tilde{w}}^-},
    \end{equation}
    where in the last equality we used that $W$ is a weak symmetry of $\rho$. Therefore,
    \begin{equation}\label{eq:bulk_indistinguishability}
        \Tr_{\partial R_{\tilde{w}}}[\rho'] \approx \rho_{R_{\tilde{w}}^-} \otimes \rho_{(R^+_{\tilde{w}})^c} \approx \Tr_{\partial R_{\tilde{w}}}[\rho].
    \end{equation}

    Finally, we extend the boundary region of $R$ even more to a width of $\tilde{w}'$ satisfying $\tilde{w}' \gg \xi_{\text{CMI}} + \tilde{w}$. By treating the extension $\partial R_{\tilde{w}'} \setminus \partial R_{\tilde{w}}$ as a buffer region $B$ separating the smaller boundary $A = \partial R_{\tilde{w}}$ from the bulk regions $R^-_{\tilde{w}'} \cup (R^+_{\tilde{w}'})^c$, the exponentially small CMI $I(A:C|B)$ from the gapped Markovian condition\footnote{In Def.~\ref{def:gapped}, we required $A$ to be a disk region, which is not the case for annulus $\partial R_{\tilde{w}}$ in dimensions $d \geq 2$. However, since it does not wrap around the base manifold nontrivially (as it is the boundary of a disk region $R$), we find it reasonable to assume exponentially decaying CMI for $A=\partial R_{\tilde{w}}$ as well. In particular, topologically ordered states do satisfy this extra assumption.} implies the existence of a channel $\E_{B \to AB}$ from $B$ to $AB = \partial R_{\tilde{w}'}$ such that~\cite{fawzi_quantum_2015}
    \begin{equation}
        \rho' \approx \E_{B \to AB}(\rho'_{BC}) \stackrel{\eqref{eq:bulk_indistinguishability}}{\approx} \E_{B \to AB}(\rho_{BC}),
    \end{equation}
    which is of the form of Eq.~\eqref{eq:weaklocalization} if we define $\E_{\partial R} \defeq \E_{B \to AB} \circ \Tr_A$.

    In $1d$, we can further guarantee $\E_{\partial R}$ to be locality-preserving by applying the steps above in each of the connected components of $\partial R$ separately. For $R$ an interval, for example, we would arrive at $\E_{\partial R} = \E_L \otimes \E_R$, where $\E_L$ and $\E_R$ are supported around the left anf right endpoints of $R$, respectively.
\end{proof}

The above theorem combined with Thm.~\ref{thm:weaknonlocalization} implies that a $1d$ mixed state with mixed strong-weak anomaly cannot be gapped Markovian if the strong symmetry is not spontaneously broken. Moreover, it was shown in Ref.~\cite{SSBtoappear} that a state with spontaneously broken strong symmetry cannot have a finite Markov length for all tripartitions $ABC$, and hence cannot be gapped Markovian. We therefore conclude that
\begin{corollary}\label{cor:anomaly_implies_gapless}
    A $1d$ mixed state $\rho$ with mixed strong-weak anomaly cannot be gapped Markovian (in the sense of Def.~\ref{def:gapped}).
\end{corollary}
The above statement directly generalizes the familiar result in pure states. To illustrate the necessity of both the MI and CMI condition, let us consider the $\Z_2^2$ symmetry generated by $\CZ$ and $X \defeq \Xe\Xo$. We can take $X$ to be strong and $\CZ$ to be weak, then the state $\rho_0\propto \one+X$ is fully symmetric and has zero MI between distant regions $A$ and $C$, $I(A:C) = 0$. However, the CMI $I(A:C|B)$ does not decay with $|B|$, instead saturating at $\log 2$. This example shows that the CMI can be positive even though the MI is zero. For the contrary, we can take $\CZ$ to be strong and $X$ to be weak. A simple symmetric state would be the clasically correlated state
\begin{equation}
    \rho\propto \ketbra{\uparrow\uparrow\cdots}{\uparrow\uparrow\cdots}+\ketbra{\downarrow\downarrow\cdots}{\downarrow\downarrow\cdots}.
\end{equation}
This state has vanishing CMI, but the MI between any two regions is $\log 2>0$.

We note that a simpler argument to corollary \ref{cor:anomaly_implies_gapless} can be made if all the symmetries are strong. Indeed, a $1d$ gapped Markovian state $\rho$ can be prepared from a tripartite separable state using a finite-depth channel\footnote{The channel can be constructed as follows. Starting from consider a tripartition $A,B,C$, we can trace out the boundary regions as in the proof of Thm.~\ref{thm:gappedMarkovian}, which leaves behind a tripartite factorized state $\rho_{A^-}\otimes\rho_{B^-}\otimes\rho_{C^-}$ due to the vanishing MI. Because of the vanishing CMI, we can then recover the original $\rho$ using Petz maps at the boundary region -- this precisely gives a local channel that prepares $\rho$ from a tripartite factorized state.}. Based on our result on long-range tripartite entanglement of Sec.~\ref{sec:tripartite_entanglement}, we conclude that gapped Markovianity is also forbidden by strong symmetry anomaly in $1d$.

Another consequence of Thm.~\ref{thm:weaknonlocalization} is that for symmetric gapped state, we can define and discuss the notion of symmetry-protected topology (SPT) by examining the symmetry properties of the boundary channels $\E_{\partial A}$ -- this leads naturally to the ``decorated domain wall'' picture advocated in Refs.~\cite{ma_average_2023,Ma_intrisic_2023}.

In this Section we discussed (the absence of) two mixed-state properties: symmetric invertibility and gapped Markovianity. In $1d$, these two notions are essentially equivalent for pure states -- for example, one can use either to define the notion of SPT. An open question is whether these two notions remain equivalent for mixed states. We leave this question to future study.

\section{Lieb-Schultz-Mattis and nonseparability}
\label{sec:LSM}

In Sec.~\ref{sec:mixed} we discussed the general consequences of strong-weak mixed anomaly in terms of symmetry localization, invertible and gapped Markovian states. It turns out that for some special (yet physically interesting) cases we can have stronger results on quantum entanglement. In this section we discuss spin-$1/2$ chains (1d) with continuous non-Abelian onsite strong symmetry ($SO(3)$ or $O(2)$) and weak lattice translation symmetry. For pure states, it is by now well established that there is a mutual 't Hooft anomaly involving the onsite and translation symmetries~\cite{Cheng2015,Cheng2023}. This mutual anomaly strongly constrains the long-range behavior of the system, resulting in the celebrated Lieb-Schultz-Mattis (LSM) theorem~\cite{Lieb1961}. Below we will show that the LSM anomaly also puts a strong constraint on mixed states, namely any mixed state $\rho$ that does not break the weak symmetries (explicitly or spontaneously) must be bipartite nonseparable. Furthermore, this nonseparability is long-ranged: $\rho$ cannot be obtained from a bipartite separable state $\rho_0$ through a finite-depth local channel.

Suppose on the contrary that for a bipartition $A,B$,
\begin{equation}
\label{eq:LSMbisep}
    \rho=\sum_i p_i |A_i\rangle\langle A_i|\otimes |B_i\rangle\langle B_i|,
\end{equation}
where $|A_i\rangle$ ($|B_i\rangle$) are states living on the subsystem $A$ ($B$). Let us consider the Pauli operator $Z_x$ that generates an $SO(2)$ symmetry, with the total $SO(2)$ charge $Q=\sum_xZ_x$. Now consider two points $x$ and $y$  inside the regions $A$ and $B$, respectively. We have
\begin{equation}
\label{eq:LSMsep}
\Tr[\rho Z_x Z_y]=\sum_i p_i\langle A_i|Z_x|A_i\rangle\langle B_i|Z_y|B_i\rangle=0.
\end{equation}
In the last equality above we have used the fact that either $O(2)$ or $SO(3)$ contains a $\pi$-rotation operation that sends $Z_x\to-Z_x$, and that we can apply this rotation on $A$ (or $B$) alone, since each $\ket{A_i}$ is individually symmetric from the strong symmetry of $\rho$ and the on-siteness of the spin flip $\prod_i X_i$.

The condition Eq.~\eqref{eq:LSMsep} is valid only for $x \in A$ and $y \in B$. To include cases in which both $x, y \in A$ or $B$, we invoke the weak translation symmetry of $\rho$, implying that $\Tr[\rho Z_xZ_y]=0$ for all $x\neq y$. The strong symmetry condition also requires $\sum_xZ_x\rho=0$, namely $\rho$ contains only states with zero total $SO(2)$ charge. We then have
\begin{eqnarray}
    0&=&\Tr \rho (\sum_xZ_x)^2 \nonumber \\
     &=&\Tr \rho (\sum_xZ_x^2+\sum_{x\neq y}Z_xZ_y ) \nonumber \\
     &=&L,
\end{eqnarray}
which is a contradiction. So our original proposition Eq.~\eqref{eq:LSMbisep}, namely $\rho$ is bipartite separable, must be false.

Notice that if we only have the Abelian on-site strong symmetry $SO(2)$, then Eq.~\eqref{eq:LSMsep} may not hold since we could have $\langle A_i|Z_x|A_i\rangle\neq0$. In this case we can have a fully separable density matrix $\rho\propto P_{Q=0}$ (the projection operator to the subspace with $Q=\sum_xZ_x=0$) that clearly does not spontaneously break translation symmetry. This is curiously in accordance with the fact that the LSM constraint no longer corresponds to a standard 't Hooft anomaly once the onsite symmetry is lowered to $SO(2)$~\cite{Cheng2015,Song2019,Else2021,Cheng2023}.

Back to the $O(2)$ or $SO(3)$ case, we now show that a symmetric state $\rho$ not only is bipartite nonseparable, but also cannot be prepared from any biparte separable state $\rho_0=\sum_i p_i|A_i\rangle\langle A_i|\otimes |B_i\rangle\langle B_i|$ using a finite depth local channel. In this statement there is no symmetry requirement on the initial state $\rho_0$ and the finite-depth channel $\E_{FD}$.

To prove the above statement, suppose on the contrary that $\rho=\E_{\rm FDLC}(\rho_0)$, then consider part of the ensemble $\rho_i \defeq \E_{\rm FDLC}(|A_i\rangle\langle A_i|\otimes |B_i\rangle\langle B_i|)$. The connected correlation function $\langle Z_x Z_y\rangle-\langle Z_x\rangle \langle Z_y\rangle=0$ on $|A_i\rangle\otimes|B_i\rangle$ for $x$ ($y$) deep in the region $A$ ($B$), simply because the two regions are uncorrelated. Now under a finite depth channel, this connected correlation function should remain zero, or at least exponentially small in $|x-y|$ if the channel is made of gates with exponential tails. But for the state $\rho_i$, $\langle Z_x\rangle=\langle Z_y\rangle=0$ by strong symmetry. Therefore $\langle Z_xZ_y\rangle=0$ on $\rho_i$ as well as on $\rho=\sum_i p_i \rho_i$, up to exponential tails.

It was proved in Ref.~\cite{KimchiNahumSenthil} that as long as $\sum_xZ_x\rho=0$, namely that $SO(2)$ is a strong symmetry at half-filling, then we have either (a) spontaneous breaking of the weak (average) translation symmetry, in the sense that for some local operator $V(x)$, $\Tr[\rho V(x) V(y)] \sim (-1)^{(x-y)}$; or (b) the two-point function $\Tr[\rho Z_x Z_y]$ should decay no faster than $1/|x-y|^2$. Since $\langle Z_xZ_y\rangle=0$ (up to exponential tails) clearly violates the second condition, we conclude that states $\rho$ that are bipartite separable up to finite depth local channels with strong on-site $O(2)$ or $SO(3)$ symmetry must break the weak translation symmetry, either explicitly or spontaneously.\footnote{Strictly speaking, the proof in Ref.~\cite{KimchiNahumSenthil} was done in the $L\to\infty$ limit, assuming a ``smooth'' thermodynamic limit. However, it is not difficult to modify the proof for finite $L$. The more precise conclusion is that the weak translation symmetry would be spontaneously broken if both the following conditions are satisfied: (1) for $|x-y|\ll L$, $\langle Z_xZ_y\rangle$ decays faster than $1/|x-y|^2$; (2) for $|x-y|\sim O(L)$, $L^2\langle Z_xZ_y\rangle\to 0$. For example, the state $\rho\propto P_{Q=0}$ with only strong $SO(2)$ symmetry has $\langle Z_xZ_y\rangle\sim 1/L$, violating the second condition, therefore is compatible with translation symmetry not being spontaneously broken. }

We note that the LSM theorem in mixed state context was also studied recently in Ref.~\cite{Zhouetal2023}, from a very different perspective.

Finally, we notice that if we lower the onsite symmetry to a discrete subgroup, for example $\Z_2\times \Z_2$, then our nonseparability result no longer applies, and the situation becomes very similar to the general discussions in Sec.~\ref{sec:mixed}. Indeed, we can have a symmetric density matrix $\rho\propto (\one+\prod_iX_i)(\one+\prod_iZ_i)$ that is $(L/2)$-separable under disjoint Bell pairs -- it is just \emph{symmetrically non-invertible} and has a \emph{diverging Markov length}.

\section{Relationship with other works}
\label{sec:relation}

We have focused on consequences of 't Hooft anomaly for mixed quantum states, in terms of constraints on long-range entanglement (for strong symmetry anomaly) and correlation (for mixed strong-weak anomaly). We now place our results in the broader context of research on mixed-state quantum phases, an area that has recently attracted significant attention. 

The notion of 't Hooft anomaly is intimately linked to symmetry-protected topological (SPT) phases. Recently SPT phases has been generalized to mixed-states\cite{deGroot2022,ma_average_2023,LeeYouXu2022,ZhangQiBi2022,Ma_intrisic_2023}. While prior works mainly focused on the bulk physics (which only have short-ranged entanglement and correlation), our work is concerned with the boundary physics. We have shown that a nontrivial bulk indeed leads to a nontrivial boundary in terms of long-range correlation and/or entanglement.

Shortly after the first version of our paper appeared on arXiv, two independent works on anomaly constraints for mixed states were posted\cite{WangLi2024,XuJian2024}. While a common theme of all these works (including ours) is that mixed states with nontrivial strong-weak t'Hooft anomaly must be nontrivial in some sense (in agreement with previous classification of bulk SPT phases), our work focused on different aspects of mixed-state phases in comparison to others. For strong symmetry anomaly, we focused on multi-partite entanglement, which allowed us to identify a class of \textit{intrinsically mixed} quantum phases that are not two-way connected to any pure states (Sec.~\ref{sec:intrinsicallymixed}). For mixed strong-weak anomaly, we first discussed its impact on \textit{symmetric invertibility}, and in a subsequent update discussed its impact on more information-theoretic aspects of long-range correlations.

For systems with only weak symmetry, there is no nontrivial anomaly in our definition. The simplest way to see this is to note that the maximally mixed state $\rho\sim \one$, which certainly has no entanglement nor correlation, is always invariant under any weak symmetry. However, a different notion of ``nontriviality'' can be formulated in terms of the degeneracies in the eigenvalue spectrum of $\rho$. Refs.~\cite{Hsin2023,Zang2023,Zhouetal2023} studied the notion of anomaly and LSM constraints from this point of view, which is different from the view taken in our work.

\section{Discussion}
\label{sec:Discussion}

In this work we explored several consequences of 't Hooft anomaly for mixed quantum states. Most notably, we showed that an anomalous strong symmetry requires the mixed state to be multipartite entangled, or nonseparable. Specifically, for a $d$-dimensional bosonic lattice system with symmetry $G$ and anomaly $\omega\in H^{d+2}(G,U(1))$, any strongly symmetric mixed state must be $(d+2)$-partite nonseparable. Furthermore, the nonseparability is long-ranged, in the sense that the symmetric state cannot be obtained from a $(d+2)$-partite separable state using a finite-depth local quantum channel. We proved this result for $d\leq1$ and provided some plausibility argument for $d>1$. { We note that, in a different context, separability properties of bulk mixed-state SPT phases have been studied in Ref.~\cite{ChenGrover2023}. It remains to understand the precise relation between the bulk and ``boundary'' perspectives.}

From a quantum informational perspective, anomalies can be useful to generate examples of states with interesting patterns of mixed state entanglement. For example, strongly symmetric infinite-temperature states $\rho_\infty$ can be $(d+2)$-nonseparable and $(d+1)$-separable at the same time. Even in $d=1$, we discovered examples (such as $\rho_\infty$ of the CZX model in Eq.\eqref{eq:infinite-temperature_CZX}) that are are long-ranged entangled in the sense that they cannot be prepared from trivial states using finite-depth local channels; and, at the same time, any proper subsystems are maximally disordered, so that all connected correlation functions behave trivially.

We also discussed mixed anomaly involving both strong and weak symmetries. In general the mixed strong-weak anomaly no longer constrains quantum entanglement. Instead such anomaly constrains the long-range correlations and forbids the states from being symmetrically invertible or gapped Markovian. However, for the special -- but physically relevant -- case of spin-$1/2$ chain with weak translation and strong $SO(3)$ (or $O(2)$) spin rotation symmetries, we again concluded that any symmetric mixed state must be long-range entangled. It remains to see to what extent we can generalize such statement to other systems.

We end with some open questions:
\begin{enumerate}
    \item In this paper we have focused on bosonic anomalies of unitary zero-form symmetries. A natural question is whether our results can be extended to include fermions and/or higher-form symmetries. Even for bosonic zero-form symmetries, our result is rigorously established only for $d\leq1$, and a controlled treatment of higher dimensional cases is left to future work.
    \item In this work we focused on entanglement properties of lattice systems. It is natural to ask whether we can formulate similar results for continuum field theories, where the notion of 't Hooft anomaly can also be naturally defined.
    \item The simple example states discussed in Sec.~\ref{sec:prototypical_examples} and \ref{sec:mixed} are infinite-temperature states that exhibits strong-to-weak spontaneous symmetry breaking~\cite{LeeJianXu2023,Ma_intrisic_2023,SSBtoappear}. For pure states with anomalous symmetry, it is known that if spontaneous symmetry breaking is forbidden, we can have stronger constraints on the states -- typically the state will exhibit either power-law correlation function or intrinsic topological order. Similarly, it is natural to ask: if there is no strong-to-weak spontaneous symmetry breaking, can we make even stronger statements about the mixed quantum states?

    \item For ground states, anomalous symmetries not only impose constraints on pure states, leading to non-trivial features like long-range entanglement, but also restrict their parent Hamiltonians, which cannot be symmetrically gapped without topological order. Can an analogous statement be established regarding the dynamics of anomalous open systems? In such context, the Lindbladian naturally replaces the role of the Hamiltonian as a generator of dynamics. One interesting direction is to investigate the conditions under which the steady-states of Lindbladian evolutions are gapped Markovian or not (Def.~\ref{def:gapped}), and how this relates to properties of the Lindbladian.
    \item In Sec.~\ref{sec:LSM} we have seen that whether the strong on-site symmetry is continuous or discrete makes a significant difference -- the former makes the states nonseparable while the latter does not, even though the 't Hooft anomaly is nontrivial in both cases. Are non-Abelian continuous symmetries intrinsically special in some sense? Or is this just a special feature of the 1d spin-$1/2$ chain?
    \item Last, but perhaps the most important: can we find some {physically observable} consequences of the multipartite nonseparability?
\end{enumerate}

\begin{acknowledgments}
We thank Tyler Ellison, Tarun Grover, Timothy Hsieh, Chaoming Jian, Anton Kapustin, Tsung-Cheng Lu, Ruochen Ma, Shengqi Sang, Amin Moharramipour, Amir-Reza Negari,  Jianhao Zhang and Yijian Zou for illuminating discussions. LAL acknowledges  supports  from  the  Natural Sciences and Engineering Research Council of Canada (NSERC)  through  Discovery  Grants. Research at Perimeter Institute (LAL and CW) is supported in part by the Government of Canada through the Department of Innovation, Science and Industry Canada and by the Province of Ontario through the Ministry of Colleges and Universities. MC acknowledges supports from NSF under award number DMR-1846109.
\end{acknowledgments}

\bibliography{bibliography}

\newpage
\appendix

\section{A ``good'' measure of multipartite entanglement}\label{appendix:measure_multip_entanglement}

In Sec. \ref{sec:intro-partial_separability}, we described how hard it is to decide whether a given multipartite state is separable or not. Even harder is to calculate a good measure of multipartite entanglement, as there is a tradeoff between practical computability and having desirable properties, such as being zero if and only if the state is ($k$-)separable.

Based on the general framework for classifying multipartite entanglement and constructing their measures given in \cite{szalayMultipartiteEntanglementMeasures2015}, we will present here one good but hard-to-compute mixed-state measure of entanglement with respect to a $k$-partition. First, given a pure state $\ket{\psi} \in \Hilb_1 \otimes \cdots \otimes \Hilb_N$ and a $k$-partition $\mathcal{P} = \{A_i\}_{i=1}^k$, we define the \emph{$k$-partite entanglement} $E_{\mathcal{P}}$ to be (half of) the sum of entanglement entropies for each set in the partition:
\begin{equation}
    E_{\mathcal{P}}(\ketbra{\psi}{\psi}) \defeq \frac{1}{2}[S(\rho_{A_1}) + S(\rho_{A_2}) + \cdots + S(\rho_{A_k})].
\end{equation}
Then, we extend $E_{\mathcal{P}}$ to all mixed states by taking its convex roof:
\begin{equation}
    E_{\mathcal{P}}(\rho) = \min_{\{p_i, \ket{\psi_i}\}} \sum_i p_i E_{\mathcal{P}}(\ketbra{\psi_i}{\psi_i}).
\end{equation}
Defined in this way, $E_{\mathcal{P}}$ enjoy the following properties:
\begin{itemize}
    \item Monotonic under LOCC,
    \item Convex under mixtures,
    \item Zero if, and only if, the state is separable with respect to $\mathcal{P}$, and
    \item Reduces to entanglement of formation for $k=2$.
\end{itemize}

 Hence, we can reformulate our main result by stating that a strongly symmetric anomalous state $\rho$ has $E_{\mathcal{P}}(\rho) > 0$ for certain $(d+2)$-partitions $\mathcal{P}$.

 Moreover, since the set of all strongly symmetric states is compact, then there exists an anomalous state that achieves a lower bound of multipartite entanglement $E^{\text{min}}_{U, d, L} > 0$ that only depends on the group action $U(g)$, the dimension $d$ and the number of sites $L$. As an example for the $d=1$ CZX symmetry, the bipartite separable states (See Sec. \ref{sec:czx-bipartite-separability}), if arranged on the $A|BC$ bipartition, have $\log 2$ tripartite entanglement. Surprisingly, there are symmetric states with lower tripartite entanglement: by randomly sampling symmetric pure states and computing their tripartite entanglement $E_{A|B|C}(\rho) \equiv E_3(\rho)$, we conjecture that for the CZX symmetry, $E_3^{\text{min}} \approx 0.9 \log 2$ independently of the system size. The candidate states with this value of tripartite entanglement constitute a complete basis of the form
 \begin{equation}\label{eq:E3_min_states}
     \ket{\alpha \beta \gamma^{(s)}} + \lambda_A \ket{\comp{\alpha} \beta \gamma^{(r_A)}} + \lambda_B \ket{\alpha \comp{\beta} \gamma^{(r_B)}} + \lambda_C \ket{\alpha \beta \comp{\gamma}^{(r_C)}},
 \end{equation}
where $\ket{\xi^{(r)}} \equiv \frac{1}{\sqrt{2}} (\ket{\xi} + r \ket{\comp{\xi}})$, with $\xi$ a bitstring and $r \in U(1)$,  $\alpha, \beta, \gamma \in \{0,1\}^{|A|, |B|, |C|}$ bitstrings in $A, B$ and $C$ respectively, and $\lambda_i, r_j \in U(1)$ phase factors. Requiring symmetry with respect to $U_{CZX}$ with eigenvalue $\mu \in \{ \pm 1 \}$ constrains
\begin{equation}
    \mu = s (-1)^{\sum_{i} (\alpha \beta \gamma)_i (\alpha \beta \gamma)_{i+1}},
\end{equation}
and
\begin{align}
    s & = (-1)^{\alpha_1 + \alpha_{|A|} + \beta_1 + \gamma_{|C|} + |A| + 1} r_A \nonumber \\
    & = (-1)^{\beta_1 + \beta_{|B|} + \gamma_1 + \alpha_{|A|} + |B| + 1} r_B \nonumber \\
    & = (-1)^{\gamma_1 + \gamma_{|C|} + \alpha_1 + \beta_{|B|} + |C| + 1} r_C,
\end{align}
and minimizing their tripartite entanglement $E_3$ implies
\begin{gather}
    s = r_A \lambda_A^2 = r_B \lambda_B^2 = r_C \lambda_C^2, \\
    \lambda_A^2 \lambda_B^2 \lambda_C^2 = -1, \\
    \text{with } \lambda_i^4 = 1, r_j^2 = 1.
\end{gather}
By varying all the parameters respecting the constraints above, one can verify that the states of Eq. \eqref{eq:E3_min_states} form a orthonormal basis, each having tripartite entanglement $E_3^{\text{min}} = \frac{3}{2} (- p \log p - q \log q) \approx 0.9 \log 2$, for $p = 1-q = \frac{2 + \sqrt{2}}{4}$. Furthermore, because they form a basis of symmetric states, then the tripartite entanglement $E_3(\rho_{\infty, \pm})$ of the CZX infinite temperature states $\rho_{\infty, \pm}$ of Eq. \eqref{eq:infinite-temperature_CZX} is upper bounded by $E_3^{\text{min}}$ above, being equal if our conjecture that they minimize $E_3$ is indeed correct.

\section{Tripartite entanglement of infinite temperature states with CZX symmetry via permutation criterion}\label{appendix:permutation_creterion}

Consider the 4-qubit infinite temperature state $\rho_{\infty} = \frac{1}{2^4} (\one + U_{\text{CZX}})$. For each $Z$-basis state $\ket{i j}$ of $\Hilb_A$, $\rho_\infty$ is block diagonal in the subspace $V_{ij}$ spanned by acting with $X_A$, $X_B$, $X_C$ or products thereof on $\ket{ij}_A\ket{0}_B\ket{0}_C$. For example, the block matrix for the $V_{00} = V_{11}$ subspace is
\begin{equation}
\frac{1}{2^4}
\begin{pmatrix}
    1 & 0 & 0 & 0 & 0 & 0 & 0 & 1 \\
    0 & 1 & 0 & 0 & 0 & 0 & 1 & 0 \\
    0 & 0 & 1 & 0 & 0 & 1 & 0 & 0 \\
    0 & 0 & 0 & 1 & -1 & 0 & 0 & 0 \\
    0 & 0 & 0 & -1 & 1 & 0 & 0 & 0 \\
    0 & 0 & 1 & 0 & 0 & 1 & 0 & 0 \\
    0 & 1 & 0 & 0 & 0 & 0 & 1 & 0 \\
    1 & 0 & 0 & 0 & 0 & 0 & 0 & 1 \\
\end{pmatrix}.
\end{equation}
Importantly, the permutation $\rho_{i_A, i_B, i_C}^{j_A, j_B, j_C} \rightarrow \rho_{i_A, j_C, i_C}^{j_A, j_B, i_B}$ keeps this block diagonal structure, so we can apply it to each subspace $V_{ij}$ separately. For $V_{00}$, it becomes
\begin{equation}
\frac{1}{2^4}
\begin{pmatrix}
    1 & 0 & \rn{13}{0} & \rn{14}{1} & 0 & 0 & \rn{17}{0} & \rn{18}{0} \\
    \rn{21}{0} & \rn{22}{0} & 0 & 0 & \rn{25}{0} & \rn{26}{1} & 1 & 0 \\
    0 & 0 & \rn{33}{0} & \rn{34}{0} & 0 & 1 & \rn{37}{-1} & \rn{38}{0} \\
    \rn{41}1 & \rn{42}{0} & 0 & 1 & \rn{45}{0} & \rn{46}{0} & 0 & 0 \\
    0 & 0 & \rn{53}{0} & \rn{54}{0} & 1 & 0 & \rn{57}{0} & \rn{58}{1} \\
    \rn{61}{0} & \rn{62}{-1} & 1 & 0 & \rn{65}{0} &\rn{66}{0} & 0 & 0 \\
    0 & 1 & \rn{73}{1} & \rn{74}{0} & 0 & 0 & \rn{77}{0} & \rn{78}{0} \\
    \rn{81}{0} & \rn{82}{0} & 0 & 0 & \rn{85}{1} & \rn{86}{0} & 0 & 1 \\
\end{pmatrix},
\end{equation}

\begin{tikzpicture}
    [overlay,remember picture,
    arrow/.style={<->, draw=black!20}]
    \draw [arrow] (21) -- (13);
    \draw [arrow] (22) -- (14);
    \draw [arrow] (41) -- (33);
    \draw [arrow] (42) -- (34);
    \draw [arrow] (25) -- (17);
    \draw [arrow] (26) -- (18);
    \draw [arrow] (45) -- (37);
    \draw [arrow] (46) -- (38);
    \draw [arrow] (61) -- (53);
    \draw [arrow] (62) -- (54);
    \draw [arrow] (81) -- (73);
    \draw [arrow] (82) -- (74);
    \draw [arrow] (65) -- (57);
    \draw [arrow] (66) -- (58);
    \draw [arrow] (85) -- (77);
    \draw [arrow] (86) -- (78);
\end{tikzpicture}
whose trace norm is $\frac{1}{2}(\frac{1}{2}+\frac{1}{\sqrt{2}})$. The same value is found for $V_{01} = V_{10}$, and since the trace norm of the direct sum of matrices is the sum of the trace norms, we find the trace norm of the state after permutation is $\frac{1}{2} + \frac{1}{\sqrt{2}} > 1$. It is also straightforward to generalize this result to the case where $A$ is larger, but $B$ and $C$ are still one-qubit long and adjacent to each other. We expect a similar argument to hold for general $A$, $B$ and $C$ that are connected.

\section{Local adaptive preparation of infinite temperature states with CZX symmetry}\label{appendix:adaptive_CZX}

Here, we will show a finite-depth local adaptive preparation of the infinite temperature states $\rho_{\infty, \pm} \propto \one \pm U_{\text{CZX}}$, as discussed in Sec. \ref{sec:infinite_T_CZX}. By ``local adaptive preparation'' we mean starting from a product state and applying a quantum circuit involving (geometrically) local unitary gates and local measurements, whose outcomes may influence future parts of the circuit. A familiar example of an adaptive circuit is the syndrome measurement of a stabilizer code followed by error correction conditioned on the syndrome.

We divide the adaptive preparation in three parts:
\begin{enumerate}
    \item \textbf{Initial state preparation}: Start from the $\ket{++\cdots +}$ state in a system of $L$ qubits with periodic boundary condition.
    \item \textbf{$ZZ$ Measurement}: Measure $Z_i Z_{i+1}$ for all $i \in \{1, 2, \ldots, L\}$, whose outcome $c_i = \pm 1$ we record for future classical processing. Having knowledge of all results $(c_i)_{i=1}^L$, we can be sure that the post-measurement state is a GHZ-like state in the Z-basis of the form $\frac{1}{\sqrt{2}}(\ket{b} + \ket{\comp{b}})$, where $b \in \{0, 1\}^{L}$ is a bistring and $\comp{b}_i \defeq 1 - b_i$ is its complement. The relation between $b$ and $(c_i)_{i=1}^L$ is given by $(-1)^{b_i} (-1)^{b_{i+1}} = (-1)^{\comp{b}_i} (-1)^{\comp{b}_{i+1}} = c_i$. In other words, from $(c_i)_{i=1}^L$ we can find the two bitstrings $\{b, \comp{b}\}$ satisfying the equation above.
    \item \textbf{Unitary correction with feedback}: The post-measurement GHZ-like state is also an eigenstate of the CXZ symmetry, with eigenvalue $\lambda = (-1)^{\sum_i b_i b_{i+1}} = (-1)^{\sum_i \comp{b}_i \comp{b}_{i+1}}$. Since we can compute $\lambda$ based \emph{only} on the $ZZ$ measurement results $(c_i)_{i=1}^L$, we can conditionally apply $Z_1$ if $\lambda = -1$ to make it an $+1$ eigenstate of the CZX symmetry, or vice-versa.
\end{enumerate}

The resulting density matrix coming from the mixture of the final states over all measurement results is exactly equal to $\rho_{\infty, \pm} \propto \one \pm U_\text{CZX}$. We note that even though the nonlocal nature of the feedback is crucial to establish the necessary global correlations, the quantum part of the circuit is entirely local and finite-depth.

\section{Long-range entanglement proof}\label{appendix:LRE_proof}

\begin{figure}[t]
    \centering
    \def\svgwidth{\linewidth}
    \begingroup%
  \makeatletter%
  \providecommand\color[2][]{%
    \errmessage{(Inkscape) Color is used for the text in Inkscape, but the package 'color.sty' is not loaded}%
    \renewcommand\color[2][]{}%
  }%
  \providecommand\transparent[1]{%
    \errmessage{(Inkscape) Transparency is used (non-zero) for the text in Inkscape, but the package 'transparent.sty' is not loaded}%
    \renewcommand\transparent[1]{}%
  }%
  \providecommand\rotatebox[2]{#2}%
  \newcommand*\fsize{\dimexpr\f@size pt\relax}%
  \newcommand*\lineheight[1]{\fontsize{\fsize}{#1\fsize}\selectfont}%
  \ifx\svgwidth\undefined%
    \setlength{\unitlength}{218.94512747bp}%
    \ifx\svgscale\undefined%
      \relax%
    \else%
      \setlength{\unitlength}{\unitlength * \real{\svgscale}}%
    \fi%
  \else%
    \setlength{\unitlength}{\svgwidth}%
  \fi%
  \global\let\svgwidth\undefined%
  \global\let\svgscale\undefined%
  \makeatother%
  \begin{picture}(1,0.48175262)%
    \lineheight{1}%
    \setlength\tabcolsep{0pt}%
    \put(0,0){\includegraphics[width=\unitlength,page=1]{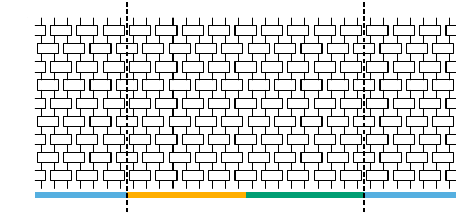}}%
    \put(0.17743778,0.00313105){\color[rgb]{0.33333333,0.68627451,0.87843137}\makebox(0,0)[t]{\lineheight{1.25}\smash{\begin{tabular}[t]{c}$A$\end{tabular}}}}%
    \put(0,0){\includegraphics[width=\unitlength,page=2]{local_unitary_surgery_appendix.pdf}}%
    \put(0.40840002,0.00313105){\color[rgb]{1,0.68235294,0}\makebox(0,0)[t]{\lineheight{1.25}\smash{\begin{tabular}[t]{c}$B$\end{tabular}}}}%
    \put(0.66828391,0.00313105){\color[rgb]{0,0.61960784,0.45098039}\makebox(0,0)[t]{\lineheight{1.25}\smash{\begin{tabular}[t]{c}$C$\end{tabular}}}}%
    \put(0.8990652,0.00313105){\color[rgb]{0.33333333,0.68627451,0.87843137}\makebox(0,0)[t]{\lineheight{1.25}\smash{\begin{tabular}[t]{c}$A$\end{tabular}}}}%
    \put(0.53829683,0.45415791){\color[rgb]{0,0,0}\makebox(0,0)[t]{\lineheight{1.25}\smash{\begin{tabular}[t]{c}$V_{BC}$\end{tabular}}}}%
    \put(0.06130651,0.36295284){\color[rgb]{0.99607843,0,0}\makebox(0,0)[rt]{\lineheight{1.25}\smash{\begin{tabular}[t]{r}$V$\end{tabular}}}}%
    \put(0.06130651,0.12506193){\color[rgb]{0.99607843,0,0}\makebox(0,0)[rt]{\lineheight{1.25}\smash{\begin{tabular}[t]{r}$V^\dagger$\end{tabular}}}}%
    \put(0.06130651,0.24400738){\color[rgb]{0,0.68627451,0.71764706}\makebox(0,0)[rt]{\lineheight{1.25}\smash{\begin{tabular}[t]{r}$U$\end{tabular}}}}%
  \end{picture}%
\endgroup%

    \caption{The original symmetry $U(g)$ conjugated by the (extended) FDLU $V$. The circuit $V_{BC}$ acts as $V$ in the bulk of $BC$ and is entirely supported on it. By removing $V_{BC}$ from $U'(g) \defeq V U(g) V^\dagger$, we ``expose'' $U$ in the bulk and thus arrive at another unitary restriction to $BC$ with the same nontrivial cohomology class $[\tilde\omega] = [\omega']$.}
    \label{fig:local_unitary_surgery_appendix}
\end{figure}

Here we will prove that the anomaly-nonseparability connection remains if a local unitary symmetry representation $U(g)$ is conjugated by another local unitary $V$. In particular, we will prove that the group cocycle $\omega'$ coming from the $U'(g) = V U(g) V^\dagger$ representation is in the same group cohomology class as $\omega$, from $U(g)$.

First, we follow the same arguments of Sec.\ref{sec:tripartite_entanglement} to reach Eq.\eqref{eq:associativity_reduced_boundary_ops} with $U_{BC}$, $\Omega$ and $\omega$ replaced by $U'_{BC}$, $\Omega'$ and $\omega'$, respectively:
\begin{equation}\label{eq:associativity_reduced_boundary_ops_LRE}
\begin{split}
    &\omega'(g_{1},g_{2},g_{3}) \Omega'_{c}(g_{1},g_{2})\Omega'_{c}(g_{1}g_{2},g_{3}) \\ & = U'_{BC}(g_{1})\Omega'_{c}(g_{2},g_{3}) U'_{BC}(g_{1})^{-1} \Omega'_{c}(g_{1},g_{2}g_{3}).
\end{split}
\end{equation}

Then, we define $V_{BC}$ to be a finite-depth local unitary supported in $BC$ and acting the same as $V$ in the bulk of $BC$. One can construct such unitary by removing the gates in the ``causal past'' of the boundary of $BC$ as in Fig.~\ref{fig:local_unitary_surgery_appendix}. By conjugating Eq.\eqref{eq:associativity_reduced_boundary_ops_LRE} with $V_{BC}$, we find that $\tilde{U}_{BC}(g) = V^\dagger_{BC} U_{BC}'(g) V_{BC}$ and $\tilde{\Omega}(g_1, g_2) \defeq V^\dagger_{BC} \Omega'(g_1, g_2) V_{BC}$ satisfy the same relation \emph{with the same cocycle} $\tilde{\omega} = \omega'$. Furthermore, $\tilde{U}_{BC}$ is still supported on $BC$ and acts as $U(g)$ in the bulk of $BC$. In other words, $\tilde{U}_{BC}$ is but another unitary restriction of the original symmetry $U$ to $BC$. Since it can be shown that two different symmetry restrictions $U_{BC}$ and $\tilde{U}_{BC}$ give rise to 3-cocycles $\omega$ and $\tilde\omega$ that differ by an exact 3-cocycle~\cite{else_classifying_2014}, then $\omega' = \tilde{\omega}$ is in the same group cohomology class as $\omega$. QED.

\section{The proof that \texorpdfstring{$u_{AB}$}{uAB} does not entangle A and B}
\label{appendix:uABproof}
Following the discussion in Sec.~\ref{sec:tripartite_entanglement}, we assume $u_{AB}$ is supported in a neighborhood near the AB boundary, etc.

\begin{lemma}
    If $u_{AB}u_{BC}u_{AC}\ket{A}\ket{B}\ket{C}=\ket{\tilde{A}}\ket{\tilde{B}}\ket{\tilde{C}}$ as in the diagram below, then $u_{AB}$ does not entangle $\ket{A}$ and $\ket{B}$, and similarly for $u_{BC}$ and $u_{AC}$:
    \begin{equation}
        \begin{tikzpicture}
	\begin{pgfonlayer}{nodelayer}
		\node [style=none] (0) at (1.75, 1.25) {};
		\node [style=none] (1) at (3.75, 1.25) {};
		\node [style=none] (2) at (2.75, 0) {};
		\node [style=none] (3) at (2.75, 0.75) {$A$};
		\node [style=none] (4) at (4.25, 1.25) {};
		\node [style=none] (5) at (6.25, 1.25) {};
		\node [style=none] (6) at (5.25, 0) {};
		\node [style=none] (7) at (5.25, 0.75) {$B$};
		\node [style=none] (8) at (2.25, 1.25) {};
		\node [style=none] (9) at (3.25, 1.25) {};
		\node [style=none] (10) at (2.75, 1.75) {};
		\node [style=none] (11) at (2.75, 2.75) {};
		\node [style=none] (14) at (3.25, 1.75) {};
		\node [style=none] (16) at (3.25, 2.75) {};
		\node [style=none] (19) at (3.25, 3.25) {};
		\node [style=none] (20) at (2.25, 3.25) {};
		\node [style=none] (21) at (5.25, 1.75) {};
		\node [style=none] (22) at (5.25, 2.75) {};
		\node [style=none] (23) at (4.75, 1.25) {};
		\node [style=none] (24) at (4.75, 1.75) {};
		\node [style=none] (25) at (4.75, 2.75) {};
		\node [style=none] (26) at (4.75, 3.25) {};
		\node [style=none] (27) at (5.75, 3.25) {};
		\node [style=none] (28) at (5.75, 1.25) {};
		\node [style=none] (29) at (4, 2.25) {};
		\node [style=none] (30) at (4, 2.25) {};
		\node [style=none] (31) at (4, 2.25) {$u_{AB}$};
		\node [style=none] (32) at (7.25, 1.25) {};
		\node [style=none] (33) at (9.25, 1.25) {};
		\node [style=none] (34) at (8.25, 0) {};
		\node [style=none] (35) at (8.25, 0.75) {$A'$};
		\node [style=none] (36) at (9.75, 1.25) {};
		\node [style=none] (37) at (11.75, 1.25) {};
		\node [style=none] (38) at (10.75, 0) {};
		\node [style=none] (39) at (10.75, 0.75) {$B'$};
		\node [style=none] (40) at (7.75, 1.25) {};
		\node [style=none] (41) at (8.75, 1.25) {};
		\node [style=none] (46) at (8.75, 3.25) {};
		\node [style=none] (47) at (7.75, 3.25) {};
		\node [style=none] (50) at (10.25, 1.25) {};
		\node [style=none] (53) at (10.25, 3.25) {};
		\node [style=none] (54) at (11.25, 3.25) {};
		\node [style=none] (55) at (11.25, 1.25) {};
		\node [style=none] (88) at (6.75, 1.75) {};
		\node [style=none] (90) at (6.75, 1.75) {$=$};
		\node [style=none] (93) at (0.75, 6) {};
		\node [style=none] (94) at (2.75, 6) {};
		\node [style=none] (95) at (1.75, 4.75) {};
		\node [style=none] (96) at (1.75, 5.5) {$A$};
		\node [style=none] (97) at (3.25, 6) {};
		\node [style=none] (98) at (5.25, 6) {};
		\node [style=none] (99) at (4.25, 4.75) {};
		\node [style=none] (100) at (4.25, 5.5) {$B$};
		\node [style=none] (101) at (1.25, 6) {};
		\node [style=none] (102) at (2.25, 6) {};
		\node [style=none] (103) at (2, 6.5) {};
		\node [style=none] (104) at (2, 7.5) {};
		\node [style=none] (105) at (2.25, 6.5) {};
		\node [style=none] (106) at (2.25, 7.5) {};
		\node [style=none] (107) at (2.25, 8) {};
		\node [style=none] (109) at (4, 6.5) {};
		\node [style=none] (110) at (4, 7.5) {};
		\node [style=none] (111) at (3.75, 6) {};
		\node [style=none] (112) at (3.75, 6.5) {};
		\node [style=none] (113) at (3.75, 7.5) {};
		\node [style=none] (114) at (3.75, 8) {};
		\node [style=none] (116) at (4.75, 6) {};
		\node [style=none] (117) at (3, 7) {};
		\node [style=none] (118) at (3, 7) {};
		\node [style=none] (119) at (3, 7) {$u_{AB}$};
		\node [style=none] (120) at (5.75, 6) {};
		\node [style=none] (121) at (7.75, 6) {};
		\node [style=none] (122) at (6.75, 4.75) {};
		\node [style=none] (123) at (6.75, 5.5) {$C$};
		\node [style=none] (125) at (6.25, 6) {};
		\node [style=none] (127) at (7.25, 6) {};
		\node [style=none] (128) at (4.5, 6.5) {};
		\node [style=none] (129) at (4.5, 7.5) {};
		\node [style=none] (130) at (4.75, 6.5) {};
		\node [style=none] (131) at (4.75, 7.5) {};
		\node [style=none] (132) at (4.75, 8) {};
		\node [style=none] (133) at (6.5, 6.5) {};
		\node [style=none] (134) at (6.5, 7.5) {};
		\node [style=none] (135) at (6.25, 6.5) {};
		\node [style=none] (136) at (6.25, 7.5) {};
		\node [style=none] (137) at (6.25, 8) {};
		\node [style=none] (138) at (5.5, 7) {};
		\node [style=none] (139) at (5.5, 7) {};
		\node [style=none] (140) at (5.5, 7) {$u_{BC}$};
		\node [style=none] (141) at (7, 6.5) {};
		\node [style=none] (142) at (7, 7.5) {};
		\node [style=none] (143) at (7.25, 6.5) {};
		\node [style=none] (144) at (7.25, 7.5) {};
		\node [style=none] (145) at (7.25, 8) {};
		\node [style=none] (148) at (8.5, 6.5) {};
		\node [style=none] (149) at (8.5, 7.5) {};
		\node [style=none] (153) at (8, 7) {$u_{AC}$};
		\node [style=none] (154) at (1.5, 6.5) {};
		\node [style=none] (155) at (1.5, 7.5) {};
		\node [style=none] (156) at (1.25, 6.5) {};
		\node [style=none] (157) at (1.25, 7.5) {};
		\node [style=none] (158) at (1.25, 8) {};
		\node [style=none] (159) at (0, 7.5) {};
		\node [style=none] (160) at (0, 6.5) {};
		\node [style=none] (161) at (0.5, 7) {$u_{AC}$};
		\node [style=none] (162) at (9.75, 6) {};
		\node [style=none] (163) at (11.75, 6) {};
		\node [style=none] (164) at (10.75, 4.75) {};
		\node [style=none] (165) at (10.75, 5.5) {$\tilde{A}$};
		\node [style=none] (166) at (12.25, 6) {};
		\node [style=none] (167) at (14.25, 6) {};
		\node [style=none] (168) at (13.25, 4.75) {};
		\node [style=none] (169) at (13.25, 5.5) {$\tilde{B}$};
		\node [style=none] (170) at (10.25, 6) {};
		\node [style=none] (171) at (11.25, 6) {};
		\node [style=none] (172) at (11.25, 8) {};
		\node [style=none] (173) at (10.25, 8) {};
		\node [style=none] (174) at (12.75, 6) {};
		\node [style=none] (175) at (12.75, 8) {};
		\node [style=none] (176) at (13.75, 8) {};
		\node [style=none] (177) at (13.75, 6) {};
		\node [style=none] (178) at (9.25, 6.5) {$=$};
		\node [style=none] (179) at (14.75, 6) {};
		\node [style=none] (180) at (16.75, 6) {};
		\node [style=none] (181) at (15.75, 4.75) {};
		\node [style=none] (182) at (15.75, 5.5) {$\tilde{C}$};
		\node [style=none] (183) at (15.25, 6) {};
		\node [style=none] (184) at (15.25, 8) {};
		\node [style=none] (185) at (16.25, 8) {};
		\node [style=none] (186) at (16.25, 6) {};
		\node [style=none] (187) at (1.25, 8.5) {$A_1$};
		\node [style=none] (188) at (2.25, 8.5) {$A_2$};
		\node [style=none] (189) at (3.75, 8.5) {$B_1$};
		\node [style=none] (190) at (4.75, 8.5) {$B_2$};
		\node [style=none] (191) at (6.25, 8.5) {$C_1$};
		\node [style=none] (192) at (7.25, 8.5) {$C_2$};
		\node [style=none] (193) at (12.75, 1.75) {$, \quad \ldots$};
		\node [style=none] (194) at (9.25, 4) {$\Downarrow$};
	\end{pgfonlayer}
	\begin{pgfonlayer}{edgelayer}
		\draw [style=new edge style 0] (2.center)
			 to (1.center)
			 to (0.center)
			 to cycle;
		\draw [style=new edge style 0] (6.center)
			 to (5.center)
			 to (4.center)
			 to cycle;
		\draw [style=new edge style 0] (21.center)
			 to (22.center)
			 to (11.center)
			 to (10.center)
			 to cycle;
		\draw [style=new edge style 0] (9.center) to (14.center);
		\draw [style=new edge style 0] (23.center) to (24.center);
		\draw [style=new edge style 0] (25.center) to (26.center);
		\draw [style=new edge style 0] (16.center) to (19.center);
		\draw [style=new edge style 0] (8.center) to (20.center);
		\draw [style=new edge style 0] (28.center) to (27.center);
		\draw [style=new edge style 0] (34.center)
			 to (33.center)
			 to (32.center)
			 to cycle;
		\draw [style=new edge style 0] (38.center) to (37.center);
		\draw [style=new edge style 0] (37.center) to (36.center);
		\draw [style=new edge style 0] (36.center) to (38.center);
		\draw [style=new edge style 0] (40.center) to (47.center);
		\draw [style=new edge style 0] (55.center) to (54.center);
		\draw [style=new edge style 0] (41.center) to (46.center);
		\draw [style=new edge style 0] (50.center) to (53.center);
		\draw [style=new edge style 0] (95.center)
			 to (94.center)
			 to (93.center)
			 to cycle;
		\draw [style=new edge style 0] (99.center)
			 to (98.center)
			 to (97.center)
			 to cycle;
		\draw [style=new edge style 0] (109.center)
			 to (110.center)
			 to (104.center)
			 to (103.center)
			 to cycle;
		\draw [style=new edge style 0] (102.center) to (105.center);
		\draw [style=new edge style 0] (111.center) to (112.center);
		\draw [style=new edge style 0] (113.center) to (114.center);
		\draw [style=new edge style 0] (106.center) to (107.center);
		\draw [style=new edge style 0] (122.center) to (121.center);
		\draw [style=new edge style 0] (121.center) to (120.center);
		\draw [style=new edge style 0] (120.center) to (122.center);
		\draw [style=new edge style 0] (133.center) to (134.center);
		\draw [style=new edge style 0] (134.center) to (129.center);
		\draw [style=new edge style 0] (129.center) to (128.center);
		\draw [style=new edge style 0] (128.center) to (133.center);
		\draw [style=new edge style 0] (136.center) to (137.center);
		\draw [style=new edge style 0] (131.center) to (132.center);
		\draw [style=new edge style 0] (142.center) to (141.center);
		\draw [style=new edge style 0] (144.center) to (145.center);
		\draw (149.center) to (142.center);
		\draw (148.center) to (141.center);
		\draw (127.center) to (143.center);
		\draw (125.center) to (135.center);
		\draw (116.center) to (130.center);
		\draw [style=new edge style 0] (154.center) to (155.center);
		\draw [style=new edge style 0] (157.center) to (158.center);
		\draw (155.center) to (159.center);
		\draw (154.center) to (160.center);
		\draw (101.center) to (156.center);
		\draw [style=new edge style 0] (164.center)
			 to (163.center)
			 to (162.center)
			 to cycle;
		\draw [style=new edge style 0] (168.center) to (167.center);
		\draw [style=new edge style 0] (167.center) to (166.center);
		\draw [style=new edge style 0] (166.center) to (168.center);
		\draw [style=new edge style 0] (170.center) to (173.center);
		\draw [style=new edge style 0] (177.center) to (176.center);
		\draw [style=new edge style 0] (171.center) to (172.center);
		\draw [style=new edge style 0] (174.center) to (175.center);
		\draw [style=new edge style 0] (181.center) to (180.center);
		\draw [style=new edge style 0] (180.center) to (179.center);
		\draw [style=new edge style 0] (179.center) to (181.center);
		\draw [style=new edge style 0] (186.center) to (185.center);
		\draw [style=new edge style 0] (183.center) to (184.center);
	\end{pgfonlayer}
\end{tikzpicture}

    \end{equation}
\end{lemma}

\begin{proof}
    We will prove the lemma by contradiction. Suppose $u_{AB}$ does entangle $A$ and $B$, i.e. $u_{AB} \ket{A}\ket{B} \neq \ket{A'}\ket{B'}$ for any $\ket{A'}$ and $\ket{B'}$.
    Consider dividing $A$ into $A_1$ and $A_2$, and $C$ into $C_1$ and $C_2$, so that $A_1$ and $C_2$ are adjacent to the AC boundary, and $u_{AC}$ is fully supported on $A_1\cup C_2$.

    Let us consider the von Neumann entropy $S$ of various regions. In the initial state $\ket{A}\ket{B}\ket{C}$, we have $S_{A}=0$, and $S_{A\cup C_2}=S_A+S_{C_2}=S_{C_2}$, the entanglement entropy, or EE, of $C_2$ in $\ket{C}$. Then, after applying $u_{AB}$, the EE of $A$ must increase, subsequently leading to an increase in the EE of $A\cup C_2$, as $A$ and $C$ remain separated at this point. Moreover, further applications of $u_{BC}$ and $u_{AC}$ do not change the EE of ${A\cup C_2}$. On the other hand, we also know that, in the final state, regions $A$, $B$ and $C$ become separated again, so
   \begin{equation}
       S_{C_2}^{\rm final}=S_{A\cup C_2}^{\rm final}>S_{A\cup C_2}=S_{C_2}.
   \end{equation}
    Importantly, the inequality above is strict since $u_{AB}$ is assumed to entangle $\ket{A}\ket{B}$.

    Similarly, by considering $A_1\cup C$ we find that $S_{A_1}^{\rm final}=S_{A_1\cup C}^{\rm final}$ does not decrease compared to $S_{A_1}$. In fact, if we assume $u_{BC}$ entangles $\ket{B}$ and $\ket{C}$, $S_{A_1}$ has to increase too.

    Because $A$ and $C$ are separated both in initial and the final state, we have
    \begin{equation}
        S_{A_1\cup C_2}^{\rm final}=S_{A_1}^{\rm final}+S_{C_2}^{\rm final}> S_{A_1}+S_{C_2}=S_{A_1\cup C_2}.
    \end{equation}
    However, this is impossible, since all the unitaries acting on the initial separated state either act entirely within $A_1\cup C_2$, or strictly outside it. Hence, they cannot change the entanglement entropy of $A_1\cup C_2$. We have thus reached a contradiction and our initial assumption of $u_{AB}$ entangling $\ket{A}$ and $\ket{B}$ must be false.
\end{proof}

\end{document}